\documentclass[12pt,english]{article}
\usepackage[T1]{fontenc}
\usepackage[latin9]{inputenc}
\usepackage{geometry}
\usepackage{verbatim}
\usepackage{amsmath}

\usepackage{amsthm}
\usepackage{amssymb}
\usepackage{graphicx}
\usepackage{setspace}
\usepackage{esint}
\usepackage[authoryear]{natbib}
\makeatletter

\providecommand{\tabularnewline}{\\}

\theoremstyle{remark}
\newtheorem{rem}{\protect\remarkname}
\theoremstyle{plain}
\newtheorem{lem}{\protect\lemmaname}
\theoremstyle{definition}
\newtheorem{defn}{\protect\definitionname}
\theoremstyle{plain}
\newtheorem{thm}{\protect\theoremname}
\theoremstyle{plain}
\newtheorem{prop}{\protect\propositionname}

\usepackage{amsfonts}\usepackage{babel}
\providecommand{\U}[1]{\protect\rule{.1in}{.1in}}
 \usepackage{fancyheadings}
\usepackage{setspace}

\makeatother

\usepackage{babel}
\providecommand{\definitionname}{Definition}
\providecommand{\lemmaname}{Lemma}
\providecommand{\propositionname}{Proposition}
\providecommand{\remarkname}{Remark}
\providecommand{\theoremname}{Theorem}

\begin{document}
\title{Berk-Nash Equilibrium: A Framework for Modeling Agents with Misspecified
Models\thanks{We thank Vladimir Asriyan, Pierpaolo Battigalli, Larry Blume, Aaron
Bodoh-Creed, Sylvain Chassang, Emilio Espino, Erik Eyster, Drew Fudenberg,
Yuriy Gorodnichenko, Stephan Lauermann, Natalia Lazzati, Kristóf Madarász,
Matthew Rabin, Ariel Rubinstein, Joel Sobel, Jörg Stoye, several seminar
participants, and especially a co-editor and four anonymous referees
for very helpful comments. Esponda: Olin Business School, Washington
University in St. Louis, 1 Brookings Drive, Campus Box 1133, St. Louis,
MO 63130, iesponda@wustl.edu; Pouzo: Department of Economics, UC Berkeley,
530-1 Evans Hall \#3880, Berkeley, CA 94720, dpouzo@econ.berkeley.edu.} \bigskip{}
}
\author{%
\begin{tabular}{cc}
Ignacio Esponda~~~~~~~~  & ~~~~~~~~Demian Pouzo\tabularnewline
(WUSTL)~~~~~~~~ & ~~~~~~~~(UC Berkeley)\tabularnewline
\end{tabular}}
\maketitle
\begin{abstract}
We develop an equilibrium framework that relaxes the standard assumption
that people have a correctly-specified view of their environment.
Each player is characterized by a (possibly misspecified) subjective
model, which describes the set of feasible beliefs over payoff-relevant
consequences as a function of actions. We introduce the notion of
a Berk-Nash equilibrium: Each player follows a strategy that is optimal
given her belief, and her belief is restricted to be the best fit
among the set of beliefs she considers possible. The notion of best
fit is formalized in terms of minimizing the Kullback-Leibler divergence,
which is endogenous and depends on the equilibrium strategy profile.
Standard solution concepts such as Nash equilibrium and self-confirming
equilibrium constitute special cases where players have correctly-specified
models. We provide a learning foundation for Berk-Nash equilibrium
by extending and combining results from the statistics literature
on misspecified learning and the economics literature on learning
in games.
\end{abstract}
\bigskip{}

\thispagestyle{empty}

\newpage{}

\setcounter{page}{1}

\section{Introduction}

Most economists recognize that the simplifying assumptions underlying
their models are often wrong. But, despite recognizing that models
are likely to be misspecified, the standard approach (with exceptions
noted below) assumes that economic agents have a correctly specified
view of their environment. We present an equilibrium framework that
relaxes this standard assumption and allows the modeler to postulate
that economic agents have a subjective and possibly incorrect view
of their world.

An objective game represents the true environment faced by the agent
(or players, in the case of several interacting agents). Payoff relevant
states and privately observed signals are drawn from an objective
probability distribution. Each player observes her own private signal
and then players simultaneously choose actions. The action profile
and the realized state determine consequences, and consequences determine
payoffs.

In addition, each player has a subjective model representing her own
view of the environment. Formally, a subjective model is a set of
probability distributions over own consequences as a function of a
player's own action and information. Crucially, we allow the subjective
model of one or more players to be misspecified, which roughly means
that the set of subjective distributions does not include the true,
objective distribution. For example, a consumer might perceive a nonlinear
price schedule to be linear and, therefore, respond to average, not
marginal, prices. Or traders might not realize that the value of trade
is partly determined by the terms of trade.

A \emph{Berk-Nash equilibrium} is a strategy profile such that, for
each player, there exists a belief with support in her subjective
model satisfying two conditions. First, the strategy is optimal given
the belief. Second, the belief puts probability one on the set of
subjective distributions over consequences that are ``closest''
to the true distribution, where the true distribution is determined
by the objective game and the actual strategy profile. The notion
of ``closest'' is given by a weighted version of the Kullback-Leibler
divergence, also known as relative entropy.

Berk-Nash equilibrium includes standard and boundedly rational solution
concepts in a common framework, such as Nash, self-confirming (e.g.,
\citet{battigalli1987compartamento}, \citet{fudenberg1993self},
\citet{dekel2004learning}), fully cursed (\citealp{eyster2005cursed}),
and analogy-based expectation equilibrium (\citet{jehiel2005analogy},
\citet{jehiel2008revisiting}). For example, suppose that the game
is \emph{correctly specified} (i.e., the support of each player's
prior contains the true distribution) and that the game is \emph{strongly
identified} (i.e., there is a unique distribution---whether or not
correct---that matches the observed data). Then Berk-Nash equilibrium
is equivalent to Nash equilibrium. If the strong identification assumption
is dropped, then Berk-Nash is a self-confirming equilibrium. In addition
to unifying previous work, our framework provides a systematic approach
for extending previous cases and exploring new types of misspecifications.

We provide a foundation for Berk-Nash equilibrium (and the use of
Kullback-Leibler divergence as a measure of ``distance'') by studying
a dynamic setup with a fixed number of players playing the objective
game repeatedly. Each player believes that the environment is stationary
and starts with a prior over her subjective model. In each period,
players use the observed consequences to update their beliefs according
to Bayes' rule. The main objective is to characterize limiting behavior
when players behave optimally but learn with a possibly misspecified
subjective model.\footnote{In the case of multiple agents, the environment need not be stationary,
and so we are ignoring repeated game considerations where players
take into account how their actions affect others' future play. We
discuss the extension to a population model with a continuum of agents
in Section \ref{subsec:Discussion}.}

The main result is that, if players' behavior converges, then it converges
to a Berk-Nash equilibrium. A converse result, showing that we can
converge to any Berk-Nash equilibrium of the game for some initial
(non-doctrinaire) prior, does not hold. But we obtain a positive convergence
result by relaxing the assumption that players exactly optimize. For
any given Berk-Nash equilibrium, we show that convergence to that
equilibrium occurs if agents are myopic and make \emph{asymptotically
optimal }choices (i.e., optimization mistakes vanish with time).

There is a longstanding interest in studying the behavior of agents
who hold misspecified views of the world. Examples come from diverse
fields including industrial organization, mechanism design, information
economics, macroeconomics, and psychology and economics (e.g., \citet{Arrow-Green},
\citet{kirman75learning}, \citet{sobel1984non}, \citet{kagel1986winner},
\citet{nyarko1991learning}, \citet{sargent-book}, \citet{rabin2002}),
although there is often no explicit reference to misspecified learning.
Most of the literature, however, focuses on particular settings, and
there has been little progress in developing a unified framework.
Our treatment unifies both ``rational'' and ``boundedly rational''
approaches, thus emphasizing that modeling the behavior of misspecified
players does not constitute a large departure from the standard framework.

\citet{Arrow-Green} provide a general treatment and make a distinction
between objective and subjective games. Their framework, though, is
more restrictive than ours in terms of the types of misspecifications
that players are allowed to have. Moreover, they do not establish
existence or provide a learning foundation for equilibrium. Recently,
\citet{spiegler2014bayesian} introduced a framework that uses Bayesian
networks to analyze decision making under imperfect understanding
of correlation structures.\footnote{Some explanations for why players may have misspecified models include
the use of heuristics (\citealp{KahnemanTversky1973}), complexity
(\citealp{aragones2005fact}), the desire to avoid over-fitting the
data (\citet{al-najjar2009decision}, \citet{al-najjar2013coarse}),
and costly attention (\citealp{schwartzstein2009selective}).}

Our paper is also related to the bandit (e.g., \citet{rothschild1974two},
\citet{mclennan1984price}, \citet{easley1988controlling}) and self-confirming
equilibrium (SCE) literatures, which highlight that agents might optimally
end up with incorrect beliefs if experimentation is costly.\footnote{In the macroeconomics literature, the term SCE is sometimes used in
a broader sense to include cases where agents have misspecified models
(e.g., \citealp{sargent-book}).} We also allow beliefs to be incorrect due to insufficient feedback,
but our main contribution is to allow for misspecified learning. When
players have misspecified models, beliefs may be incorrect and endogenously
depend on own actions even if there is persistent experimentation;
thus, an equilibrium framework is needed to characterize steady-state
behavior even in single-agent settings.\footnote{Two extensions of SCE are also potentially applicable: restrictions
on beliefs based on introspection (e.g., \citealp{rubinstein1994rationalizable}),
and ambiguity aversion (\citealp{battigalli2012selfconfirming}).}

From a technical perspective, we extend and combine results from two
literatures. First, the idea that equilibrium is a result of a learning
process comes from the literature on learning in games. This literature
studies explicit learning models to justify Nash and SCE (e.g., \citet{fudenberg1988theory},
\citet{fudenberg1993learning}, \citet{fudenberg1995learning}, \citet{fudenberg1993steady},
\citet{kalai1993rational}).\footnote{See Fudenberg and Levine (1998, 2009)\nocite{fudenberg1998theory}\nocite{fudenberg2009learning}
for a survey of this literature.} We extend this literature by allowing players to learn with models
of the world that are misspecified even in steady state. 

Second, we rely on and contribute to the literature studying the limiting
behavior of Bayesian posteriors. The results from this literature
have been applied to decision problems with correctly specified agents
(e.g., \citealp{easley1988controlling}). In particular, an application
of the martingale convergence theorem implies that beliefs converge
almost surely under the agent's subjective prior. This result, however,
does not guarantee convergence of beliefs according to the true distribution
if the agent has a misspecified model and the support of her prior
does not include the true distribution. Thus, we take a different
route and follow the statistics literature on misspecified learning.
This literature characterizes limiting beliefs in terms of the Kullback-Leibler
divergence (e.g., \citet{berk1966limiting}, \citet{bunke1998asymptotic}).\footnote{\citet{white1982maximum} shows that the Kullback-Leibler divergence
also characterizes the limiting behavior of the maximum quasi-likelihood
estimator.} We extend the statistics literature on misspecified learning to the
case where agents are not only passively learning about their environment
but are also actively learning by taking actions.

We present the framework and examples in Section \ref{sec:framework},
discuss the relationship to other solution concepts in Section \ref{sec:Relationship-to-other},
and provide a learning foundation in Section \ref{sec:foundation}.
We discuss assumptions and extension in Section \ref{subsec:Discussion}.

\section{\label{sec:framework}The framework}

\subsection{\label{subsec:The-environment}The environment}

A (simultaneous-move) \textbf{game}\emph{ }$\mathcal{G}=<\mathcal{O},\mathcal{Q}>$
is composed of a (simultaneous-move) objective game $\mathcal{O}$
and a subjective model $\mathcal{Q}$.

\emph{$\textsc{Objective game}$.} A\emph{ }(simultaneous-move)\emph{
}\textbf{objective game} is a tuple 
\[
\mathcal{O}=\left\langle I,\Omega,\mathbb{S},p,\mathbb{X},\mathbb{Y},f,\pi\right\rangle ,
\]
where: $I$ is the set of players; $\Omega$ is the set of payoff-relevant
states; $\mathbb{S}=\times_{i\in I}\mathbb{S}^{i}$ is the set of
profiles of signals, where $\mathbb{S}^{i}$ is the set of signals
of player $i$; $p$ is a probability distribution over $\Omega\times\mathbb{S}$,
and, for simplicity, it is assumed to have marginals with full support;
we use standard notation to denote marginal and conditional distributions,
e.g., $p_{\Omega\mid S^{i}}(\cdot\mid s^{i})$ denotes the conditional
distribution over $\Omega$ given $S^{i}=s^{i}$; $\mathbb{X}=\times_{i\in I}\mathbb{X}^{i}$
is a set of profiles of actions, where $\mathbb{X}^{i}$ is the set
of actions of player $i$; $\mathbb{Y}=\times_{i\in I}\mathbb{Y}^{i}$
is a set of profiles of (observable) consequences, where $\mathbb{Y}^{i}$
is the set of consequences of player $i$; $f=(f^{i})_{i\in I}$ is
a profile of feedback or consequence functions, where $f^{i}:\mathbb{X}\times\Omega\rightarrow\mathbb{Y}^{i}$
maps outcomes in $\Omega\times\mathbb{\mathbb{X}}$ into consequences
of player $i$; and $\pi=(\pi^{i})_{i\in I}$, where $\pi^{i}:\mathbb{X}^{i}\times\mathbb{Y}^{i}\rightarrow\mathbb{R}$
is the payoff function of player $i$.\footnote{The concept of a feedback function is borrowed from the SCE literature.
Also, while it is redundant to have $\pi^{i}$ depend on $x^{i}$,
it simplifies the notation in applications.} For simplicity, we prove the results for the case where all of the
above sets are finite.\footnote{In the working paper version (\citealp{EP14arxiv}), we provide technical
conditions under which the results extend to nonfinite $\Omega$ and
$\mathbb{Y}$.}

The timing of the objective game is as follows: First, a state and
a profile of signals are drawn according to $p$. Second, each player
privately observes her own signal. Third, players simultaneously choose
actions. Finally, each player observes her consequence and obtains
a payoff. We implicitly assume that players observe at least their
own actions and payoffs.\footnote{See Online Appendix \ref{sec:Lack-of-payoff} for the case where players
do not observe own payoffs.}

A strategy of player $i$ is a mapping $\sigma^{i}:\mathbb{S}^{i}\rightarrow\Delta(\mathbb{X}^{i})$.
The probability that player $i$ chooses action $x^{i}$ after observing
signal $s^{i}$ is denoted by $\sigma^{i}(x^{i}\mid s^{i})$. A strategy
profile is a vector of strategies $\sigma=(\sigma^{i})_{i\in I}$;
let $\Sigma$ denote the space of all strategy profiles.

Fix an objective game. For each strategy profile $\sigma$, there
is an \textbf{objective distribution} over player $i$'s consequences,
$Q_{\sigma}^{i}:\mathbb{S}^{i}\times\mathbb{X}^{i}\rightarrow\Delta(\mathbb{Y}^{i})$,
where 
\begin{equation}
Q_{\sigma}^{i}(y^{i}\mid s^{i},x^{i})=\sum_{\left\{ (\omega,x^{-i}):f^{i}(x^{i},x^{-i},\omega)=y^{i}\right\} }\sum_{s^{-i}}\prod_{j\neq i}\sigma^{j}(x^{j}\mid s^{j})p_{\Omega\times S^{-i}\mid S^{i}}(\omega,s^{-i}\mid s^{i}),\label{eq:Q_sigma-1}
\end{equation}
for all $(s^{i},x^{i},y^{i})\in\mathbb{S}^{i}\times\mathbb{X}^{i}\times\mathbb{Y}^{i}$.\footnote{As usual, the superscript $-i$ denotes a profile where the $i$'th
component is excluded} The objective distribution represents the true distribution over
consequences, conditional on a player's own action and signal, given
the objective game and a strategy profile followed by the players.

\emph{$\textsc{Subjective model}$}. The subjective model represents
the set of distributions over consequences that players consider possible
a priori. For a fixed objective game, a \textbf{subjective model}
is a tuple 
\[
\mathcal{Q}=\left\langle \Theta,\left(Q_{\theta}\right){}_{\theta\in\Theta}\right\rangle ,
\]
where $\Theta=\times_{i\in I}\Theta^{i}$ and $\Theta^{i}$ is player
$i$'s parameter set; and $Q_{\theta}=(Q_{\theta^{i}}^{i})_{i\in I}$,
where $Q_{\theta^{i}}^{i}:\mathbb{S}^{i}\times\mathbb{X}^{i}\rightarrow\Delta(\mathbb{Y}^{i})$
is the conditional distribution over player $i$'s consequences parameterized
by $\theta^{i}\in\Theta^{i}$; we denote the conditional distribution
by $Q_{\theta^{i}}(\cdot\mid s^{i},x^{i})$.\footnote{For simplicity, we assume that players know the distribution over
own signals.}

While the objective game represents the true environment, the subjective
model represents the players' perception of their environment. This
separation between objective and subjective models is crucial in this
paper.
\begin{rem}
\label{rem:subjective_special}A special case of a subjective model
is one where each player understands the objective game being played
but is uncertain about the distribution over states, the consequence
function, and (in the case of multiple players) the strategies of
other players. In this special case, player $i$'s uncertainty about
$p$, $f^{i}$, and $\sigma^{-i}$ can be described by a parametric
model $p_{\theta^{i}},$ $f_{\theta^{i}}^{i}$, $\sigma_{\theta^{i}}^{-i}$
, where $\theta^{i}\in\Theta^{i}$. A subjective distribution $Q_{\theta^{i}}^{i}$
is then derived by replacing $p$, $f^{i}$, and $\sigma^{-i}$ with
$p_{\theta^{i}},$ $f_{\theta^{i}}^{i}$, $\sigma_{\theta^{i}}^{-i}$
in equation (\ref{eq:Q_sigma-1}).\footnote{In this case, a player understands that other players mix independently
but, due to uncertainty over the parameter $\theta^{i}$ that indexes
$\sigma_{\theta^{i}}^{-i}=(\sigma_{\theta^{i}}^{j})_{j\neq i}$, she
may have correlated beliefs about her opponents' strategies, as in
\citet{fudenberg1993self}.} $\square$
\end{rem}
By defining $Q_{\theta^{i}}^{i}$ as a primitive, we stress two points.
First, this object is sufficient to characterize behavior. Second,
working with general subjective distributions allows for more general
types of misspecifications, where players do not even have to understand
the structural elements that determine their payoff relevant consequences.

We maintain the following assumptions about the subjective model.\bigskip{}

\textbf{Assumption 1.} For all $i\in I$: (i) $\Theta^{i}$ is a compact
subset of an Euclidean space, (ii) $Q_{\theta^{i}}^{i}(y^{i}\mid s^{i},x^{i})$
is continuous as a function of $\theta^{i}\in\Theta^{i}$ for all
$(y^{i},s^{i},x^{i})\in\mathbb{Y}^{i}\times\mathbb{S}^{i}\times\mathbb{X}^{i}$,
(iii) For all $\theta^{i}\in\Theta^{i}$, there exists a sequence
$(\theta_{n}^{i})_{n}$ in $\Theta^{i}$ such that $\lim_{n\rightarrow\infty}\theta_{n}^{i}=\theta^{i}$
and such that, for all $n$, $Q_{\theta_{n}^{i}}^{i}(y^{i}\mid s^{i},x^{i})>0$
for all $(s^{i},x^{i})\in\mathbb{S}^{i}\times\mathbb{X}^{i}$, $y^{i}\in f^{i}(x^{i},\mathbb{X}^{-i},\omega)$,
and $\omega\in supp(p_{\Omega\mid S^{i}}(\cdot\mid s^{i}))$.

\bigskip{}

Conditions (i) and (ii) are the standard conditions used to define
a \emph{parametric} model in statistics (e.g., \citet{bickel1993efficient}).
Condition (iii) plays two roles. First, it guarantees that there exists
at least one parameter value that attaches positive probability to
every feasible observation. In particular, it rules out what can be
viewed as a stark misspecification in which every element of the subjective
model attaches zero probability to an event that occurs with positive
true probability. Second, it imposes a ``richness'' condition on
the subjective model: If a feasible event is deemed impossible by
some parameter value, then that parameter value is not isolated in
the sense that there are nearby parameter values that consider every
feasible event to be possible. In Section \ref{subsec:Discussion},
we show that equilibrium may fail to exist and steady-state behavior
need not be characterized by equilibrium without this assumption.

\subsection{Examples\label{subsec:Examples}}

We illustrate the environment by presenting several examples that
had previously not been integrated into a common framework.\footnote{\citet{nyarko1991learning} studies a special case of Example 2.1
and shows that a steady state does not exist in pure strategies; \citet{sobel1984non}
considers a misspecification similar to Example 2.2; Tversky and Kahneman's
(1973)\nocite{KahnemanTversky1973} story motivates Example 2.3; Sargent
(1999, Chapter 7)\nocite{sargent-book} studies Example 2.4; and \citet{kagel1986winner},
\citet{eyster2005cursed}, \citet{jehiel2008revisiting}, and \citet{esponda2008behavioral}
study Example 2.5. See \citet{EP14arxiv} for additional examples.} In examples with a single agent, we drop the $i$ subscript from
the notation.

\bigskip{}

\textbf{Example 2.1. }\textbf{\emph{Monopolist with unknown demand}}\textbf{.
\label{Monopolist}} A monopolist faces demand $y=f(x,\omega)=\phi_{0}(x)+\omega$,
where $x\in\mathbb{X}$ is the price chosen by the monopolist and
$\omega$ is a mean-zero shock with distribution $p\in\Delta(\Omega)$.
The monopolist observes sales $y$, but not the shock. The monopolist
does not observe any signal, and so we omit signals from the notation.
The monopolist's payoff is $\pi(x,y)=xy$ (i.e., there are no costs).
The monopolist's uncertainty about $p$ and $f$ is described by a
parametric model $f_{\theta},p_{\theta}$, where $y=f_{\theta}(x,\omega)=a-bx+\omega$
is the demand function, $\theta=(a,b)\in\Theta$ is a parameter vector,
 and $\omega\sim N(0,1)$ (i.e., $p_{\theta}$ is a standard normal
distribution for all $\theta\in\Theta$). In particular, this example
corresponds to the special case discussed in Remark \ref{rem:subjective_special},
and $Q_{\theta}(\cdot\mid x)$ is a normal density with mean $a-bx$
and unit variance. $\square$\bigskip{}

\textbf{Example 2.2. }\textbf{\emph{Nonlinear taxation}}\textbf{.}
\textbf{\label{Nonlinear}} An agent chooses effort $x\in\mathbb{X}$
at cost $c(x)$ and obtains income $z=x+\omega$, where $\omega$
is a zero-mean shock with distribution $p\in\Delta(\Omega)$. The
agent pays taxes $t=\tau(z)$, where $\tau(\cdot)$ is a nonlinear
tax schedule. The agent does not observe any signal, and so we omit
them. The agent observes $y=(z,t)$ and obtains payoff $\pi(x,z,t)=z-t-c(x)$.\footnote{Formally, $f(x,\omega)=(z(x,\omega),t(x,\omega))$, where $z(x,\omega)=x+\omega$
and $t(x,\omega)=\tau(x+\omega)$.} She understands how effort translates into income but fails to realize
that the marginal tax rate depends on income. We compare two models
that capture this misspecification. In model A, the agent believes
in a random coefficient model, $t=(\theta^{A}+\varepsilon)z$, in
which the marginal and average tax rates are both equal to $\theta^{A}+\varepsilon$,
where $\theta^{A}\in\Theta^{A}=\mathbb{R}$. In model B, the agent
believes that $t=\theta_{1}^{B}+\theta_{2}^{B}z+\varepsilon$, where
$\theta_{2}^{B}$ is the constant marginal tax rate and $\theta^{B}=(\theta_{1}^{B},\theta_{2}^{B})\in\Theta^{B}=\mathbb{R}^{2}$.\footnote{It is not necessary to assume that $\Theta^{A}$ and $\Theta^{B}$
are compact for an equilibrium to exist; the same comment applies
to Examples 2.3 and 2.4} In both models, $\varepsilon\sim N(0,1)$ measures uncertain aspects
of the schedule (e.g., variations in tax rates or credits). Thus,
$Q_{\theta}^{j}(t,z\mid x)=Q_{\theta}^{j}(t\mid z)p(z-x)$, where
$Q_{\theta}^{j}(\cdot\mid z)$ is a normal density with mean $\theta^{A}z$
and variance $z^{2}$ in model $j=A$ and mean $\theta_{1}^{B}+\theta_{2}^{B}z$
and unit variance in model $j=B$. $\square$ \bigskip{}

\textbf{Example 2.3. }\textbf{\emph{Regression to the mean}}\textbf{.}
\textbf{\label{Regression}} An instructor observes the initial performance
$s$ of a student and decides to praise or criticize him, $x\in\{C,P\}$.
The student then performs again and the instructor observes his final
performance, $s'$. The truth is that performances $y=(s,s')$ are
independent, standard normal random variables. The instructor's payoff
is $\pi(x,s,s')=s'-c(x,s)$, where $c(x,s)=\kappa\left|s\right|>0$
if either $s>0,x=C$ or $s<0,x=P$, and, in all other cases, $c(x,s)=0$.\footnote{Formally, $\omega=(s,s')$, $p$ is the product of standard normal
distributions, and $y=f(x,\omega)=\omega$.} The function $c$ represents a (reputation) cost from lying (i.e.,
criticizing above-average performances or praising below-average ones)
that increases in the size of the lie. Because the instructor cannot
influence performance, it is optimal to praise if $s>0$ and to criticize
if $s<0$. The instructor, however, does not admit the possibility
of regression to the mean and believes that $s'=s+\theta_{x}+\varepsilon$,
where $\varepsilon\sim N(0,1)$, and $\theta=(\theta_{C},\theta_{P})\in\Theta$
parameterizes her perceived influence on performance.\footnote{A model that allows for regression to the mean is $s'=\alpha s+\theta_{x}+\varepsilon$;
in this case, the agent would correctly learn that $\alpha=0$ and
$\theta_{x}=0$ for all $x$. \citet{rabin2010gambler} study a related
setup in which the agent believes that shocks are autoregressive when
in fact they are i.i.d.} Thus, letting $\bar{Q}_{\theta}(\cdot\mid x)$ be the a normal density
with mean $s+\theta_{x}$ and unit variance, it follows that $Q_{\theta}(\hat{s},s^{'}\mid s,x)=\bar{Q}_{\theta}(s'\mid s,x)$
if $\hat{s}=s$ and $0$ otherwise. $\square$

\bigskip{}

\textbf{Example 2.4. }\textbf{\emph{Monetary policy}}\textbf{.} \textbf{\label{Classical}}Two
players, the government (G) and the public (P), i.e., $I=\{G,P\}$,
choose monetary policy $x^{G}$ and inflation forecasts $x^{P}$,
respectively. They do not observe signals, and so we omit them. Inflation,
$e$, and unemployment, $U$, are determined by\footnote{Formally, $\omega=(\varepsilon_{e},\varepsilon_{U})$ and $y=(e,U)=f(x^{G},x^{P},\omega)$
is given by equations (\ref{eq:infl}) and (\ref{eq:unempl}).}
\begin{eqnarray}
e & = & x^{G}+\varepsilon_{e}\label{eq:infl}\\
U & = & u^{\ast}-\lambda(e-x^{P})+\varepsilon_{U},\label{eq:unempl}
\end{eqnarray}
where $u^{*}>0$, $\lambda\in(0,1)$ and $\omega=(\varepsilon_{e},\varepsilon_{U})\in\Omega=\mathbb{R}^{2}$
are shocks with a full support distribution $p\in\Delta(\Omega)$
and $Var(\varepsilon_{e})>0$. The public and the government observe
realized inflation and unemployment, but not the error terms. The
government's payoff is $\pi(x^{G},e,U)=-(U^{2}+e^{2})$. For simplicity,
we focus on the government's problem and assume that the public has
correct beliefs and chooses $x^{P}=x^{G}$. The government understands
how its policy $x^{G}$ affects inflation, but does not realize that
unemployment is affected by \emph{surprise} inflation:

\begin{equation}
\begin{array}{lll}
U & = & \theta_{1}-\theta_{2}e+\varepsilon_{U}.\end{array}\label{eq:(e,U)-1-1}
\end{equation}
The subjective model is parameterized by $\theta=(\theta_{1},\theta_{2})\in\Theta$,
and it follows that $Q_{\theta}(e,U\mid x^{G})$ is the density implied
by the equations (\ref{eq:infl}) and (\ref{eq:(e,U)-1-1}). $\square$

\bigskip{}

\textbf{Example 2.5. }\textbf{\emph{Trade with adverse selection.}}
\textbf{\label{Trade}} A buyer with valuation $v\in\mathbb{V}$ and
a seller submit a (bid) price $x\in\mathbb{X}$ and an ask price $a\in\mathbb{A}$,
respectively. The seller's ask price and the buyer's value are drawn
from $p\in\Delta(\mathbb{A}\times\mathbb{V})$, so that $\Omega=\mathbb{A}\times\mathbb{V}$
is the state space. Thus, the buyer is the only decision maker.\footnote{The typical story is that there is a population of sellers each of
whom follows the weakly dominant strategy of asking for her valuation;
thus, the ask price is a function of the seller's valuation and, if
buyer and seller valuations are correlated, then the ask price and
buyer valuation are also correlated.} After submitting a price, the buyer observes $y=\omega=(a,v)$ and
gets payoff $\pi(x,a,v)=v-x$ if $a\leq x$ and zero otherwise. In
other words, the buyer observes perfect feedback, gets $v-x$ if there
is trade, and $0$ otherwise. When making an offer, she does not know
her value or the seller's ask price. She also does not observe any
signals, and so we omit them. Finally, suppose that $A$ and $V$
are correlated but that the buyer believes they are independent. This
is captured by letting $Q_{\theta}=\theta$ and $\Theta=\Delta(\mathbb{A})\times\Delta(\mathbb{V})$.
$\square$

\subsection{Definition of equilibrium}

$\textsc{Distance to true model}.$ In equilibrium, we will require
players' beliefs to put probability one on the set of subjective distributions
over consequences that are ``closest'' to the objective distribution.
The following function, which we call the \textbf{weighted Kullback-Leibler
divergence}\emph{ }(wKLD)\emph{ }function of player\emph{ $i$}, is
a weighted version of the standard Kullback-Leibler divergence in
statistics (\citealp{KullbackLeibler1951}). It represents a ``distance''
between the objective distribution over $i$'s consequences given
a strategy profile $\sigma\in\Sigma$ and the distribution as parameterized
by $\theta^{i}\in\Theta^{i}$:\footnote{The notation $E_{Q}$ denotes expectation with respect to the probability
distribution $Q$. Also, we use the convention that $-\ln0=\infty$
and $0\ln0=0$.} 
\begin{equation}
K^{i}(\sigma,\theta^{i})=\sum_{(s^{i},x^{i})\in\mathbb{S}^{i}\times\mathbb{X}^{i}}E_{Q_{\sigma}^{i}(\cdot\mid s^{i},x^{i})}\left[\ln\frac{Q_{\sigma}^{i}(Y^{i}\mid s^{i},x^{i})}{Q_{\theta^{i}}^{i}(Y^{i}\mid s^{i},x^{i})}\right]\sigma^{i}(x^{i}\mid s^{i})p_{S^{i}}(s^{i}).\label{eq:wKLD-1}
\end{equation}
The set of closest parameter values of player $i$ given $\sigma$
is the set 
\[
\Theta^{i}(\sigma)\equiv\arg\min_{\theta^{i}\in\Theta^{i}}K^{i}(\sigma,\theta^{i}).
\]
The interpretation is that $\Theta^{i}(\sigma)\subset\Theta^{i}$
is the set of parameter values that player $i$ can believe to be
possible after observing feedback consistent with strategy profile
$\sigma$.
\begin{rem}
We show in Section \ref{sec:foundation} that wKLD is the right notion
of distance in a learning model with Bayesian players. Here, we provide
an heuristic argument for a Bayesian agent (we drop $i$ subscripts
for clarity) with parameter set $\Theta=\{\theta_{1},\theta_{2}\}$
who observes data over $t$ periods, $(s_{\tau},x_{\tau},y_{\tau})_{\tau=0}^{t-1}$,
that comes from repeated play of the objective game under strategy
$\sigma$. Let $\rho_{0}=\mu_{0}(\theta_{2})/\mu_{0}(\theta_{1})$
denote the agent's ratio of priors. Applying Bayes' rule and simple
algebra, the posterior probability over $\theta_{1}$ after $t$ periods
is 
\begin{align*}
\mu_{t}(\theta_{1}) & =\left(1+\rho_{0}\Pi_{\tau=0}^{t-1}\frac{Q_{\theta_{2}}(y_{\tau}\mid s_{\tau},x_{\tau})}{Q_{\theta_{1}}(y_{\tau}\mid s_{\tau},x_{\tau})}\right)^{-1}=\left(1+\rho_{0}\Pi_{\tau=0}^{t-1}\frac{Q_{\theta_{2}}(y_{\tau}\mid s_{\tau},x_{\tau})/Q_{\sigma}(y_{\tau}\mid s_{\tau},x_{\tau})}{Q_{\theta_{1}}(y_{\tau}\mid s_{\tau},x_{\tau})/Q_{\sigma}(y_{\tau}\mid s_{\tau},x_{\tau})}\right)^{-1}\\
 & =\left(1+\rho_{0}\exp\left\{ -t\left(\frac{1}{t}\sum_{\tau=0}^{t-1}\ln\frac{Q_{\sigma}(y_{\tau}\mid s_{\tau},x_{\tau})}{Q_{\theta_{2}}(y_{\tau}\mid s_{\tau},x_{\tau})}-\frac{1}{t}\sum_{\tau=0}^{t-1}\ln\frac{Q_{\sigma}(y_{\tau}\mid s_{\tau},x_{\tau})}{Q_{\theta_{1}}(y_{\tau}\mid s_{\tau},x_{\tau})}\right)\right\} \right)^{-1}
\end{align*}
where the second equality follows by multiplying and dividing by $\Pi_{\tau=0}^{t-1}Q_{\sigma}(y_{\tau}\mid s_{\tau},x_{\tau})$.
By a law of large numbers argument and the fact that the true joint
distribution over $(s,x,y)$ is given by $Q_{\sigma}(y\mid x,s)\sigma(x\mid s)p_{S}(s)$,
the difference in the log-likelihood ratios converges to $K(\sigma,\theta_{2})-K(\sigma,\theta_{1})$.
Suppose that $K(\sigma,\theta_{1})>K(\sigma,\theta_{2})$. Then, for
sufficiently large $t$, the posterior belief $\mu_{t}(\theta_{1})$
is approximately equal to $1/(1+\rho_{0}\exp\left(-t\left(K(\sigma,\theta_{2})-K(\sigma,\theta_{1})\right)\right))$,
which converges to 0. Therefore, the posterior eventually assigns
zero probability to $\theta_{1}$. On the other hand, if $K(\sigma,\theta_{1})<K(\sigma,\theta_{2})$,
then the posterior eventually assigns zero probability to $\theta_{2}$.
Thus, the posterior eventually assigns zero probability to parameter
values that do not minimize $K(\sigma,\cdot)$. $\square$
\end{rem}
\smallskip{}

\begin{rem}
Because the wKLD function is weighted by a player's own strategy,
it places no restrictions on beliefs about outcomes that only arise
following out-of-equilibrium actions (beyond the restrictions imposed
by $\Theta$). $\square$
\end{rem}
The next result collects some useful properties of the wKLD function.

\bigskip{}

\begin{lem}
\label{lemma:Theta-1}(i) For all $i\in I$, $\theta^{i}\in\Theta^{i}$,
and $\sigma\in\Sigma$, $K^{i}(\sigma,\theta^{i})\geq0$, with equality
holding if and only if $Q_{\theta^{i}}(\cdot\mid s^{i},x^{i})=Q_{\sigma}^{i}(\cdot\mid s^{i},x^{i})$
for all $(s^{i},x^{i})$ such that $\sigma^{i}(x^{i}\mid s^{i})>0$.
(ii) For all $i\in I$, $\Theta^{i}(\cdot)$ is nonempty, upper hemicontinuous,
and compact valued.
\end{lem}
\begin{proof}
See the Appendix.
\end{proof}
\bigskip{}

The upper-hemicontinuity of $\Theta^{i}(\cdot)$ would follow from
the Theorem of the Maximum had we assumed $Q_{\theta^{i}}^{i}$ to
be positive for all feasible events and $\theta^{i}\in\Theta^{i}$,
since the wKLD function would then be finite and continuous. But this
assumption may be strong in some cases.\footnote{For example, it rules out cases where a player believes others follow
pure strategies.} Assumption 1(iii) weakens this assumption by requiring that it holds
for a dense subset of $\Theta$, and still guarantees that $\Theta^{i}(\cdot)$
is upper hemicontinuous.

$\textsc{Optimality.}$ In equilibrium, we will require each player
to choose a strategy that is optimal given her beliefs. A strategy
$\sigma^{i}$ for player $i$ is \textbf{optimal} given\emph{ $\mu^{i}\in\Delta(\Theta^{i})$}
if $\sigma^{i}(x^{i}\mid s^{i})>0$ implies that 
\begin{equation}
x^{i}\in\arg\max_{\bar{x}^{i}\in\mathbb{X}^{i}}E_{\bar{Q}_{\mu^{i}}^{i}(\cdot\mid s^{i},\bar{x}^{i})}\left[\pi^{i}(\bar{x}^{i},Y^{i})\right]\label{eq:optimize-1}
\end{equation}
where $\bar{Q}_{\mu^{i}}^{i}(\cdot\mid s^{i},x^{i})=\int_{\Theta^{i}}Q_{\theta^{i}}^{i}(\cdot\mid s^{i},x^{i})\mu^{i}(d\theta^{i})$
is the distribution over consequences of player $i$, conditional
on $(s^{i},x^{i})\in\mathbb{S}^{i}\times\mathbb{X}^{i}$, induced
by\emph{ $\mu^{i}$}.

$\textsc{Definition of equilibrium.}$ We propose the following solution
concept.\bigskip{}

\begin{defn}
\label{def:equilibrium-1}A strategy profile $\sigma$ is a \textbf{Berk-Nash
equilibrium} of game\emph{ $\mathcal{G}$} if, for all players $i\in I$,
there exists $\mu^{i}\in\Delta(\Theta^{i})$ such that

(i) $\sigma^{i}$ is optimal given $\mu^{i}$, and 

(ii) $\mu^{i}\in\Delta(\Theta^{i}(\sigma))$, i.e., if $\hat{\theta}^{i}$
is in the support of $\mu^{i}$, then 
\[
\hat{\theta}^{i}\in\arg\min_{\theta^{i}\in\Theta^{i}}K^{i}(\sigma,\theta^{i}).
\]
\end{defn}
\bigskip{}

Definition \ref{def:equilibrium-1} places two restrictions on equilibrium
behavior: (i) optimization given beliefs, and (ii) endogenous restrictions
on beliefs. For comparison, note that the definition of Nash equilibrium
is identical to Definition \ref{def:equilibrium-1} except that condition
(ii) is replaced with the condition that players have correct beliefs,
i.e., $\bar{Q}_{\mu^{i}}^{i}=Q_{\sigma}^{i}$.

$\textsc{existence of equilibrium.}$ The standard existence proof
of Nash equilibrium cannot be used here because the analogous version
of a best response correspondence is not necessarily convex valued.
To prove existence, we first perturb payoffs and establish that equilibrium
exists in the perturbed game. We then consider a sequence of equilibria
of perturbed games, where perturbations go to zero, and establish
that the limit is a Berk-Nash equilibrium of the (unperturbed) game.\footnote{The idea of perturbations and the strategy of the existence proof
date back to \citet{harsanyi1973games}; \citet{selten1975reexamination}
and \citet{kreps1982sequential} also used these ideas to prove existence
of perfect and sequential equilibrium, respectively.} The nonstandard part of the proof is to prove existence of equilibrium
in the perturbed game. The perturbed best response correspondence
is still not necessarily convex valued. Our approach is to characterize
equilibrium as a fixed point of a \emph{belief correspondence} and
show that it satisfies the requirements of a generalized version of
Kakutani's fixed point theorem.\bigskip{}

\begin{thm}
\label{theo:Existence}Every game has at least one Berk-Nash equilibrium.
\end{thm}
\begin{proof}
See the Appendix.
\end{proof}

\subsection{Examples: Finding a Berk-Nash equilibrium}

\textbf{Example 2.1, continued from pg. \pageref{Monopolist}. }\textbf{\emph{Monopolist
with unknown demand}}\textbf{.} Let $\sigma=(\sigma_{x})_{x\in\mathbb{X}}$
denote a strategy, where $\sigma_{x}$ is the probability of choosing
price $x\in\mathbb{X}$. Because this is a single-agent problem, the
objective distribution does not depend on $\sigma$; hence, we denote
it by $Q_{0}(\cdot\mid x)$, which is a normal density with mean $\phi_{0}(x)$
and unit variance. Similarly, $Q_{\theta}(\cdot\mid x)$ is a normal
density with mean $\phi_{\theta}(x)=a-bx$ and unit variance. It follows
from equation (\ref{eq:wKLD-1}) that

\begin{align*}
K(\sigma,\theta) & =\sum_{x\in\mathbb{X}}\sigma_{x}\frac{1}{2}E_{Q_{0}(\cdot\mid x)}\left[\left(Y-\phi_{\theta}(x)\right)^{2}-\left(Y-\phi_{0}(x)\right)^{2}\right]=\sum_{x\in\mathbb{X}}\sigma_{x}\frac{1}{2}\left(\phi_{0}(x)-\phi_{\theta}(x)\right)^{2}.
\end{align*}

For concreteness, let $\mathbb{X}=\{2,10\}$, $\phi_{0}(2)=34$ and
$\phi_{0}(10)=2$, and $\Theta=[33,40]\times[3,3.5]$.\footnote{In particular, the deterministic part of the demand function can have
any functional form provided it passes through $(2,\phi_{o}(2))$
and $(10,\phi_{0}(10))$.} Let $\theta^{0}\in\mathbb{R}^{2}$ provide a perfect fit for demand,
i.e., $\phi_{\theta^{0}}(x)=\phi_{0}(x)$ for all $x\in\mathbb{X}$.
In this example, $\theta^{0}=(a^{0},b^{0})=(42,4)\notin\Theta$ and,
therefore, we say that the monopolist has a \emph{misspecified} model.
The dashed line in Figure \ref{fig:Monopolist} depicts optimal behavior:
the optimal price is 10 to the left, it is 2 to the right, and the
monopolist is indifferent for parameter values on the dashed line.

To solve for equilibrium, we first consider pure strategies. If $\sigma=(0,1)$
(i.e., the price is $x=10$), the first order conditions $\partial K(\sigma,\theta)/\partial a=\partial K(\sigma,\theta)/\partial b=0$
imply $\phi_{0}(10)=\phi_{\theta}(10)=a-b10$, and any $(a,b)\in\Theta$
on the segment $AB$ in Figure \ref{fig:Monopolist} minimizes $K(\sigma,\cdot)$.
These minimizers, however, lie to the right of the dashed line, where
it is \emph{not} optimal to set a price of 10. Thus, $\sigma=(0,1)$
is not an equilibrium. A similar argument establishes that $\sigma=(1,0)$
is not an equilibrium: If it were, the minimizer would be at $D$,
where it is in fact not optimal to choose a price of 2.

\begin{figure}
\begin{raggedright}
\hspace*{-1.75cm}\includegraphics[scale=0.6]{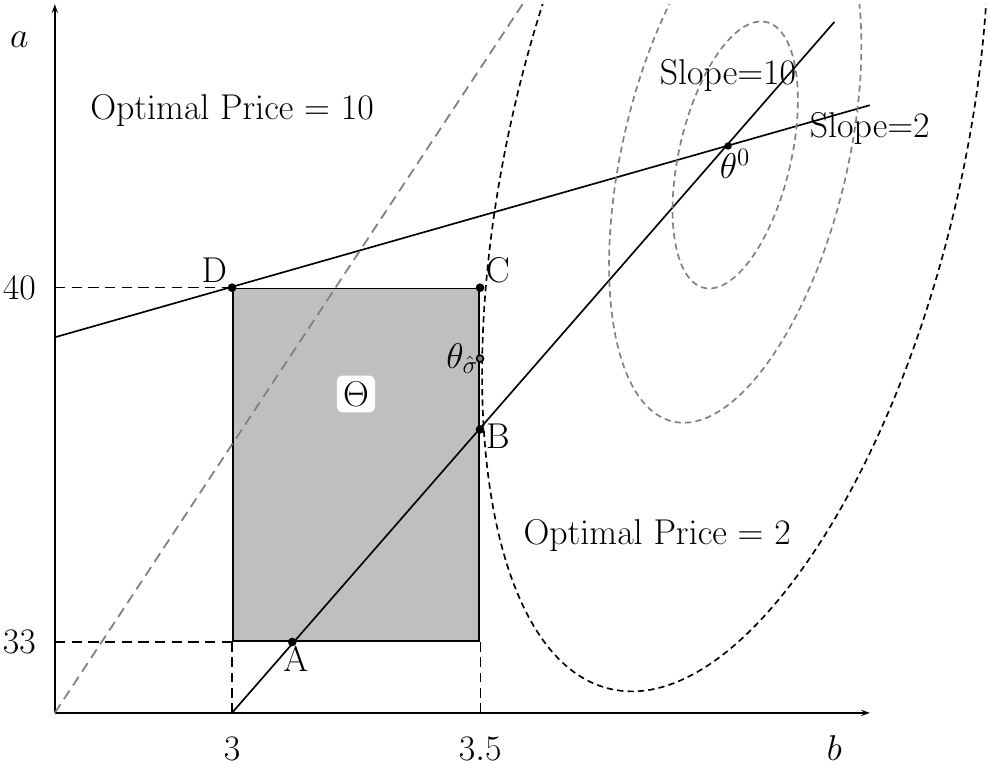}
\includegraphics[scale=0.6]{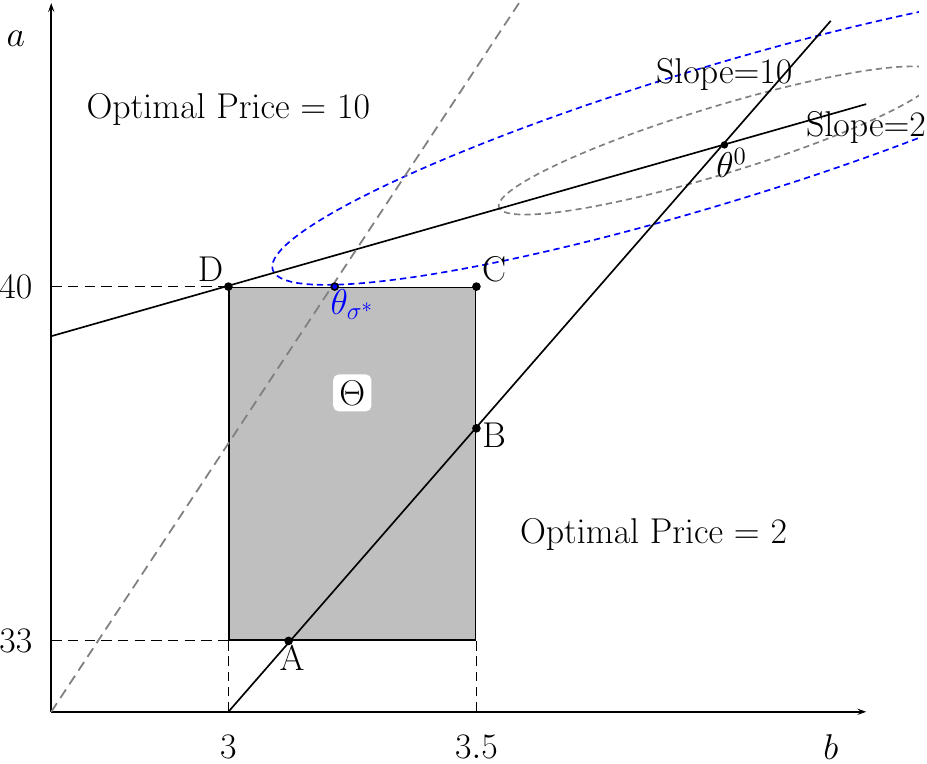}
\par\end{raggedright}
\caption{\label{fig:Monopolist}Monopolist with misspecified demand function.\protect \\
{\small{}Left panel: the parameter value that minimizes the wKLD function
given strategy $\hat{\sigma}$ is $\theta_{\hat{\sigma}}$. Right
panel: $\sigma^{*}$ is a Berk-Nash equilibrium ($\sigma^{*}$ is
optimal given $\theta_{\sigma^{*}}$---because $\theta_{\sigma^{*}}$
lies on the indifference line---and $\theta_{\sigma^{*}}$ minimizes
the wKLD function given $\sigma^{*}$).}}
\end{figure}
Finally, consider mixed strategies. Because both first order conditions
cannot hold simultaneously, the parameter value that minimizes $K(\sigma,\theta)$
lies on the boundary of $\Theta$. A bit of algebra shows that, for
any totally mixed $\sigma$, there is a unique minimizer $\theta_{\sigma}=(a_{\sigma},b_{\sigma})$
characterized as follows. If $\sigma_{2}\leq3/4$, the minimizer is
on the segment $BC$: $b_{\sigma}=3.5$ and $a_{\sigma}=4\sigma_{2}+37$
solves $\partial K(\sigma,\theta)/\partial a=0$. The left panel of
Figure \ref{fig:Monopolist} depicts an example where the unique minimizer
$\theta_{\hat{\sigma}}$ under strategy $\hat{\sigma}$ is given by
the tangency between the contour lines of $K(\hat{\sigma},\cdot)$
and the feasible set $\Theta$.\footnote{It can be shown that $K(\sigma,\theta)=\left(\theta-\theta^{0}\right)'M_{\sigma}\left(\theta-\theta^{0}\right)$,
where $M_{\sigma}$ is a weighting matrix that depends on $\sigma$.
In particular, the contour lines of $K(\sigma,\cdot)$ are ellipses.} If $\sigma_{2}\in[3/4,15/16]$, then $\theta_{\sigma}=C$ is the
northeast vertex of $\Theta$. Finally, if $\sigma_{2}>15/16$, the
minimizer is on the segment $DC$: $a_{\sigma}=40$ and $b_{\sigma}=(380-368\sigma_{2})/(100-96\sigma_{2})$
solves $\partial K(\sigma,\theta)/\partial b=0$.

Because the monopolist mixes, optimality requires that the equilibrium
belief $\theta_{\sigma}$ lies on the dashed line. The unique Berk-Nash
equilibrium is $\sigma^{*}=(35/36,1/36)$, and its supporting belief,
$\theta_{\sigma^{*}}=(40,10/3)$, is given by the intersection of
the dashed line and the segment $DC$, as depicted in the right panel
of Figure \ref{fig:Monopolist}. It is \emph{not} the case, however,
that the equilibrium belief about the mean of $Y$ is correct. Thus,
an approach that had focused on fitting the mean, rather than minimizing
$K$, would have led to the wrong conclusion.\footnote{The example also illustrates the importance of mixed strategies for
existence of Berk-Nash equilibrium, even in single-agent settings.
As an antecedent, \citet{esponda2011learning} argue that this is
the reason why mixed strategy equilibrium cannot be purified in a
voting application.} $\square$

\bigskip{}

\textbf{Example 2.2, continued from pg. \pageref{Nonlinear}. }\textbf{\emph{Nonlinear
}}\textbf{taxation.} For any pure strategy $x$ and parameter value
$\theta\in\Theta^{A}=\mathbb{R}$ (model A) or $\theta\in\Theta^{B}=\mathbb{R}^{2}$
(model B), the wKLD function $K^{j}(x,\theta)$ for model $j\in\{A,B\}$
equals 
\begin{align*}
E\Bigl[\ln\frac{Q(T\mid Z)p(Z-x)}{Q_{\theta}^{j}(T\mid Z)p(Z-x)}\mid X=x\Bigr]=\begin{cases}
-\frac{1}{2}E\left[\left(\tau(Z)/Z-\theta^{A}\right)^{2}\mid X=x\right]+C_{A} & \mbox{(model A)}\\
-\frac{1}{2}E\left[\left(\tau(Z)-\theta_{1}^{B}-\theta_{2}^{B}Z\right)^{2}\mid X=x\right]+C_{B} & \mbox{(model B)}
\end{cases}
\end{align*}
where $E$ denotes the true conditional expectation and $C_{A}$ and
$C_{B}$ are constants.

For model A, $\theta^{A}(x)=E\left[\tau(x+W)/(x+W)\right]$ is the
unique parameter that minimizes $K^{A}(x,\cdot)$.\footnote{We use $W$ to denote the random variable that takes on realizations
$\omega$.} Intuitively, the agent believes that the expected marginal tax rate
is equal to the true expected average tax rate. For model B, 
\begin{align*}
\theta_{2}^{B}(x) & =Cov(\tau(x+W),x+W)/Var(x+W)=E\left[\tau'(x+W)\right],
\end{align*}
where the second equality follows from Stein's lemma (\citet{stein1972}),
provided that $\tau$ is differentiable. Intuitively, the agent believes
that the marginal tax rate is constant and given by the true expected
marginal tax rate.\footnote{By linearity and normality, the minimizers of $K^{B}(x,\cdot)$ coincide
with the OLS estimands. We assume normality for tractability, although
the framework allows for general distributional assumptions. There
are other tractable distributions; for example, the minimizer of wKLD
under the Laplace distribution corresponds to the estimates of a median
(not a linear) regression.}

We now compare equilibrium under these two models with the case in
which the agent has correct beliefs and chooses an optimal strategy
$x^{opt}$ that maximizes $x-E[\tau(x+W)]-c(x)$. In contrast, a strategy
$x_{*}^{j}$ is a Berk-Nash equilibrium of model $j$ if and only
if $x=x_{*}^{j}$ maximizes $x-\theta^{j}(x_{*}^{j})x-c(x)$.

For example, suppose that the cost of effort and true tax schedule
are both smooth functions, increasing and convex (e.g., taxes are
progressive) and that $\mathbb{X}\subset\mathbb{R}$ is a compact
interval. Then first order conditions are sufficient for optimality,
and $x^{opt}$ is the unique solution to $1-E[\tau'(x^{opt}+W)]=c'(x^{opt})$.
Moreover, the unique Berk-Nash equilibrium solves $1-E\left[\tau(x_{*}^{A}+W)/(x_{*}^{A}+W)\right]=c'(x_{*}^{A})$
for model A and $1-E[\tau'(x_{*}^{B}+W)]=c'(x_{*}^{B})$ for model
B. In particular, effort in model B is optimal, $x_{*}^{B}=x^{opt}$.
Intuitively, the agent has correct beliefs about the true expected
marginal tax rate at her equilibrium choice of effort, and so she
has the right incentives on the margin, despite believing incorrectly
that the marginal tax rate is constant. In contrast, effort is higher
than optimal in model A, $x_{*}^{A}>x^{opt}$. Intuitively, the agent
believes that the expected marginal tax rate equals the true expected
average tax rate, which is lower than the true expected marginal tax
rate in a progressive system.$\square$

\bigskip{}
\textbf{Example 2.3, continued from pg. \pageref{Regression}. }\textbf{\emph{Regression
to the mean}}\textbf{.} Since optimal strategies are characterized
by a cutoff, we let $\sigma\in\mathbb{R}$ represent the strategy
where the instructor praises an initial performance if it is above
$\sigma$ and criticizes it otherwise. The wKLD function for any $\theta\in\Theta=\mathbb{R}^{2}$
is
\[
K(\sigma,\theta)=\int_{-\infty}^{\sigma}E\left[\ln\frac{\varphi(S_{2})}{\varphi(S_{2}-(\theta_{C}+s_{1}))}\right]\varphi(s_{1})ds_{1}+\int_{\sigma}^{\infty}E\left[\ln\frac{\varphi(S_{2})}{\varphi(S_{2}-(\theta_{P}+s_{1}))}\right]\varphi(s_{1})ds_{1},
\]
where $\varphi$ is the density of $N(0,1)$ and $E$ denotes the
true expectation. For each $\sigma$, the unique parameter vector
that minimizes $K(\sigma,\cdot)$ is 
\begin{align*}
\theta_{C}(\sigma) & =E\left[S_{2}-S_{1}\mid S_{1}<\sigma\right]=-E\left[S_{1}\mid S_{1}<\sigma\right]>0
\end{align*}
and, similarly, $\theta_{P}(\sigma)=-E\left[S_{1}\mid S_{1}>\sigma\right]<0$.
Intuitively, the instructor is critical for performances below a threshold
and, therefore, the mean performance conditional on a student being
criticized is lower than the unconditional mean performance; thus,
a student who is criticized delivers a better next performance in
expectation. Similarly, a student who is praised delivers a worse
next performance in expectation.

The instructor who follows a strategy cutoff $\sigma$ believes, after
observing initial performance $s_{1}>0$, that her expected payoff
is $s_{1}+\theta_{C}(\sigma)-\kappa s_{1}$ if she criticizes and
$s_{1}+\theta_{P}(\sigma)$ if she praises. By optimality, the cutoff
makes her indifferent between praising and criticizing. Thus, $\sigma^{*}=(1/\kappa)\left(\theta_{C}(\sigma^{*})-\theta_{P}(\sigma^{*})\right)>0$
is the unique equilibrium cutoff. An instructor who ignores regression
to the mean has incorrect beliefs about the influence of her feedback
on the student's performance: She is excessively critical in equilibrium
because she incorrectly believes that criticizing a student improves
performance and that praising a student worsens it. Moreover, as the
reputation cost $\kappa\rightarrow0$, meaning that instructors care
only about performance and not about lying, $\sigma^{*}\rightarrow\infty$:
instructors only criticize (as in Tversky and Kahneman's (1973)\nocite{KahnemanTversky1973}
story). $\square$

\section{\label{sec:Relationship-to-other}Relationship to other solution
concepts}

We show that Berk-Nash equilibrium includes several solution concepts
(both standard and boundedly rational) as special cases.

\subsection{Properties of games}

\emph{$\textsc{correctly-specified games}.$ }In Bayesian statistics,
a model is correctly specified if the support of the prior includes
the true data generating process. The extension to single-agent decision
problems is straightforward. In games, however, we must account for
the fact that the objective distribution over consequences (i.e.,
the true model) depends on the strategy profile.\footnote{It would be more precise to say that the game is correctly specified
\emph{in steady state}.}\bigskip{}

\begin{defn}
\label{def:CSSS}A game is \textbf{correctly specified given $\boldsymbol{\sigma}$}
if, for all $i\in I$, there exists $\theta^{i}\in\Theta^{i}$ such
that $Q_{\theta^{i}}^{i}(\cdot\mid s^{i},x^{i})=Q_{\sigma}^{i}\left(\cdot\mid s^{i},x^{i}\right)$
for all for all $(s^{i},x^{i})\in\mathbb{S}^{i}\times\mathbb{X}^{i}$;
otherwise, the game is \textbf{misspecified} \textbf{given} \textbf{$\boldsymbol{\sigma}$}.
A game is \textbf{correctly specified} if it is correctly specified
for all $\sigma$; otherwise, it is \textbf{misspecified}.\bigskip{}
\end{defn}
\emph{$\textsc{identification}$.} From the player's perspective,
what matters is identification of the distribution over consequences
$Q_{\theta}^{i}$, not the parameter $\theta$. If the model is correctly
specified, then the true $Q_{\sigma}^{i}$ is trivially identified.
Of course, this is not true if the model is misspecified, because
the true distribution will never be learned. But we want a definition
that captures the same spirit: If two distributions are judged to
be equally a best fit (given the true distribution), then we want
these two distributions to be identical; otherwise, we cannot identify
which distribution is a best fit. The fact that players take actions
introduces an additional nuance to the definition of identification.
We can ask for identification of the distribution over consequences
either for those actions that are taken by the player (i.e., on the
path of play) or for all actions (i.e., on and off the path).

\bigskip{}

\begin{defn}
\label{def:identifiable}A game is \textbf{weakly identified given
$\boldsymbol{\sigma}$} if, for all $i\in I$: if $\theta_{1}^{i},\theta_{2}^{i}\in\Theta^{i}(\sigma)$,
then $Q_{\theta_{1}^{i}}^{i}(\cdot\mid s^{i},x^{i})=Q_{\theta_{2}^{i}}^{i}(\cdot\mid s^{i},x^{i})$
for all $(s^{i},x^{i})\in\mathbb{S}^{i}\times\mathbb{X}^{i}$ such
that $\sigma^{i}(x^{i}\mid s^{i})>0$ (recall that $p_{S^{i}}$ has
full support). If the condition is satisfied for \emph{all} $(s^{i},x^{i})\in\mathbb{S}^{i}\times\mathbb{X}^{i}$,
then we say that the game is \textbf{strongly identified given $\boldsymbol{\sigma}$}.
A game is \textbf{{[}weakly or strongly{]} identified} if it is {[}weakly
or strongly{]} identified for all $\sigma$.
\end{defn}
\bigskip{}

A correctly specified game is weakly identified. Also, two games that
are identical except for their feedback may differ in terms of being
correctly specified or identified.

\subsection{Relationship to Nash and self-confirming equilibrium}

The next result shows that Berk-Nash equilibrium is equivalent to
Nash equilibrium when the game is both correctly specified and strongly
identified.

\bigskip{}

\begin{prop}
\label{prop:NE}(i) Suppose that the game is correctly specified given
$\sigma$ and that $\sigma$ is a Nash equilibrium of its objective
game. Then $\sigma$ is a Berk-Nash equilibrium of the (objective
and subjective) game; (ii) Suppose that $\sigma$ is a Berk-Nash equilibrium
of a game that is correctly specified and strongly identified given
$\sigma$. Then $\sigma$ is a Nash equilibrium of the corresponding
objective game.
\end{prop}
\begin{proof}
(i) Let $\sigma$ be a Nash equilibrium and fix any $i\in I$. Then
$\sigma^{i}$ is optimal given $Q_{\sigma}^{i}$. Because the game
is correctly specified given $\sigma$, there exists $\theta_{*}^{i}\in\Theta^{i}$
such that $Q_{\theta_{*}^{i}}^{i}=Q_{\sigma}^{i}$ and, therefore,
by Lemma \ref{lemma:Theta-1}, $\theta_{*}^{i}\in\Theta^{i}(\sigma)$.
Thus, $\sigma^{i}$ is also optimal given $Q_{\theta_{*}^{i}}^{i}$
and $\theta_{*}^{i}\in\Theta^{i}(\sigma)$, so that $\sigma$ is a
Berk-Nash equilibrium. (ii) Let $\sigma$ be a Berk-Nash equilibrium
and fix any $i\in I$. Then $\sigma^{i}$ is optimal given $\bar{Q}_{\mu^{i}}^{i}$,
for some $\mu^{i}\in\Delta(\Theta^{i}(\sigma))$. Because the game
is correctly specified given $\sigma$, there exists $\theta_{*}^{i}\in\Theta^{i}$
such that $Q_{\theta_{*}^{i}}^{i}=Q_{\sigma}^{i}$ and, therefore,
by Lemma \ref{lemma:Theta-1}, $\theta_{*}^{i}\in\Theta^{i}(\sigma)$.
Moreover, because the game is strongly identified given $\sigma$,
any $\hat{\theta}^{i}\in\Theta^{i}(\sigma)$ satisfies $Q_{\hat{\theta}^{i}}^{i}=Q_{\theta_{*}^{i}}^{i}=Q_{\sigma}^{i}$.
Then $\sigma^{i}$ is also optimal given $Q_{\sigma}^{i}$. Thus,
$\sigma$ is a Nash equilibrium.
\end{proof}
\bigskip{}

\textbf{Example 2.4, continued from pg. \pageref{Classical}. }\textbf{\emph{Monetary
policy}}\textbf{. }Fix a strategy $x_{*}^{P}$ for the public. Note
that $U=u^{*}-\lambda(x^{G}-x_{*}^{P}+\varepsilon_{e})+\varepsilon_{U}$,
whereas the government believes $U=\theta_{1}-\theta_{2}(x^{G}+\varepsilon_{e})+\varepsilon_{U}$.
Thus, by choosing $\theta^{*}\in\Theta=\mathbb{R}^{2}$ such that
$\theta_{1}^{*}=u^{*}+\lambda x_{*}^{P}$ and $\theta_{2}^{*}=\lambda$,
it follows that the distribution over $Y=(U,e)$ parameterized by
$\theta^{*}$ coincides with the objective distribution given $x_{*}^{P}$.
So, despite appearances, the game is correctly specified given $x_{*}^{P}$.
Moreover, since $Var(\varepsilon_{e})>0$, $\theta^{*}$ is the unique
minimizer of the wKLD function given $x_{*}^{P}$. Because there is
a unique minimizer, then the game is strongly identified given $x_{*}^{P}$.
Since these properties hold for all $x_{*}^{P}$, Proposition \ref{prop:NE}
implies that Berk-Nash equilibrium is equivalent to Nash equilibrium.
Thus, the equilibrium policies are the same whether or not the government
realizes that unemployment is driven by surprise, not actual, inflation.\footnote{\citet{sargent-book} derived this result for a government doing OLS-based
learning (a special case of our example when errors are normal). We
assumed linearity for simplicity, but the result is true for the more
general case with true unemployment $U=f^{U}(x^{G},x^{P},\omega)$
and subjective model $f_{\theta}^{U}(x^{G},x^{P},\omega)$ if, for
all $x^{P}$, there exists $\theta$ such that $f^{U}(x^{G},x^{P},\omega)=f_{\theta}^{U}(x^{G},x^{P},\omega)$
for all $(x^{P},\omega)$.} $\square$ \bigskip{}

The next result shows that a Berk-Nash equilibrium is a self-confirming
equilibrium (SCE) in games that are correctly specified, but not necessarily
strongly identified.\footnote{A strategy profile $\sigma$ is a SCE if, for all players $i\in I$,
$\sigma^{i}$ is optimal given $\hat{Q}_{\sigma}^{i}$, where $\hat{Q}_{\sigma}^{i}(\cdot\mid s^{i},x^{i})=Q_{\sigma}^{i}(\cdot\mid s^{i},x^{i})$
for all $(s^{i},x^{i})$ such that $\sigma^{i}(x^{i}\mid s^{i})>0$.
This definition is slightly more general than the typical one, e.g.,
\citet{dekel2004learning}, because it does not restrict players to
believe that consequences are driven by other players' strategies.}\bigskip{}

\begin{prop}
\label{prop:Berk-NashvsSCE}Suppose that the game is correctly specified
given $\sigma$, and that $\sigma$ is a Berk-Nash equilibrium. Then
$\sigma$ is also a self-confirming equilibrium.\footnote{A converse does not necessarily hold for a fixed game. The reason
is that the definition of SCE does not impose any restrictions on
off-equilibrium beliefs, while a particular subjective game may impose
ex-ante restrictions on beliefs. The following converse, however,
does hold: For any $\sigma$ that is an SCE, there exists a game that
is correctly specified for which $\sigma$ is a Berk-Nash equilibrium.}
\end{prop}
\begin{proof}
Fix any $i\in I$ and let $\hat{\theta}^{i}$ be in the support of
$\mu^{i}$, where $\mu^{i}$ is player $i$'s belief supporting the
Berk-Nash equilibrium strategy $\sigma^{i}$. Because the game is
correctly specified given $\sigma$, there exists $\theta_{*}^{i}\in\Theta^{i}$
such that $Q_{\theta_{*}^{i}}^{i}=Q_{\sigma}^{i}$ and, therefore,
by Lemma \ref{lemma:Theta-1}, $K^{i}(\sigma,\theta_{*}^{i})=0$.
Thus, it must also be that $K^{i}(\sigma,\hat{\theta}^{i})=0$. By
Lemma \ref{lemma:Theta-1}, it follows that $Q_{\hat{\theta}^{i}}^{i}(\cdot\mid s^{i},x^{i})=Q_{\sigma}^{i}(\cdot\mid s^{i},x^{i})$
for all $(s^{i},x^{i})$ such that $\sigma^{i}(x^{i}\mid s^{i})>0$.
In particular, $\sigma^{i}$ is optimal given $Q_{\hat{\theta}^{i}}^{i}$,
and $Q_{\hat{\theta}^{i}}^{i}$ satisfies the desired self-confirming
restriction.
\end{proof}
\bigskip{}

For games that are \emph{not} correctly specified, beliefs can be
incorrect on the equilibrium path, and so a Berk-Nash equilibrium
is not necessarily Nash or SCE.

\subsection{Relationship to fully cursed and ABEE}

An \textbf{analogy-based game} satisfies the following four properties:
(i) \emph{States and information structure}: The state space $\Omega$
is finite with distribution $p_{\Omega}\in\Delta(\Omega)$. In addition,
for each $i$, there is a partition $\mathcal{S}^{i}$ of $\Omega$,
and the element of $\mathcal{S}^{i}$ that contains $\omega$ (i.e.,
the signal of player $i$ in state $\omega$) is denoted by $s^{i}(\omega)$;\footnote{This assumption is made to facilitate comparison with Jehiel and Koessler's
(2008) ABEE.\nocite{jehiel2008revisiting}} (ii) \emph{Perfect feedback}: For each $i$, $f^{i}(x,\omega)=(x^{-i},\omega)$
for all $(x,\omega)$; (iii) \emph{Analogy partition}: For each $i$,
there exists a partition of $\Omega$, denoted by $\mathcal{A}^{i}$,
and the element of $\mathcal{A}^{i}$ that contains $\omega$ is denoted
by $\alpha^{i}(\omega)$; (iv) \emph{Conditional independence}: $(Q_{\theta^{i}}^{i})_{\theta^{i}\in\Theta^{i}}$
is the set of all joint probability distributions over $\mathbb{X}^{-i}\times\Omega$
that satisfy 
\[
Q_{\theta^{i}}^{i}\left(x^{-i},\omega\mid s^{i}(\omega'),x^{i}\right)=Q_{\Omega,\theta^{i}}^{i}(\omega\mid s^{i}(\omega'))Q_{\mathbb{X}^{-i},\theta^{i}}^{i}(x^{-i}\mid\alpha^{i}(\omega)).
\]
In other words, every player $i$ believes that $x^{-i}$ and $\omega$
are independent conditional on the analogy partition. For example,
if $\mathcal{A}^{i}=\mathcal{S}^{i}$ for all $i$, then each player
believes that the actions of other players are independent of the
state, conditional on their own private information.

\bigskip{}

\begin{defn}
(\citealp{jehiel2008revisiting}) A strategy profile $\sigma$ is
an analogy-based expectation equilibrium (ABEE) if for all $i\in I$,
$\omega\in\Omega$, and $x^{i}$ such that $\sigma^{i}(x^{i}\mid s^{i}(\omega))>0$,
$x^{i}\in\arg\max_{\bar{x}^{i}\in\mathbb{X}^{i}}\sum_{\omega'\in\Omega}p_{\Omega\mid S^{i}}(\omega'\mid s^{i}(\omega))\sum_{x^{-i}\in\mathbb{X}^{-i}}\bar{\sigma}^{-i}(x^{-i}\mid\omega')\pi^{i}(\bar{x}^{i},x^{-i},\omega')$,
where $\bar{\sigma}^{-i}(x^{-i}\mid\omega')=\sum_{\omega''\in\Omega}p_{\Omega\mid\mathcal{A}^{i}}(\omega''\mid\alpha^{i}(\omega'))\prod_{j\ne i}\sigma^{j}(x^{j}\mid s^{j}(\omega''))$.
\end{defn}
\bigskip{}

\begin{prop}
\label{prop:ABEE}In an analogy-based game, $\sigma$ is a Berk-Nash
equilibrium if and only if it is an ABEE.
\end{prop}
\begin{proof}
See the Appendix.
\end{proof}
\bigskip{}

As mentioned by \citet{jehiel2008revisiting}, ABEE is equivalent
to Eyster and Rabin's (2005) fully cursed equilibrium in the special
case where $\mathcal{A}^{i}=\mathcal{S}^{i}$ for all $i$. In particular,
Proposition \ref{prop:ABEE} provides a misspecified-learning foundation
for these two solution concepts. \citet{jehiel2008revisiting} discuss
an alternative foundation for ABEE, where players receive coarse feedback
aggregated over past play and multiple beliefs are consistent with
this feedback. Under this different feedback structure, ABEE can be
viewed as a natural selection of the set of SCE.\bigskip{}

\textbf{Example 2.5, continued from pg. \pageref{Trade}. }\textbf{\emph{Trade
with adverse selection.}}\textbf{ }In Online Appendix \ref{sec:OA_trade},
we show that $x^{*}$ is a Berk-Nash equilibrium price if and only
if $x=x^{*}$ maximizes an \textbf{equilibrium belief function} $\Pi(x,x^{*})$
which represents the belief about expected profit from choosing any
price $x$ under a steady-state $x^{*}$. The equilibrium belief function
depends on the feedback/misspecification assumptions, and we discuss
the following four cases:
\begin{eqnarray*}
\Pi^{NE}(x) & = & \Pr(A\leq x)\left(E\left[V\mid A\leq x\right]-x\right)\\
\Pi^{CE}(x) & = & \Pr(A\leq x)\left(E\left[V\right]-x\right)\\
\Pi^{BE}(x,x^{*}) & = & \Pr(A\leq x)\left(E\left[V\mid A\leq x^{*}\right]-x\right)\\
\Pi^{ABEE}(x) & = & \sum_{j=1}^{k}\Pr(V\in\mathbb{V}_{j})\left\{ \Pr(A\leq x\mid V\in\mathbb{V}_{j})\left(E\left[V\mid V\in\mathbb{V}_{j}\right]-x\right)\right\} .
\end{eqnarray*}

The first case, $\Pi^{NE}$, is the benchmark case in which beliefs
are correct. The second case, $\Pi^{CE}$, corresponds to perfect
feedback and subjective model $\Theta=\Delta(\mathbb{A})\times\Delta(\mathbb{V})$,
as described in page \pageref{Trade}. This is an example of an analogy-based
game with single analogy class $\mathbb{V}$. The buyer learns the
true marginal distributions of $A$ and $V$ and believes the joint
distribution equals the product of the marginal distributions. Berk-Nash
coincides with fully cursed equilibrium. The third case, $\Pi^{BE}$,
has the same misspecified model as the second case, but assumes partial
feedback, in the sense that the ask price $a$ is always observed
but the valuation $v$ is only observed if there is trade. The equilibrium
price $x^{*}$ affects the sample of valuations observed by the buyer
and, therefore, her beliefs. Berk-Nash coincides with naive behavioral
equilibrium.

The last case, $\Pi^{ABEE}$, corresponds to perfect feedback and
the following misspecification: Consider a partition of $\mathbb{V}$
into $k$ ``analogy classes'' $(\mathbb{V}_{j})_{j=1,...,k}$. The
buyer believes that $(A,V)$ are independent conditional on $V\in\mathbb{V}_{i}$,
for each $i=1,...,k$. The parameter set is $\Theta_{A}=\times_{j}\Delta(\mathbb{A})\times\Delta(\mathbb{V})$,
where, for a value $\theta=(\theta_{1},....,\theta_{k},\theta_{\mathbb{V}})\in\Theta_{A}$,
$\theta_{\mathbb{V}}$ parameterizes the marginal distribution over
$\mathbb{V}$ and, for each $j=1,...,k$, $\theta_{j}\in\Delta(\mathbb{A})$
parameterizes the distribution over $\mathbb{A}$ conditional on $V\in\mathbb{V}_{j}$.
Berk-Nash coincides with the ABEE of the game with analogy classes
$(\mathbb{V}_{j})_{j=1,...,k}$.\footnote{In Online Appendix \ref{sec:OA_trade}, we also consider the case
of ABEE with partial feedback.} $\square$

\section{\label{sec:foundation}Equilibrium foundation}

We provide a learning foundation for equilibrium. We follow \citet{fudenberg1993learning}
in considering games with (slightly) perturbed payoffs because, as
they highlight in the context of providing a learning foundation for
mixed-strategy Nash equilibrium, behavior need not be continuous in
beliefs without perturbations. Thus, even if beliefs were to converge,
behavior need not settle down in the unperturbed game. Perturbations
guarantee that if beliefs converge, then behavior also converges.

\subsection{\label{subsec:Perturbed-game}Perturbed game}

A \textbf{perturbation structure} is a tuple $\mathcal{P}=\left\langle \Xi,P_{\xi}\right\rangle $,
where: $\Xi=\times_{i\in I}\Xi^{i}$ and $\Xi^{i}\subseteq\mathbb{R}^{\#\mathbb{X}^{i}}$
is a set of payoff perturbations for each action of player $i$; $P_{\xi}=(P_{\xi^{i}})_{i\in I}$,
where $P_{\xi^{i}}\in\Delta(\Xi^{i})$ is a distribution over payoff
perturbations of player $i$ that is absolutely continuous with respect
to the Lebesgue measure, satisfies $ $$\int_{\Xi^{i}}||\xi^{i}||P_{\xi}(d\xi^{i})<\infty$,
and is independent from the perturbations of other players.%
{} A \textbf{perturbed game }$\mathcal{G}^{\mathcal{P}}=\left\langle \mathcal{G},\mathcal{P}\right\rangle $
is composed of a game $\mathcal{G}$ and a perturbation structure
$\mathcal{P}$. The timing of a perturbed game $\mathcal{G}^{\mathcal{P}}$
coincides with the timing of $\mathcal{G}$, except for two differences.
First, before taking an action, each player not only observes her
signal $s^{i}$ but also privately observes a vector of own-payoff
perturbations $\xi^{i}\in\Xi^{i}$, where $\xi^{i}(x^{i})$ denotes
the perturbation for action $x^{i}$. Second, her payoff given action
$x^{i}$ and consequence $y^{i}$ is $\pi^{i}(x^{i},y^{i})+\xi^{i}(x^{i})$.

A strategy\emph{ $\sigma^{i}$ }for player $i$ is \textbf{optimal
in the perturbed game} given\emph{ $\mu^{i}\in\Delta(\Theta^{i})$}
if, for all $(s^{i},x^{i})\in\mathbb{S}^{i}\times\mathbb{X}^{i}$,
$\sigma^{i}(x^{i}\mid s^{i})=P_{\xi}\left(\xi^{i}:x^{i}\in\Psi^{i}(\mu^{i},s^{i},\xi^{i})\right)$,
where 
\[
\Psi^{i}(\mu^{i},s^{i},\xi^{i})\equiv\arg\max_{x^{i}\in\mathbb{X}^{i}}E_{\bar{Q}_{\mu^{i}}^{i}(\cdot\mid s^{i},x^{i})}\left[\pi^{i}(x^{i},Y^{i})\right]+\xi^{i}(x^{i}).
\]
In other words, if $\sigma^{i}$ is an optimal strategy, then $\sigma^{i}(x^{i}\mid s^{i})$
is the probability that $x^{i}$ is optimal when the state is $s^{i}$
and the perturbation is $\xi^{i}$, taken over all possible realizations
of $\xi^{i}$. The definition of Berk-Nash equilibrium of a perturbed
game $\mathcal{G}^{\mathcal{P}}$ is analogous to Definition \ref{def:equilibrium-1},
with the only difference that optimality must be required with respect
to the \emph{perturbed} game.

\subsection{Learning foundation}

We fix a perturbed game $\mathcal{G}^{\mathcal{P}}$ and assume that
players repeatedly play the corresponding objective game at each $t=0,1,2,...$,
where the time-$t$ state and signals, $(\omega_{t},s_{t})$, and
perturbations $\xi_{t}$, are independently drawn every period from
the same distribution $p$ and $P_{\xi}$, respectively. In addition,
each player $i$ has a prior $\mu_{0}^{i}$ with full support over
her (finite-dimensional) parameter set, $\Theta^{i}$.\footnote{We restrict attention to parametric models (i.e., finite-dimensional
parameter spaces) because, otherwise, Bayesian updating need not converge
to the truth for most priors and parameter values even in correctly
specified statistical settings (\citet{freedman1963asymptotic}, \citet{diaconis1986consistency}).} At the end of every period $t$, each player uses Bayes' rule and
the information obtained in all past periods (her own signals, actions,
and consequences) to update beliefs. Players believe that they face
a stationary environment and myopically maximize the current period's
expected payoff.

Let $\Delta^{0}(\Theta^{i})$ denote the set of probability distributions
on $\Theta^{i}$ with full support. Let $B^{i}:\Delta^{0}(\Theta^{i})\times\mathbb{S}^{i}\times\mathbb{X}^{i}\times\mathbb{Y}^{i}\rightarrow\Delta^{0}(\Theta^{i})$
denote the Bayesian operator of player $i$: for all $A\subseteq\Theta$
Borel measurable and all $(\mu^{i},s^{i},x^{i},y^{i})\in\Delta^{0}(\Theta^{i})\times\mathbb{S}^{i}\times\mathbb{X}^{i}\times\mathbb{Y}^{i}$,
\[
B^{i}(\mu^{i},s^{i},x^{i},y^{i})(A)=\frac{\int_{A}Q_{\theta^{i}}^{i}(y^{i}\mid s^{i},x^{i})\mu^{i}(d\theta)}{\int_{\Theta}Q_{\theta^{i}}^{i}(y^{i}\mid s^{i},x^{i})\mu^{i}(d\theta)}.
\]
Bayesian updating is well defined by Assumption 1.\footnote{By Assumption 1(ii)-(iii), there exists $\theta\in\Theta$ and an
open ball containing it, such that $Q_{\theta'}^{i}>0$ for any $\theta'$
in the ball. Thus the Bayesian operator is well-defined for any $\mu^{i}\in\Delta^{0}(\Theta^{i})$.
Moreover, by Assumption 1(iii), such $\theta$'s are dense in $\Theta$,
so the Bayesian operator maps $\Delta^{0}(\Theta^{i})$ into itself.} Because players believe they face a stationary environment with i.i.d.
perturbations, it is without loss of generality to restrict player
$i$'s behavior at time $t$ to depend on $(\mu_{t}^{i},s_{t}^{i},\xi_{t}^{i})$.
\bigskip{}

\begin{defn}
A \textbf{policy} of player $i$ is a sequence of functions\emph{
$\phi^{i}=(\phi_{t}^{i})_{t}$, }where\emph{ $\phi_{t}^{i}:\Delta(\Theta^{i})\times\mathbb{S}^{i}\times\Xi^{i}\rightarrow\mathbb{X}^{i}$.
A }policy $\phi^{i}$ is \textbf{optimal} if \emph{$\phi_{t}^{i}\in\Psi^{i}$
}for all $t$. A policy profile \emph{$\phi=(\phi^{i}){}_{i\in I}$
is }\textbf{optimal}\emph{ }if $\phi^{i}$ is optimal for all $i\in I$.
\end{defn}
\bigskip{}

Let $\mathbb{H}\subseteq(\mathbb{S}\times\Xi\times\mathbb{X}\times\mathbb{Y})^{\infty}$
denote the set of histories, where any history $h=(s_{0},\xi_{0},x_{0},y_{0},...,s_{t},\xi_{t},x_{t},y_{t}...)\in\mathbb{H}$
satisfies the feasibility restriction: for all $i\in I$, $y_{t}^{i}=f^{i}(x_{t}^{i},x_{t}^{-i},\omega)$
for some $\omega\in supp(p_{\Omega\mid S^{i}}(\cdot\mid s_{t}^{i}))$
for all $t$. Let $\mathbf{P}{}^{\mu_{0},\phi}$ denote the probability
distribution over $\mathbb{H}$ that is induced by the priors $\mu_{0}=(\mu_{0}^{i})_{i\in I}$,
and the policy profiles $\phi=(\phi^{i})_{i\in I}$. Let $(\mu_{t})_{t}$
denote the sequence of beliefs $\mu_{t}:\mathbb{H}\rightarrow\times_{i\in I}\Delta(\Theta^{i})$
such that, for all $t\geq1$ and all $i\in I$, $\mu_{t}^{i}$ is
the posterior at time $t$ defined recursively by $\mu_{t}^{i}(h)=B^{i}(\mu_{t-1}^{i}(h),s_{t-1}^{i}(h),x_{t-1}^{i}(h),y_{t-1}^{i}(h))$
for all $h\in\mathbb{H}$, where $s_{t-1}^{i}(h)$ is player $i$'s
signal at $t-1$ given history $h$, and similarly for $x_{t-1}^{i}(h)$
and $y_{t-1}^{i}(h)$.

\bigskip{}

\begin{defn}
The \textbf{sequence of intended strategy profiles} given policy profile\emph{
}\textbf{\emph{$\phi=(\phi^{i})_{i\in I}$}} is the sequence $(\sigma_{t})_{t}$
of random variables $ $ $\sigma_{t}:\mathbb{H}\rightarrow\times_{i\in I}\Delta(\mathbb{X}^{i})^{\mathbb{S}^{i}}$
such that, for all $t$, all $i\in I$, and all $(x^{i},s^{i})\in\mathbb{X}^{i}\times\mathbb{S}^{i}$,
\begin{equation}
\sigma_{t}^{i}(h)(x^{i}\mid s^{i})=P_{\xi}\left(\xi^{i}:\phi_{t}^{i}(\mu_{t}^{i}(h),s^{i},\xi^{i})=x^{i}\right).\label{eq:intended}
\end{equation}
\end{defn}
\bigskip{}

An intended strategy profile $\sigma_{t}$ describes how each player
would behave at time $t$ for each possible signal; it is a random
variable because it depends on the players' beliefs at time $t$,
$\mu_{t}$, which in turn depend on the past history.

One reasonable criteria to claim that the players' behavior stabilizes
is that their intended behavior stabilizes with positive probability
(cf. \citealp{fudenberg1993learning}).\bigskip{}

\begin{defn}
\label{def:stability}A strategy profile $\sigma$ is \textbf{stable
{[}or strongly stable{]}} under policy profile\emph{ $\phi$} if the
sequence of intended strategies, $(\sigma_{t})_{t}$, converges to
$\sigma$ with positive probability {[}or with probability one{]},
i.e., 
\[
\mathbf{P}{}^{\mu_{0},\phi}\left(\lim_{t\rightarrow\infty}\left\Vert \sigma_{t}(h)-\sigma\right\Vert =0\right)>0\,\,\mbox{[or}=1\mbox{]}.
\]
\end{defn}
\bigskip{}

Lemma \ref{lem:Berk} says that, if behavior stabilizes to a strategy
profile $\sigma$, then, for each player $i$, beliefs become increasingly
concentrated on $\Theta^{i}(\sigma)$. This result extends findings
from the statistics of misspecified learning (\citet{berk1966limiting},
\citet{bunke1998asymptotic}) to a setting with active learning (i.e.,
players learn from data that is endogenously generated by their own
actions). Three new issues arise: (i) Previous results need to be
extended to the case of non-i.i.d. and endogenous data; (ii) It is
not obvious that steady-state beliefs can be characterized based on
steady-state behavior, independently of the path of play (Assumption
1 plays an important role here; See Section \ref{subsec:Discussion}
for an example); (iii) We allow the wKLD function to be nonfinite
so that players can believe that other players follow pure strategies.\footnote{For example, if player 1 believes that player 2 plays $A$ with probability
$\theta$ and $B$ with $1-\theta$, then the wKLD function is infinity
at $\theta=1$ if player 2 plays $B$ with positive probability.} \bigskip{}

\begin{lem}
\label{lem:Berk}Suppose that, for a policy profile $\phi$, the sequence
of intended strategies, $(\sigma_{t})_{t}$, converges to $\sigma$
for all histories in a set $\mathcal{H}\subseteq\mathbb{H}$ such
that $\mathbf{P}^{\mu_{0},\phi}\left(\mathcal{H}\right)>0$. Then,
for all open sets $U^{i}\supseteq\Theta^{i}(\sigma)$, $\lim_{t\rightarrow\infty}\mu_{t}^{i}\left(U^{i}\right)=1$,
a.s.-$\mathbf{P}^{\mu_{0},\phi}$ in $\mathcal{H}$. 
\end{lem}
\begin{proof}
See the Appendix.
\end{proof}
\medskip{}

The sketch of the proof of Lemma \ref{lem:Berk} is as follows (we
omit the $i$ subscript to ease the notational burden). Consider an
arbitrary $\epsilon>0$ and an open set $\Theta_{\epsilon}(\sigma)\subseteq\Theta$
defined as the points which are within $\epsilon$ distance of $\Theta(\sigma)$.
The time $t$ posterior over the complement of $\Theta_{\epsilon}(\sigma)$,
$\mu_{t}(\Theta\setminus\Theta_{\epsilon}(\sigma))$, can be expressed
as
\[
\frac{\int_{\Theta\setminus\Theta_{\epsilon}(\sigma)}\prod_{\tau=0}^{t-1}Q_{\theta}(y_{\tau}\mid s_{\tau},x_{\tau})\mu_{0}(d\theta)}{\int_{\Theta}\prod_{\tau=0}^{t-1}Q_{\theta}(y_{\tau}\mid s_{\tau},x_{\tau})\mu_{0}(d\theta)}=\frac{\int_{\Theta\setminus\Theta_{\epsilon}(\sigma)}e^{tK_{t}(\theta)}\mu_{0}(d\theta)}{\int_{\Theta}e^{tK_{t}(\theta)}\mu_{0}(d\theta)}
\]
where $K_{t}(\theta)$ equals minus the log-likelihood ratio, $-\frac{1}{t}\sum_{\tau=0}^{t-1}\ln\frac{Q_{\sigma_{\tau}}(y_{\tau}\mid s_{\tau},x_{\tau})}{Q_{\theta}(y_{\tau}\mid s_{\tau},x_{\tau})}$.
This expression and straightforward algebra implies that
\[
\mu_{t}(\Theta\setminus\Theta_{\epsilon}(\sigma))\leq\frac{\int_{\Theta\setminus\Theta_{\epsilon}(\sigma)}e^{t\left(K_{t}(\theta)+K(\sigma,\theta_{0})+\delta\right)}\mu_{0}(d\theta)}{\int_{\Theta_{\eta}(\sigma)}e^{t\left(K_{t}(\theta)+K(\sigma,\theta_{0})+\delta\right)}\mu_{0}(d\theta)}
\]
for any $\delta>0$ and $\theta_{0}\in\Theta(\sigma)$ and $\eta>0$
taken to be ``small''. Roughy speaking, the integral in the numerator
in the RHS is taken over points which are ``$\epsilon$-separated''
from $\Theta(\sigma)$, whereas the integral in the denominator is
taken over points which are ``$\eta$-close'' to $\Theta(\sigma)$.

Intuitively, if $K_{t}(\cdot)$ behaves asymptotically like $-K(\sigma,\cdot)$,
there exist sufficiently small $\delta>0$ and $\eta>0$ such that
$K_{t}(\theta)+K(\sigma,\theta_{0})+\delta$ is negative for all $\theta$
which are ``$\epsilon$-separated'' from $\Theta(\sigma)$, and
positive for all $\theta$ which are ``$\eta$-close'' to $\Theta(\sigma)$.
Thus, the numerator converges to zero, whereas the denominator diverges
to infinity, provided that $\Theta_{\eta}(\sigma)$ has positive measure
under the prior.

The nonstandard part of the proof consists of establishing that $\Theta_{\eta}(\sigma)$
has positive measure under the prior, which relies on Assumption 1,
and that indeed $K_{t}(\cdot)$ behaves asymptotically like $-K(\sigma,\cdot)$.
By virtue of Fatou's lemma, for $\theta\in\Theta_{\eta}(\sigma)$
it suffices to show almost sure pointwise convergence of $K_{t}(\theta)$
to $-K(\sigma,\theta)$; this is done in Claim B(i) in the Appendix
and relies on a LLN argument for non-iid variables. On the other hand,
over $\theta\in\Theta\setminus\Theta_{\epsilon}(\sigma)$, we need
to control the asymptotic behavior of $K_{t}(.)$ uniformly to be
able to interchange the limit and integral. In Claims B(ii) and B(iii)
in the Appendix, we establish that there exists $\alpha>0$ such that
asymptotically and over $\Theta\setminus\Theta_{\epsilon}(\sigma)$,
$K_{t}(\cdot)<-K(\sigma,\theta_{0})-\alpha$.

While Lemma \ref{lem:Berk} implies that the \emph{support} of posteriors
converges, posteriors need not converge. We can always find, however,
a subsequence of posteriors that converges. By continuity of behavior
in beliefs and the assumption that players are myopic, the stable
strategy profile must be statically optimal. Thus, we obtain the following
characterization of the set of stable strategy profiles when players
follow optimal policies.

\bigskip{}

\begin{thm}
\label{theo:Stability_implies_equilibrium}Suppose that a strategy
profile $\sigma$ is stable under an optimal policy profile for a
perturbed game. Then $\sigma$ is a Berk-Nash equilibrium of the perturbed
game.
\end{thm}
\begin{proof}
Let $\phi$ denote the optimal policy function under which $\sigma$
is stable. By Lemma \ref{lem:Berk}, there exists $\mathcal{H}\subseteq\mathbb{H}$
with $\mathbf{P}^{\mu_{0},\phi}\left(\mathcal{H}\right)>0$ such that,
for all $h\in\mathcal{H}$, $\lim_{t\rightarrow\infty}\sigma_{t}(h)=\sigma$
and $\lim_{t\rightarrow\infty}\mu_{t}^{i}\left(U^{i}\right)=1$ for
all $i\in I$ and all open sets $U^{i}\supseteq\Theta^{i}(\sigma)$;
for the remainder of the proof, fix any $h\in\mathcal{H}$. For all
$i\in I$, compactness of $\Delta(\Theta^{i})$ implies the existence
of a subsequence, which we denote as $(\mu_{t(j)}^{i})_{j}$, such
that $\mu_{t(j)}^{i}$ converges (weakly) to $\mu_{\infty}^{i}$ (the
limit could depend on $h$). We conclude by showing, for all $i\in I$: 

(i) $\mu_{\infty}^{i}\in\Delta(\Theta^{i}(\sigma))$: Suppose not,
so that there exists $\hat{\theta}^{i}\in supp(\mu_{\infty}^{i})$
such that $\hat{\theta}^{i}\notin\Theta^{i}(\sigma)$. Then, since
$\Theta^{i}(\sigma)$ is closed (by Lemma \ref{lemma:Theta-1}), there
exists an open set $U^{i}\supset\Theta^{i}(\sigma)$ with closure
$\bar{U}^{i}$ such that $\hat{\theta}^{i}\notin\bar{U}^{i}$. Then
$\mu_{\infty}^{i}(\bar{U}^{i})<1$, but this contradicts the fact
that $\mu_{\infty}^{i}\left(\bar{U}^{i}\right)\geq\limsup_{j\rightarrow\infty}\mu_{t(j)}^{i}\left(\bar{U}^{i}\right)\geq\lim_{j\rightarrow\infty}\mu_{t(j)}^{i}\left(U^{i}\right)=1$,
where the first inequality holds because $\bar{U}^{i}$ is closed
and $\mu_{t(j)}^{i}$ converges (weakly) to $\mu_{\infty}^{i}$.

(ii) $\sigma^{i}$ is optimal\emph{ }for the perturbed game given\emph{
$\mu_{\infty}^{i}\in\Delta(\Theta^{i})$}: 
\begin{align*}
\sigma^{i}(x^{i}\mid s^{i}) & =\lim_{j\rightarrow\infty}\sigma_{t(j)}^{i}(h)(x^{i}|s^{i})=\lim_{j\rightarrow\infty}P_{\xi}\left(\xi^{i}:x^{i}\in\Psi^{i}(\mu_{t(j)}^{i},s^{i},\xi^{i})\right)\\
 & =P_{\xi}\left(\xi^{i}:x^{i}\in\Psi^{i}(\mu_{\infty}^{i},s^{i},\xi^{i})\right),
\end{align*}
where the second equality follows because $\phi^{i}$ is optimal and
$\Psi^{i}$ is single-valued, a.s.- $P_{\xi^{i}}$,\footnote{$\Psi^{i}$ is single-valued a.s.-$\ensuremath{P_{\xi^{i}}}$ because
the set of $\xi^{i}$ such that $\#\Psi^{i}(\mu^{i},s^{i},\xi^{i})>1$
is of dimension lower than $\#\mathbb{X}^{i}$ and, by absolute continuity
of $P_{\xi^{i}}$, this set has measure zero.} and the third equality follows from a standard continuity argument.
\end{proof}

\subsection{A converse result}

Theorem \ref{theo:Stability_implies_equilibrium} provides our main
justification for Berk-Nash equilibria: any strategy profile that
is not an equilibrium cannot represent limiting behavior of optimizing
players. Theorem \ref{theo:Stability_implies_equilibrium}, however,
does not imply that behavior stabilizes. It is well known that convergence
is not guaranteed for Nash equilibrium, which is a special case of
Berk-Nash equilibrium.\footnote{\citet{jordan1993three} shows that non-convergence is robust to the
choice of initial conditions; \citet{benaim1999mixed} replicate this
finding for the perturbed version of Jordan's game. In the game-theory
literature, general global convergence results have only been obtained
in special classes of games---e.g. zero-sum, potential, and supermodular
games (\citealp{hofbauer2002global}).} Thus, some assumption needs to be relaxed to prove convergence for
general games. \citet{fudenberg1993learning} show that a converse
for the case of Nash equilibrium can be obtained by relaxing optimality
and allowing players to make vanishing optimization mistakes.\bigskip{}

\begin{defn}
A policy profile $\phi$ is \textbf{asymptotically optimal} if there
exists a positive real-valued sequence\emph{ $(\varepsilon_{t})_{t}$
}with\emph{ $\lim_{t\rightarrow\infty}\varepsilon_{t}=0$ }such that,
for all\emph{ $i\in I$, }all\emph{ }$(\mu^{i},s^{i},\xi^{i})\in\Delta(\Theta^{i})\times\mathbb{S}^{i}\times\Xi^{i}$\emph{,
}all\emph{ $t$,} and all $x^{i}\in\mathbb{X}^{i}$,
\end{defn}
\[
E_{\bar{Q}_{\mu^{i}}^{i}(\cdot\mid s^{i},x_{t}^{i})}\left[\pi^{i}(x_{t}^{i},Y^{i})\right]+\xi^{i}(x_{t}^{i})\geq E_{\bar{Q}_{\mu^{i}}^{i}(\cdot\mid s^{i},x^{i})}\left[\pi^{i}(x^{i},Y^{i})\right]+\xi^{i}(x^{i})-\varepsilon_{t}
\]
where $x_{t}^{i}=\phi_{t}^{i}(\mu^{i},s^{i},\xi^{i})$.\bigskip{}

Fudenberg and Kreps' (1993) insight is to suppose that players are
convinced early on that the equilibrium strategy is the right one
to play, and continue to play this strategy unless they have strong
enough evidence to think otherwise. And, as they continue to play
the equilibrium strategy, the evidence increasingly convinces them
that it is the right thing to do. This idea, however, need not work
for Berk-Nash equilibrium because beliefs may not converge if the
model is misspecified (see \citet{berk1966limiting} for an example).
If the game is weakly identified, however, Lemma \ref{lem:Berk} and
Fudenberg and Kreps' (1993) insight can be combined to obtain the
following converse of Theorem \ref{theo:Stability_implies_equilibrium}.\bigskip{}

\emph{Erratum (11/19/2019): We thank Yuichi Yamamoto for pointing
out that the statement of Theorem} \emph{\ref{Theo:converse-1} should
be corrected as follows:}\bigskip{}

\begin{thm}
\label{Theo:converse-1}\emph{Suppose that $\sigma$ is a Berk-Nash
equilibrium of a perturbed game that is weakly identified given $\sigma$
and let }$(\bar{\mu}^{i})_{i\in I}$ be a belief profile that supports
$\sigma$ as an equilibrium\emph{. Then, for any profile of priors
$\mu_{0}$ satisfying $\mu_{0}^{i}(\cdot|\Theta^{i}(\sigma))=\bar{\mu}^{i}$
for all $i\in I$ and for any $a>0$, there exists an asymptotically
optimal policy profile $\phi$ such that $P^{\mu_{0},\phi}\left(\lim_{t\rightarrow\infty}\left\Vert \sigma_{t}(h)-\sigma\right\Vert =0\right)>1-a$.}
\end{thm}
\begin{proof}
See Online Appendix \ref{sec:converse}.\footnote{The statement about the prior highlights that we are not choosing
the prior to be degenerate in a way that would make the result trivial.}
\end{proof}

\section{\label{subsec:Discussion}Discussion}

\emph{\indent$\textsc{importance of assumption 1. }$} The following
example illustrates that equilibrium may not exist and Lemma \ref{lem:Berk}
fails if Assumption 1 does not hold.\emph{ }A single agent chooses
action $x\in\{A,B\}$ and obtains outcome $y\in\{0,1\}$. The agent's
model is parameterized by $\theta=(\theta_{A},\theta_{B})$, where
$Q_{\theta}(y=1\mid A)=\theta_{A}$ and $Q_{\theta}(y=1\mid B)=\theta_{B}$.
The true model is $\theta^{0}=(1/4,3/4)$. The agent, however, is
misspecified and considers only $\theta_{1}=(0,3/4)$ and $\theta_{2}=(1/4,1/4)$
to be possible, i.e., $\Theta=\left\{ \theta_{1},\theta_{2}\right\} $.
In particular, Assumption 1(iii) fails for parameter value $\theta_{1}$.\footnote{Assumption 1(iii) would hold if, for some $\bar{\varepsilon}>0$,
$\theta=(\varepsilon,3/4)$ were also in $\Theta$ for all $0<\varepsilon\leq\bar{\varepsilon}$.} Suppose that $A$ is uniquely optimal for parameter value $\theta_{1}$
and $B$ is uniquely optimal for $\theta_{2}$ (further details about
payoffs are not needed).

A Berk-Nash equilibrium does not exist: If $A$ is played with positive
probability, then the wKLD is infinity at $\theta_{1}$ (i.e., $\theta_{1}$
cannot rationalize $y=1$ given $A$) and $\theta_{2}$ is the best
fit; but then $A$ is not optimal. If $B$ is played with probability
1, then $\theta_{1}$ is the best fit; but then $B$ is not optimal.
In addition, Lemma \ref{lem:Berk} fails: Suppose that the path of
play converges to pure strategy $B$. The best fit given $B$ is $\theta_{1}$,
but the posterior need not converge weakly to a degenerate probability
distribution on $\theta_{1}$; it is possible that, along the path
of play, the agent tried action $A$ and observed $y=1$, in which
case the posterior would immediately assign probability 1 to $\theta_{2}$.

\emph{$\textsc{forward-looking agents}$.} In the dynamic model, we
assumed that players are myopic. In Online Appendix \ref{sec:Non-myopic},
we extend Theorem \ref{theo:Stability_implies_equilibrium} to the
case of non-myopic players who solve a dynamic optimization problem
with beliefs as a state variable. A key fact used in the proof of
Theorem \ref{theo:Stability_implies_equilibrium} is that myopically
optimal behavior is continuous in beliefs. Non-myopic optimal behavior
is also continuous in beliefs, but the issue is that it may not coincide
with myopic behavior in the steady state if players still have incentives
to experiment. We prove the extension by requiring that the game is
weakly identified, which guarantees that players have no incentives
to experiment in steady state.

\emph{$\textsc{large population models.}$} The framework assumes
that there is a fixed number of players but, by focusing on stationary
subjective models, rules out aspects of ``repeated games'' where
players attempt to influence each others' play. In Online Appendix
\ref{sec:Population-models}, we adapt the equilibrium concept to
settings in which there is a population of a large number of agents
in the role of each player, so that agents have negligible incentives
to influence each other's play.

\emph{$\textsc{extensive-form games}$.} Our results hold for an alternative
timing where player $i$ commits to a signal-contingent plan of action
(i.e., a strategy) and observes both the realized signal $s^{i}$
and the consequence $y^{i}$ ex post. In particular, Berk-Nash equilibrium
is applicable to extensive-form games provided that players compete
by choosing contingent plan of actions and know the extensive form.
But the right approach is less clear if players have a misspecified
view of the extensive form (for example, they may not even know the
set of strategies available to them) or if players play the game sequentially
(for example, we would need to define and update beliefs at each information
set). The extension to extensive-form games is left for future work.\footnote{\citet{jehiel1995limited} considers the class of repeated alternating-move
games and assumes that players only forecast a limited number of time
periods into the future; see \citet{jehiel1998learning} for a learning
foundation. \citet{jehiel2007valuation} consider the general class
of extensive form games with perfect information and assume that players
simplify the game by partitioning the nodes into similarity classes.
In both cases, players are required to have correct beliefs, given
their limited or simplified view of the game.}

$\textsc{relationship to bounded rationality literature.}$ By providing
a language that makes the underlying misspecification explicit, we
offer some guidance for choosing between different models of bounded
rationality. For example, we could model the observed behavior of
an instructor in Example 2.3 by directly assuming that she believes
criticism improves performance and praise worsens it.\footnote{This assumption corresponds to a singleton set $\Theta$, thus fixing
beliefs at the outset and leaving no space for learning. This approach
is common in past work that assumes that agents have a misspecified
model but there is no learning about parameter values, e.g., \citet{barberis1998model}.} But extrapolating this observed belief to other contexts may lead
to erroneous conclusions. Instead, we postulate what we think is a
plausible misspecification (i.e., failure to account for regression
to the mean) and then derive beliefs endogenously, as a function of
the context.

We mentioned in the paper several instances of bounded rationality
that can be formalized via misspecified, endogenous learning. Other
examples in the literature can also be viewed as restricting beliefs
using the wKLD measure, but fall outside the scope of our paper either
because interactions are mediated by a price or because the problem
is dynamic (we focus on the repetition of a static problem). For example,
\citet{blume1982learning} and \citet{rabin2010gambler} explicitly
characterize beliefs using the limit of a likelihood function, while
\citet{bray1982learning}, \citet{radner82}, \citet{sargent1993bounded},
and \citet{Evans-book} focus specifically on OLS learning with misspecified
models. \citet{piccione2003modeling}, \citet{eyster2013approach},
and \citet{spiegler2013placebo} study pattern recognition in dynamic
settings and impose consistency requirements on beliefs that could
be interpreted as minimizing the wKLD measure. In the sampling equilibrium
of \citet{osborne1998games} and \citet{spiegler2006}, beliefs may
be incorrect due to learning from a limited sample, rather than from
misspecified learning. Other instances of bounded rationality that
do not seem naturally fitted to misspecified learning, include biases
in information processing due to computational complexity (e.g., \citet{rubinstein1986finite},
\citet{salant2011procedural}), bounded memory (e.g., \citealp{wilson2003bounded}),
self-deception (e.g., \citet{benabou2002self}, \citet{compte2004confidence}),
or sparsity-based optimization (\citet{gabaix2014sparsity}).

\newpage{}

\begin{spacing}{1}

\bibliographystyle{aer}
\bibliography{bibtex}

\end{spacing}

\appendix
\newpage{}

\part*{Appendix}

Let $\mathbb{Z}^{i}=\left\{ (s^{i},x^{i},y^{i})\in\mathbb{S}^{i}\times\mathbb{X}^{i}\times\mathbb{Y}^{i}:y^{i}=f^{i}(x^{i},x^{-i},\omega),x^{-i}\in\mathbb{X}^{-i},\omega\in supp(p_{\Omega\mid S^{i}}(\cdot\mid s^{i}))\right\} $.
For all $z^{i}=(s^{i},x^{i},y^{i})\in\mathbb{Z}^{i}$, define $\bar{P}_{\sigma}^{i}(z^{i})=Q_{\sigma}^{i}(y^{i}\mid s^{i},x^{i})\sigma^{i}(x^{i}\mid s^{i})p_{S^{i}}(s^{i})$.
We sometimes abuse notation and write $Q_{\sigma}^{i}(z^{i})\equiv Q_{\sigma}^{i}(y^{i}\mid s^{i},x^{i})$,
and similarly for $Q_{\theta^{i}}^{i}$. The following claim is used
in the proofs below.\smallskip{}

\textbf{Claim A.} \emph{For all $i\in I$}\textbf{\emph{: (i)}}\emph{
There exists $\theta_{*}^{i}\in\Theta^{i}$ and $\bar{K}<\infty$
such that, $\forall\sigma\in\Sigma$, $K^{i}(\sigma,\theta_{*}^{i})\leq\bar{K}$;
}\textbf{\emph{(ii)}}\emph{ Fix any $\theta^{i}\in\Theta^{i}$ and
$(\sigma_{n})_{n}$ such that $Q_{\theta^{i}}^{i}(z^{i})>0$ $\forall z^{i}\in\mathbb{Z}^{i}$
and $\lim_{n\rightarrow\infty}\sigma_{n}=\sigma$. Then $\lim_{n\rightarrow\infty}K^{i}(\sigma_{n},\theta^{i})=K^{i}(\sigma,\theta^{i})$;
}\textbf{\emph{(iii)}}\emph{ $K^{i}$ is (jointly) lower semicontinuous:
Fix any $(\theta_{n}^{i})_{n}$ and $(\sigma_{n})_{n}$ such that
$\lim_{n\rightarrow\infty}\theta_{n}^{i}=\theta^{i}$, $\lim_{n\rightarrow\infty}\sigma_{n}=\sigma$.
Then $\liminf_{n\rightarrow\infty}K^{i}(\sigma_{n},\theta_{n}^{i})\geq K(\sigma,\theta^{i})$;
}\textbf{\emph{(iv)}}\emph{ Let $\xi^{i}$ be a random vector in $\mathbb{R}^{\#\mathbb{X}^{i}}$
with absolutely continuous probability distribution $P_{\xi}$. Then,
$\forall(s^{i},x^{i})\in\mathbb{S}^{i}\times\mathbb{X}^{i}$, $\mu^{i}\mapsto\sigma^{i}(\mu^{i})(x^{i}\mid s^{i})=P_{\xi}\bigl(\xi^{i}:x^{i}\in\arg\max_{\bar{x}^{i}\in\mathbb{X}^{i}}E_{\bar{Q}_{\mu^{i}}^{i}(\cdot\mid s^{i},\bar{x}^{i})}\left[\pi^{i}(\bar{x}^{i},Y^{i})\right]+\xi^{i}(\bar{x}^{i})\bigr)$
is continuous.}

\smallskip{}

\emph{Proof}. (i) By Assumption 1 and finiteness of $\mathbb{Z}^{i}$,
there exist $\theta_{*}^{i}\in\Theta$ and $\alpha\in(0,1)$ such
that $Q_{\theta_{*}^{i}}^{i}(z^{i})\geq\alpha$ $\forall z^{i}\in\mathbb{Z}^{i}$.
Thus, $\forall\sigma\in\Sigma$, $K(\sigma,\theta_{*})\leq-E_{\bar{P}_{\sigma}^{i}}[\ln Q_{\theta_{*}}^{i}(Z^{i})]\leq-\ln\alpha$. 

(ii) $K^{i}(\sigma_{n},\theta^{i})-K(\sigma,\theta^{i})=\sum_{z^{i}\in\mathbb{Z}^{i}}\bigl\{\bigl(\bar{P}_{\sigma_{n}}^{i}(z^{i})\ln Q_{\sigma_{n}}^{i}(z^{i})-\bar{P}_{\sigma}^{i}(z^{i})\ln Q_{\sigma}^{i}(z^{i})\bigr)+\bigl(\bar{P}_{\sigma}^{i}(z^{i})-\bar{P}_{\sigma_{n}}^{i}(z^{i})\bigr)\ln Q_{\theta^{i}}^{i}(z^{i})\bigr\}$.
The first term in the RHS converges to 0 because $\lim_{n\rightarrow\infty}\sigma_{n}=\sigma$,
$Q_{\sigma}$ is continuous, and $x\ln x$ is continuous $\forall x\in[0,1]$.
The second term converges to 0 because $\lim_{n\rightarrow\infty}\sigma_{n}=\sigma$,
$\bar{P}_{\sigma}^{i}$ is continuous, and $\ln Q_{\theta^{i}}^{i}(z^{i})\in(-\infty,0]$
$\forall z^{i}\in\mathbb{Z}^{i}$. 

(iii) $K^{i}(\sigma_{n},\theta_{n}^{i})-K(\sigma,\theta^{i})=\sum_{z^{i}\in\mathbb{Z}^{i}}\bigl\{\bigl(\bar{P}_{\sigma_{n}}^{i}(z^{i})\ln Q_{\sigma_{n}}^{i}(z^{i})-\bar{P}_{\sigma}^{i}(z^{i})\ln Q_{\sigma}^{i}(z^{i})\bigr)+\bigl(\bar{P}_{\sigma}^{i}(z^{i})\ln Q_{\theta^{i}}^{i}(z^{i})-\bar{P}_{\sigma_{n}}^{i}(z^{i})\ln Q_{\theta_{n}^{i}}^{i}(z^{i})\bigr)\bigr\}$.
The first term in the RHS converges to 0 (same argument as in part
(ii)). The proof concludes by showing that, $\forall z^{i}\in\mathbb{Z}^{i}$,
\begin{equation}
\lim\inf_{n\rightarrow\infty}-\bar{P}_{\sigma_{n}}^{i}(z^{i})\ln Q_{\theta_{n}^{i}}^{i}(z^{i})\geq-\bar{P}_{\sigma}^{i}(z^{i})\ln Q_{\theta^{i}}^{i}(z^{i}).\label{eq:lsm}
\end{equation}
Suppose $\lim\inf_{n\rightarrow\infty}-\bar{P}_{\sigma_{n}}^{i}(z^{i})\ln Q_{\theta_{n}^{i}}^{i}(z^{i})\leq M<\infty$
(if not, (\ref{eq:lsm}) holds trivially). Then either (i) $\bar{P}_{\sigma_{n}}^{i}(z^{i})\rightarrow\bar{P}_{\sigma}^{i}(z^{i})>0$,
in which case (\ref{eq:lsm}) holds with equality (by continuity of
$\theta^{i}\mapsto Q_{\theta^{i}}^{i}$), or (ii) $\bar{P}_{\sigma_{n}}^{i}(z^{i})\rightarrow\bar{P}_{\sigma}^{i}(z^{i})=0$,
in which case (\ref{eq:lsm}) holds because its RHS is 0 (by convention
that $0\ln0=0$) and its LHS is always nonnegative. 

(iv) The proof is standard and, therefore, omitted. $\square$

\smallskip{}

\textbf{Proof of Lemma \ref{lemma:Theta-1}.} Part (i). Note that
\begin{align*}
K^{i}(\sigma,\theta^{i}) & \geq-\sum_{(s^{i},x^{i})\in\mathbb{S}^{i}\times\mathbb{X}^{i}}\ln\bigl(E_{Q_{\sigma}^{i}(\cdot\mid s^{i},x^{i})}\bigl[\frac{Q_{\theta^{i}}^{i}(Y^{i}\mid s^{i},x^{i})}{Q_{\sigma}^{i}(Y^{i}\mid s^{i},x^{i})}\bigr]\bigr)\sigma^{i}(x^{i}\mid s^{i})p_{S^{i}}(s^{i})=0,
\end{align*}
where the inequality follows from Jensen's inequality and the strict
concavity of $\ln(\cdot)$, and it holds with equality if and only
if $Q_{\theta^{i}}^{i}(\cdot\mid s^{i},x^{i})=Q_{\theta^{i}}^{i}(\cdot\mid s^{i},x^{i})$
$\forall(s^{i},x^{i})$ such that $\sigma^{i}(x^{i}\mid s^{i})>0$
(recall that, by assumption, $p_{S^{i}}(s^{i})>0$).

Part (ii). \emph{$\Theta^{i}(\sigma)$ is nonempty}: By Claim A(i),
$\exists K<\infty$ such that the minimizers are in the constraint
set $\{\theta^{i}\in\Theta^{i}:K^{i}(\sigma,\theta^{i})\leq K\}$.
Because $K^{i}(\sigma,\cdot)$ is continuous over a compact set, a
minimum exists.

\emph{$\Theta^{i}(\sigma)$ is upper hemicontinuous (uhc}): Fix any
$(\sigma_{n})_{n}$ and $(\theta_{n}^{i})_{n}$ such that $\lim_{n\rightarrow\infty}\sigma_{n}=\sigma$,
$\lim_{n\rightarrow\infty}\theta_{n}^{i}=\theta^{i}$, and $\theta_{n}^{i}\in\Theta^{i}(\sigma_{n})$
$\forall n$. We show that $\theta^{i}\in\Theta^{i}(\sigma)$ (so
that $\Theta^{i}(\cdot)$ has a closed graph and, by compactness of
$\Theta^{i}$, is uhc). Suppose, to obtain a contradiction, that $\theta^{i}\notin\Theta^{i}(\sigma)$.
By Claim A(i), there exist $\hat{\theta}^{i}\in\Theta^{i}$ and $\varepsilon>0$
such that $K^{i}(\sigma,\hat{\theta}^{i})\leq K^{i}(\sigma,\theta^{i})-3\varepsilon$
and $K^{i}(\sigma,\hat{\theta}^{i})<\infty$. By Assumption 1, $\exists(\hat{\theta}_{j}^{i})_{j}$
with $\lim_{j\rightarrow\infty}\hat{\theta}_{j}^{i}=\hat{\theta}^{i}$
and, $\forall j$, $Q_{\hat{\theta}_{j}^{i}}^{i}(z^{i})>0$ $\forall z^{i}\in\mathbb{Z}^{i}$.
We show that there is an element of the sequence, $\hat{\theta}_{J}^{i}$,
that ``does better'' than $\theta_{n}^{i}$ given $\sigma_{n}$,
which is a contradiction. Because $K^{i}(\sigma,\hat{\theta}^{i})<\infty$,
continuity of $K^{i}(\sigma,\cdot)$ implies that there exists $J$
large enough such that $\left|K^{i}(\sigma,\hat{\theta}_{J}^{i})-K^{i}(\sigma,\hat{\theta}^{i})\right|\leq\varepsilon/2$.
Moreover, Claim A(ii) applied to $\theta^{i}=\hat{\theta}_{J}^{i}$
implies that there exists $N_{\varepsilon,J}$ such that, $\forall n\geq N_{\varepsilon,J}$,
$\left|K^{i}(\sigma_{n},\hat{\theta}_{J}^{i})-K^{i}(\sigma,\hat{\theta}_{J}^{i})\right|\leq\varepsilon/2$.
Thus, $\forall n\geq N_{\varepsilon,J}$, $\bigl|K^{i}(\sigma_{n},\hat{\theta}_{J}^{i})-K^{i}(\sigma,\hat{\theta}^{i})\bigr|\leq\bigl|K^{i}(\sigma_{n},\hat{\theta}_{J}^{i})-K^{i}(\sigma,\hat{\theta}_{J}^{i})\bigr|+\bigl|K^{i}(\sigma,\hat{\theta}_{J}^{i})-K^{i}(\sigma,\hat{\theta}^{i})\bigr|\leq\varepsilon$.
Therefore, 
\begin{equation}
K^{i}(\sigma_{n},\hat{\theta}_{J}^{i})\leq K^{i}(\sigma,\hat{\theta}^{i})+\varepsilon\leq K^{i}(\sigma,\theta^{i})-2\varepsilon.\label{eq:uhc_1-1}
\end{equation}
Suppose $K^{i}(\sigma,\theta^{i})<\infty$. By Claim A(iii), $\exists n_{\varepsilon}\geq N_{\varepsilon,J}$
such that $K^{i}(\sigma_{n_{\varepsilon}},\theta_{n_{\varepsilon}}^{i})\geq K^{i}(\sigma,\theta^{i})-\varepsilon$.
This result and expression (\ref{eq:uhc_1-1}), imply $K^{i}(\sigma_{n_{\varepsilon}},\hat{\theta}_{J}^{i})\leq K(\sigma_{n_{\varepsilon}},\theta_{n_{\varepsilon}}^{i})-\varepsilon$.
But this contradicts $\theta_{n_{\varepsilon}}^{i}\in\Theta^{i}(\sigma_{n_{\varepsilon}})$.
Finally, if $K^{i}(\sigma,\theta^{i})=\infty$, Claim A(iii) implies
that $\exists n_{\varepsilon}\geq N_{\varepsilon,J}$ such that $K^{i}(\sigma_{n_{\varepsilon}},\theta_{n_{\varepsilon}}^{i})\geq2K$,
where $K$ is the bound defined in Claim A(i). But this also contradicts
$\theta_{n_{\varepsilon}}^{i}\in\Theta^{i}(\sigma_{n_{\varepsilon}})$.

\emph{$\Theta^{i}(\sigma)$ is compact}: As shown above, $\Theta^{i}(\cdot)$
has a closed graph, and so $\Theta^{i}(\sigma)$ is a closed set.
Compactness of $\Theta^{i}(\sigma)$ follows from compactness of $\Theta^{i}$.
$\square$

\smallskip{}

\textbf{Proof of Theorem \ref{theo:Existence}. }We prove the result
in two parts. \emph{Part 1}. %
{} We show existence of equilibrium in the perturbed game (defined in
Section \ref{subsec:Perturbed-game}). Let $\Gamma:\times_{i\in I}\Delta(\Theta^{i})\rightrightarrows\times_{i\in I}\Delta(\Theta^{i})$
be a correspondence such that, $\forall\mu=(\mu^{i})_{i\in I}\in\times_{i\in I}\Delta(\Theta^{i})$,
$\Gamma(\mu)=\times_{i\in I}\Delta\left(\Theta^{i}(\sigma(\mu))\right),$
where $\sigma(\mu)=(\sigma^{i}(\mu^{i}))_{i\in I}\in\Sigma$ and is
defined as
\begin{equation}
\sigma^{i}(\mu^{i})(x^{i}|s^{i})=P_{\xi}\left(\xi^{i}:x^{i}\in\arg\max_{\bar{x}^{i}\in\mathbb{X}^{i}}E_{\bar{Q}_{\mu^{i}}^{i}(\cdot\mid s^{i},\bar{x}^{i})}\left[\pi^{i}(\bar{x}^{i},Y^{i})\right]+\xi^{i}(\bar{x}^{i})\right)\label{eq:sigma(mu)_perturbed}
\end{equation}
 $\forall(x^{i},s^{i})\in\mathbb{X}^{i}\times\mathbb{S}^{i}$. Note
that if $\exists\mu_{\ast}\in\times_{i\in I}\Delta(\Theta^{i})$ such
that $\mu_{\ast}\in\Gamma(\mu_{\ast})$, then $\sigma_{*}\equiv\left(\sigma^{i}(\mu_{*}^{i})\right)_{i\in I}$
is an equilibrium of the perturbed game. We show that such $\mu_{\ast}$
exists by checking the conditions of the Kakutani-Fan-Glicksberg fixed-point
theorem: (i) $\times_{i\in I}\Delta(\Theta^{i})$ is compact, convex
and locally convex Hausdorff: The set $\Delta(\Theta^{i})$ is convex,
and since $\Theta^{i}$ is compact $\Delta(\Theta^{i})$ is also compact
under the weak topology (\citet{aliprantis2006infinite}, Theorem
15.11). By Tychonoff's theorem, $\times_{i\in I}\Delta(\Theta^{i})$
is compact too. Finally, the set is also locally convex under the
weak topology;\footnote{This last claim follows since the weak topology is induced by a family
of semi-norms of the form: $\rho(\mu,\mu')=|E_{\mu}[f]-E_{\mu'}[f]|$
for $f$ continuous and bounded for any $\mu$ and $\mu'$ in $\Delta(\Theta^{i})$.} (ii) $\Gamma$ has convex, nonempty images: It is clear that $\Delta\left(\Theta^{i}(\sigma(\mu))\right)$
is convex valued $\forall\mu$. Also, by Lemma \ref{lemma:Theta-1},
$\Theta^{i}(\sigma(\mu))$ is nonempty $\forall\mu$; (iii) $\Gamma$
has a closed graph:\emph{ }Let $(\mu_{n},\hat{\mu}_{n})_{n}$ be such
that $\hat{\mu}_{n}\in\Gamma(\mu_{n})$ and $\mu_{n}\rightarrow\mu$
and $\hat{\mu}_{n}\rightarrow\hat{\mu}$ (under the weak topology).
By Claim A(iv), $\mu^{i}\mapsto\sigma^{i}(\mu^{i})$ is continuous.
Thus, $\sigma_{n}\equiv\left(\sigma^{i}(\mu_{n}^{i})\right)_{i\in I}\rightarrow\sigma\equiv\left(\sigma^{i}(\mu^{i})\right)_{i\in I}$.
By Lemma \ref{lemma:Theta-1}, $\sigma\mapsto\Theta^{i}(\sigma)$
is uhc; thus, by Theorem 17.13 in \citet{aliprantis2006infinite},
$\sigma\mapsto\times_{i\in I}\Delta\left(\Theta^{i}(\sigma)\right)$
is also uhc. Therefore, $\hat{\mu}\in\times_{i\in I}\Delta\left(\Theta^{i}(\sigma)\right)=\Gamma(\mu)$. 

\emph{Part 2}. Fix a sequence of perturbed games indexed by the probability
of perturbations $(P_{\xi,n})_{n}$. By Part 1, there is a corresponding
sequence of fixed points $(\mu_{n})_{n}$, such that $\mu_{n}\in\times_{i\in I}\Delta\left(\Theta^{i}(\sigma_{n})\right)$
$\forall n$, where $\sigma_{n}\equiv(\sigma^{i}(\mu_{n}^{i},P_{\xi,n}))_{i\in I}$
(see equation (\ref{eq:sigma(mu)_perturbed}), where we now explicitly
account for dependance on $P_{\xi,n}$). By compactness, there exist
subsequences of $(\mu_{n})_{n}$ and $(\sigma_{n})_{n}$ that converge
to $\mu$ and $\sigma$, respectively. Since $\sigma\mapsto\times_{i\in I}\Delta\left(\Theta^{i}(\sigma)\right)$
is uhc, then $\mu\in\times_{i\in I}\Delta\left(\Theta^{i}(\sigma)\right)$.
We now show that if we choose $(P_{\xi,n})_{n}$ such that, $\forall\varepsilon>0$,
$\lim_{n\rightarrow\infty}P_{\xi,n}\left(\left\Vert \xi_{n}^{i}\right\Vert \geq\varepsilon\right)=0$,
then $\sigma$ is optimal given $\mu$ in the unperturbed game---this
establishes existence of equilibrium in the unperturbed game. Suppose
not, so that there exist $i,s^{i},x^{i},\hat{x}^{i},$ and $\varepsilon>0$
such that $\sigma^{i}(x^{i}\mid s^{i})>0$ and $E_{\bar{Q}_{\mu^{i}}^{i}(\cdot\mid s^{i},x^{i})}\left[\pi^{i}(x^{i},Y^{i})\right]+4\varepsilon\leq E_{\bar{Q}_{\mu^{i}}^{i}(\cdot\mid s^{i},\hat{x}^{i})}\left[\pi^{i}(\hat{x}^{i},Y^{i})\right]$.
By continuity of $\mu^{i}\mapsto\bar{Q}_{\mu^{i}}^{i}$ and the fact
that $\lim_{n\rightarrow\infty}\mu_{n}^{i}=\mu^{i}$, $\exists n_{1}$
such that, $\forall n\geq n_{1}$, $E_{\bar{Q}_{\mu_{n}^{i}}^{i}(\cdot\mid s^{i},x^{i})}\left[\pi^{i}(x^{i},Y^{i})\right]+2\varepsilon\leq E_{\bar{Q}_{\mu_{n}^{i}}^{i}(\cdot\mid s^{i},\hat{x}^{i})}\left[\pi^{i}(\hat{x}^{i},Y^{i})\right]$.
It then follows from (\ref{eq:sigma(mu)_perturbed}) and $\lim_{n\rightarrow\infty}P_{\xi,n}\left(\left\Vert \xi_{n}^{i}\right\Vert \geq\varepsilon\right)=0$
that $\lim_{n\rightarrow\infty}\sigma^{i}(\mu_{n}^{i},P_{\xi,n})(x^{i}|s^{i})=0$.
But this contradicts $\lim_{n\rightarrow\infty}\sigma^{i}(\mu_{n}^{i},P_{\xi,n})(x^{i}|s^{i})=\sigma^{i}(x^{i}\mid s^{i})>0$.
$\square$

\smallskip{}

\textbf{Proof of Proposition \ref{prop:ABEE}.} In the next paragraph,
we prove the following result: For all $\sigma$ and $\bar{\theta}_{\sigma}^{i}\in\Theta^{i}(\sigma)$,
(a) $Q_{\Omega,\bar{\theta}_{\sigma}^{i}}^{i}(\omega'\mid s^{i})=p_{\Omega\mid S^{i}}(\omega'\mid s^{i})$
for all $s^{i}\in\mathcal{S}^{i},\omega'\in\Omega$ and (b) $Q_{\mathbb{X}^{-i},\bar{\theta}_{\sigma}^{i}}^{i}(x^{-i}\mid\alpha^{i})=\sum_{\omega''\in\Omega}p_{\Omega\mid\mathcal{A}^{i}}(\omega''\mid\alpha^{i})\prod_{j\ne i}\sigma^{j}(x^{j}\mid s^{j}(\omega''))$
for all $\alpha^{i}\in\mathcal{A}^{i},x^{-i}\in\mathbb{X}^{-i}$.
Equivalence between Berk-Nash and ABEE follows immediately from (a),
(b), and the fact that expected utility of player $i$ with signal
$s^{i}$ and beliefs $\bar{\theta}_{\sigma}$ equals $\sum_{\omega'\in\Omega}Q_{\Omega,\bar{\theta}_{\sigma}^{i}}^{i}(\omega'\mid s^{i})\sum_{x^{-i}\in\mathbb{X}^{-i}}Q_{\mathbb{X}^{-i},\bar{\theta}_{\sigma}^{i}}^{i}(x^{-i}\mid\alpha^{i}(\omega'))\pi^{i}(\bar{x}^{i},x^{-i},\omega')$.

Proof of (a) and (b): $-K^{i}(\sigma,\theta^{i})$ equals, up to a
constant, 
\begin{align*}
 & \sum_{s^{i},\tilde{w},\tilde{x}^{-i}}\ln\left(Q_{\Omega,\theta^{i}}^{i}(\tilde{\omega}\mid s^{i})Q_{\mathbb{X}^{-i},\theta^{i}}^{i}(\tilde{x}^{-i}\mid\alpha^{i}(\tilde{\omega}))\right)\prod_{j\ne i}\sigma^{j}(\tilde{x}^{j}\mid s^{j}(\tilde{\omega}))p_{\Omega\mid S^{i}}(\tilde{\omega}\mid s^{i})p_{S^{i}}(s^{i})\\
= & \sum_{s^{i},\tilde{\omega}}\ln\left(Q_{\Omega,\theta^{i}}^{i}(\tilde{\omega}\mid s^{i})\right)p_{\Omega\mid S^{i}}(\tilde{\omega}\mid s^{i})p_{S^{i}}(s^{i})+\sum_{\tilde{x}^{-i},\alpha^{i}\in\mathcal{A}^{i}}\ln\left(Q_{\mathbb{X}^{-i},\theta^{i}}^{i}(\tilde{x}^{-i}\mid\alpha^{i})\right)\sum_{\tilde{\omega}\in\alpha^{i}}\prod_{j\ne i}\sigma^{j}(\tilde{x}^{j}\mid s^{j}(\tilde{\omega}))p_{\Omega}(\tilde{\omega}).
\end{align*}
It is straightforward to check that any parameter value that maximizes
the above expression satisfies (a) and (b). $\square$

\smallskip{}

\textbf{Proof of Lemma \ref{lem:Berk}.} The proof uses Claim B, which
is stated and proven after this proof. It is sufficient to establish
that $\lim_{t\rightarrow\infty}\int_{\Theta^{i}}d^{i}(\sigma,\theta^{i})\mu_{t}^{i}(d\theta^{i})=0$
a.s. in $\mathcal{H}$, where $d^{i}(\sigma,\theta^{i})=\inf_{\hat{\theta}^{i}\in\Theta^{i}(\sigma)}\|\theta^{i}-\hat{\theta}^{i}\|$.
Fix $i\in I$ and $h\in\mathbb{H}$. Then, by Bayes' rule,
\begin{align*}
\int_{\Theta^{i}}d^{i}(\sigma,\theta^{i})\mu_{t}^{i}(d\theta^{i}) & =\frac{\int_{\Theta^{i}}d^{i}(\sigma,\theta^{i})\prod_{\tau=0}^{t-1}\frac{Q_{\theta^{i}}^{i}(y_{\tau}^{i}\mid s_{\tau}^{i},x_{\tau}^{i})}{Q_{\sigma_{\tau}}^{i}(y_{\tau}^{i}\mid s_{\tau}^{i},x_{\tau}^{i})}\mu_{0}^{i}(d\theta^{i})}{\int_{\Theta^{i}}\prod_{\tau=0}^{t-1}\frac{Q_{\theta^{i}}^{i}(y_{\tau}^{i}\mid s_{\tau}^{i},x_{\tau}^{i})}{Q_{\sigma_{\tau}}^{i}(y_{\tau}^{i}\mid s_{\tau}^{i},x_{\tau}^{i})}\mu_{0}^{i}(d\theta^{i})}=\frac{\int_{\Theta^{i}}d^{i}(\sigma,\theta^{i})e^{tK_{t}^{i}(h,\theta^{i})}\mu_{0}^{i}(d\theta^{i})}{\int_{\Theta^{i}}e^{tK_{t}^{i}(h,\theta^{i})}\mu_{0}^{i}(d\theta^{i})},
\end{align*}
where the first equality is well-defined by Assumption 1, full support
of $\mu_{0}^{i}$, and the fact that $\mathbf{P}^{\mu_{0},\phi}\left(\mathcal{H}\right)>0$
implies that all the $Q_{\sigma_{\tau}}^{i}$terms are positive, and
where we define $K_{t}^{i}(h,\theta^{i})=-\frac{1}{t}\sum_{\tau=0}^{t-1}\ln\frac{Q_{\sigma_{\tau}}^{i}(y_{\tau}^{i}\mid s_{\tau}^{i},x_{\tau}^{i})}{Q_{\theta^{i}}^{i}(y_{\tau}^{i}\mid s_{\tau}^{i},x_{\tau}^{i})}$
for the second equality.\footnote{If, for some $\theta^{i}$, $Q_{\theta^{i}}^{i}(y_{\tau}^{i}\mid s_{\tau}^{i},x_{\tau}^{i})=0$
for some $\tau\in\{0,...,t-1\},$ then we define $K_{t}^{i}(h,\theta^{i})=-\infty$
and $\exp\left\{ tK_{t}^{i}(h,\theta^{i})\right\} =0$. } For any $\alpha>0$, define $\Theta_{\alpha}^{i}(\sigma)\equiv\left\{ \theta^{i}\in\Theta^{i}:d^{i}(\sigma,\theta^{i})<\alpha\right\} $.
Then, for all $\varepsilon>0$ and $\eta>0$, $\int_{\Theta^{i}}d^{i}(\sigma,\theta^{i})\mu_{t}^{i}(d\theta^{i})\leq\varepsilon+C\frac{A_{t}^{i}(h,\sigma,\varepsilon)}{B_{t}^{i}(h,\sigma,\eta)}$,
where $C\equiv\sup_{\theta_{1}^{i},\theta_{2}^{i}\in\Theta^{i}}\left\Vert \theta_{1}^{i}-\theta_{2}^{i}\right\Vert <\infty$
(because $\Theta^{i}$ is bounded) and where $A_{t}^{i}(h,\sigma,\varepsilon)=\int_{\Theta^{i}\backslash\Theta_{\varepsilon}^{i}(\sigma)}e^{tK_{t}^{i}(h,\theta^{i})}\mu_{0}^{i}(d\theta^{i})$
and $B_{t}^{i}(h,\sigma,\eta)=\int_{\Theta_{\eta}^{i}(\sigma)}e^{tK_{t}^{i}(h,\theta^{i})}\mu_{0}^{i}(d\theta^{i})$.The
proof concludes by showing that, for all (sufficiently small) $\varepsilon>0$,
$\exists\eta_{\varepsilon}>0$ such that $\lim_{t\rightarrow\infty}A_{t}^{i}(h,\sigma,\varepsilon)/B_{t}^{i}(h,\sigma,\eta_{\varepsilon})=0$.
This result is achieved in several steps. First, $\forall\varepsilon>0$,
define $K_{\varepsilon}^{i}(\sigma)=\inf\left\{ K^{i}(\sigma,\theta^{i})\mid\theta^{i}\in\Theta^{i}\backslash\Theta_{\varepsilon}^{i}(\sigma)\right\} $
and $\alpha_{\varepsilon}=\left(K_{\varepsilon}^{i}(\sigma)-K_{0}^{i}(\sigma)\right)/3$,
where $K_{0}^{i}(\sigma)=\inf_{\theta^{i}\in\Theta^{i}}K^{i}(\sigma,\theta^{i})$.
By continuity of $K^{i}(\sigma,\cdot)$, there exist $\bar{\varepsilon}$
and $\bar{\alpha}$ such that, $\forall\varepsilon\leq\bar{\varepsilon}$,
$0<\alpha_{\varepsilon}\leq\bar{\alpha}<\infty$. Henceforth, let
$\varepsilon\leq\bar{\varepsilon}$. It follows that 
\begin{equation}
K^{i}(\sigma,\theta^{i})\geq K_{\varepsilon}^{i}(\sigma)>K_{0}^{i}(\sigma)+2\alpha_{\varepsilon}\label{eq:K_theta>K_0+2alpha_v2}
\end{equation}
 $\forall\theta^{i}$ such that $d^{i}(\sigma,\theta^{i})\geq\varepsilon$.
Also, by continuity of $K^{i}(\sigma,\cdot)$, $\exists\eta_{\varepsilon}>0$
such that, $\forall\theta^{i}\in\Theta_{\eta_{\varepsilon}}^{i}$,
\begin{equation}
K^{i}(\sigma,\theta^{i})<K_{0}^{i}(\sigma)+\alpha_{\varepsilon}/2.\label{eq:K+theta<K_0+alpha/2}
\end{equation}

Second, let $\hat{\Theta}^{i}=\left\{ \theta^{i}\in\Theta^{i}:Q_{\theta^{i}}^{i}(y_{\tau}^{i}\mid s_{\tau}^{i},x_{\tau}^{i})>0\,\,\mbox{for all}\,\,\mbox{\ensuremath{\tau}}\right\} $
and $\hat{\Theta}_{\eta_{\varepsilon}}^{i}(\sigma)=\hat{\Theta}^{i}\cap\Theta_{\eta_{\varepsilon}}^{i}(\sigma)$.
We now show that $\mu_{0}^{i}(\hat{\Theta}_{\eta_{\varepsilon}}^{i}(\sigma))>0$.
By Lemma \ref{lemma:Theta-1}, $\Theta^{i}(\sigma)$ is nonempty.
Pick any $\theta^{i}\in\Theta^{i}(\sigma)$. By Assumption 1, $\exists(\theta_{n}^{i})_{n}$
in $\Theta^{i}$ such that $\lim_{n\rightarrow\infty}\theta_{n}^{i}=\theta^{i}$
and $Q_{\theta_{n}^{i}}^{i}(y^{i}\mid s^{i},x^{i})>0$ $\forall y^{i}\in f^{i}(\Omega,x^{i},\mathbb{X}^{-i})$
and all $(s^{i},x^{i})\in\mathbb{S}^{i}\times\mathbb{X}^{i}$. In
particular, $\exists\theta_{\bar{n}}^{i}$ such that $d^{i}(\sigma,\theta_{\bar{n}}^{i})<.5\eta_{\varepsilon}$
and, by continuity of $Q_{\cdot}$, there exists an open set $U$
around $\theta_{\bar{n}}^{i}$ such that $U\subseteq\hat{\Theta}_{\eta_{\varepsilon}}^{i}(\sigma)$.
By full support, $\mu_{0}^{i}(\hat{\Theta}_{\eta_{\varepsilon}}^{i}(\sigma))>0$.%
{} Next, note that, 
\begin{align}
\liminf_{t\rightarrow\infty}B_{t}^{i}(h,\sigma,\eta_{\varepsilon})e^{t(K_{0}^{i}(\sigma)+\frac{\alpha_{\varepsilon}}{2})} & \geq\liminf_{t\rightarrow\infty}\int_{\hat{\Theta}_{\eta_{\varepsilon}}^{i}(\sigma)}e^{t(K_{0}^{i}(\sigma)+\frac{\alpha_{\varepsilon}}{2}+K_{t}^{i}(h,\theta^{i}))}\mu_{0}^{i}(d\theta^{i})\,\,\,\,\,\,\,\,\,\,\,\nonumber \\
 & \geq\int_{\hat{\Theta}_{\eta_{\varepsilon}}^{i}(\sigma)}e^{\lim_{t\rightarrow\infty}t(K_{0}^{i}(\sigma)+\frac{\alpha_{\varepsilon}}{2}-K^{i}(\sigma,\theta^{i}))}\mu_{0}^{i}(d\theta^{i})=\infty\label{eq:Bt}
\end{align}
a.s. in $\mathcal{H}$, $ $where the first inequality follows because
$\hat{\Theta}_{\eta_{\varepsilon}}^{i}(\sigma)\subseteq\Theta_{\eta_{\varepsilon}}^{i}(\sigma)$
and $\exp$ is a positive function, the second inequality follows
from Fatou's Lemma and a LLN for non-iid random variables that implies
$\lim_{t\rightarrow\infty}K_{t}^{i}(h,\theta^{i})=-K^{i}(\sigma,\theta^{i})$
$\forall\theta^{i}\in\hat{\Theta}^{i}$, a.s. in $\mathcal{H}$ (see
Claim B(i) below), and the last equality follows from (\ref{eq:K+theta<K_0+alpha/2})
and the fact that $\mu_{0}^{i}(\Theta_{\eta_{\varepsilon}}^{i}(\sigma))>0$.

Next, we consider the term $A_{t}^{i}(h,\sigma,\varepsilon)$. Claims
B(ii) and B(iii) (see below) imply that $\exists T$ such that, $\forall t\geq T$,
$K_{t}^{i}(h,\theta^{i})<-(K_{0}^{i}(\sigma)+(3/2)\alpha_{\varepsilon})$
$\forall\theta^{i}\in\Theta^{i}\backslash\Theta_{\varepsilon}^{i}(\sigma)$,
a.s. in $\mathcal{H}$. Thus, 
\begin{align*}
\lim_{t\rightarrow\infty}A_{t}^{i}(h,\sigma,\varepsilon)e^{t(K_{0}^{i}(\sigma)+\alpha_{\varepsilon})} & =\lim_{t\rightarrow\infty}\int_{\Theta^{i}\backslash\Theta_{\varepsilon}^{i}(\sigma)}e^{t(K_{0}^{i}(\sigma)+\alpha_{\varepsilon}+K_{t}^{i}(h,\theta^{i}))}\mu_{0}^{i}(d\theta^{i})\\
 & \leq\mu_{0}^{i}(\Theta^{i}\backslash\Theta_{\varepsilon}^{i}(\sigma))\lim_{t\rightarrow\infty}e^{-t\alpha_{\varepsilon}/2}=0.
\end{align*}
a.s. in $\mathcal{H}$. The above expression and equation (\ref{eq:Bt})
imply that $\lim_{t\rightarrow\infty}A_{t}^{i}(h,\sigma,\varepsilon)/B_{t}^{i}(h,\sigma,\eta_{\varepsilon})=0$
a.s.-$\mathbf{P}^{\mu_{0},\phi}$. $\square$

\smallskip{}

We state and prove Claim B used in the proof above. For any $\xi>0$,
define $\Theta_{\sigma,\xi}^{i}$ to be the set such that $\theta^{i}\in\Theta_{\sigma,\xi}^{i}$
if and only if $Q_{\theta^{i}}^{i}(y^{i}\mid s^{i},x^{i})\geq\xi$
for all $(s^{i},x^{i},y^{i})$ such that $Q_{\sigma}^{i}(y^{i}\mid s^{i},x^{i})\sigma^{i}(x^{i}\mid s^{i})p_{S^{i}}(s^{i})>0$.\smallskip{}

\textbf{Claim B.} For all $i\in I$\textbf{: (i) }For all $\theta^{i}\in\hat{\Theta}^{i}$,
\emph{$\lim_{t\rightarrow\infty}K_{t}^{i}(h,\theta^{i})=-K^{i}(\sigma,\theta^{i})$,
a.s. in $\mathcal{H}$}; \textbf{(ii)} \emph{There exist $\xi^{*}>0$
and $T_{\xi^{*}}$ such that, $\forall t\geq T_{\xi^{*}}$, }$K_{t}^{i}(h,\theta^{i})<-(K_{0}^{i}(\sigma)+(3/2)\alpha_{\varepsilon})$
$\forall\theta^{i}\notin\Theta_{\sigma,\xi}^{i}$\emph{, a.s. in $\mathcal{H}$};
\textbf{(iii)} \emph{For all $\xi>0$, $\exists\hat{T}_{\xi}$ such
that, $\forall t\geq\hat{T}_{\xi}$, $K_{t}^{i}(h,\theta^{i})<-(K_{0}^{i}(\sigma)+(3/2)\alpha_{\varepsilon})$
$\forall\theta^{i}\in\Theta_{\sigma,\xi}^{i}\backslash\Theta_{\varepsilon}^{i}(\sigma)$,
a.s. in $\mathcal{H}$.}

\smallskip{}

\emph{Proof}: Define $freq_{t}^{i}(z^{i})=\frac{1}{t}\sum_{\tau=0}^{t-1}\boldsymbol{1}_{z^{i}}(z_{\tau}^{i})$
$\forall z^{i}\in\mathbb{Z}^{i}$. $K_{t}^{i}$ can be written as
$K_{t}^{i}(h,\theta^{i})=\kappa_{1t}^{i}(h)+\kappa_{2t}^{i}(h)+\kappa_{3t}^{i}(h,\theta^{i})$,
where $\kappa_{1t}^{i}(h)=-t^{-1}\sum_{\tau=0}^{t-1}\sum_{z^{i}\in\mathbb{Z}^{i}}\left(\boldsymbol{1}_{z^{i}}(z_{\tau}^{i})-\bar{P}_{\sigma_{\tau}}^{i}(z^{i})\right)\ln Q_{\sigma_{\tau}}^{i}(z^{i})$,
$\kappa_{2t}^{i}(h)=-t^{-1}\sum_{\tau=0}^{t-1}\sum_{z^{i}\in\mathbb{Z}^{i}}\bar{P}_{\sigma_{\tau}}^{i}(z^{i})\ln Q_{\sigma_{\tau}}^{i}(z^{i})$,
and $\kappa_{3t}^{i}(h,\theta^{i})=\sum_{z^{i}\in\mathbb{Z}^{i}}freq_{t}^{i}(z^{i})\ln Q_{\theta^{i}}^{i}(z^{i})$.
The statements made below hold almost surely in $\mathcal{H}$, but
we omit this qualification. 

First, we show $\lim_{t\rightarrow\infty}\kappa_{1t}^{i}(h)=0$. Define
$l_{\tau}^{i}(h,z^{i})=\left(\boldsymbol{1}_{z^{i}}(z_{\tau}^{i})-\bar{P}_{\sigma_{\tau}}^{i}(z^{i})\right)\ln Q_{\sigma_{\tau}}^{i}(z^{i})$
and $L_{t}^{i}(h,z^{i})=\sum_{\tau=1}^{t}\tau^{-1}l_{\tau}^{i}(h,z^{i})$
$\forall z^{i}\in\mathbb{Z}^{i}$. Fix any $z^{i}\in\mathbb{Z}^{i}$.
We show that $L_{t}^{i}(\cdot,z^{i})$ converges a.s. to an integrable,
and, therefore, finite function $L_{\infty}^{i}(\cdot,z^{i})$. To
show this, we use martingale convergence results. Let $h^{t}$ denote
the partial history until time $t$. Since $E_{\mathbf{P}^{\mu_{0},\phi}(\cdot\mid h^{t})}\left[l_{t+1}^{i}(h,z^{i})\right]=0$,
then $E_{\mathbf{P}^{\mu_{0},\phi}(\cdot\mid h^{t})}\bigl[L_{t+1}^{i}(h,z^{i})\bigr]=L_{t}^{i}(h,z^{i})$
and so $(L_{t}^{i}(h,z^{i}))_{t}$ is a martingale with respect to
$\mathbf{P}^{\mu_{0},\phi}$. Next, we show that $\sup_{t}E_{\mathbf{P}^{\mu_{0},\phi}}\left[|L_{t}^{i}(h,z^{i})|\right]\leq M$
for $M<\infty$. Note that $E_{\mathbf{P}^{\mu_{0},\phi}}\bigl[\bigl(L_{t}^{i}(h,z^{i})\bigr)^{2}\bigr]=E_{\mathbf{P}^{\mu_{0},\phi}}\bigl[\sum_{\tau=1}^{t}\tau^{-2}\bigl(l_{\tau}^{i}(h,z^{i})\bigr)^{2}+2\sum_{\tau'>\tau}\frac{1}{\tau'\tau}l_{\tau}^{i}(h,z^{i})l_{\tau'}^{i}(h,z^{i})\bigr]$.
Since $(l_{t}^{i})_{t}$ is a martingale difference sequence, for
$\tau'>\tau$, $E_{\mathbf{P}^{\mu_{0},\phi}}\bigl[l_{\tau}^{i}(h,z^{i})l_{\tau'}^{i}(h,z^{i})\bigr]=0$.
Therefore, $E_{\mathbf{P}^{\mu_{0},\phi}}\bigl[\bigl(L_{t}^{i}(h,z^{i})\bigr)^{2}\bigr]=\sum_{\tau=1}^{t}\tau^{-2}E_{\mathbf{P}^{\mu_{0},\phi}}\bigl[\bigl(l_{\tau}^{i}(h,z^{i})\bigr){}^{2}\bigr]$.
Note also that $E_{\mathbf{P}^{\mu_{0},\phi}(\cdot\mid h^{\tau-1})}\bigl[\bigl(l_{\tau}^{i}(h,z^{i})\bigr)^{2}\bigr]\leq\bigl(\ln Q_{\sigma_{\tau}}^{i}(z^{i})\bigr)^{2}Q_{\sigma_{\tau}}^{i}(z^{i})$.
Therefore, by the law of iterated expectations, $E_{\mathbf{P}^{\mu_{0},\phi}}\bigl[\bigl(L_{t}^{i}(h,z^{i})\bigr)^{2}\bigr]\leq\sum_{\tau=1}^{t}\tau^{-2}E_{\mathbf{P}^{\mu_{0},\phi}}\Bigl[\bigl(\ln Q_{\sigma_{\tau}}^{i}(z^{i})\bigr)^{2}Q_{\sigma_{\tau}}^{i}(z^{i})\Bigr]$,
which in turn is bounded above by 1 because $(\ln x)^{2}x\leq1$ $\forall x\in[0,1]$,
where we use the convention that $(\ln0)^{2}0=0$. Therefore, $\sup_{t}E_{\mathbf{P}^{\mu_{0},\phi}}\left[|L_{t}^{i}|\right]\leq1$.
By Theorem 5.2.8 in \citet{Durrett2010}, $L_{t}^{i}(h,z^{i})$ converges
a.s.-$\mathbf{P}^{\mu_{0},\phi}$ to a finite $L_{\infty}^{i}(h,z^{i})$.
Thus, by Kronecker's lemma (\citet{pollard2001}, page 105)\footnote{This lemma implies that for a sequence $(\ell_{t})_{t}$ if $\sum_{\tau}\ell_{\tau}<\infty$,
then $\sum_{\tau=1}^{t}\frac{b_{\tau}}{b_{t}}\ell_{\tau}\rightarrow0$
where $(b_{t})_{t}$ is a nondecreasing positive real valued that
diverges to $\infty$. We can apply the lemma with $\ell_{t}\equiv t^{-1}l_{t}$
and $b_{t}=t$. }, $\lim_{t\rightarrow\infty}\sum_{z^{i}\in\mathbb{Z}^{i}}\bigl\{ t^{-1}\sum_{\tau=1}^{t}\bigl(1_{z^{i}}(z_{\tau}^{i})-\bar{P}_{\sigma_{\tau}}^{i}(z^{i})\bigr)\ln Q_{\sigma_{\tau}}^{i}(z^{i})\bigr\}=0$.
Therefore, $\lim_{t\rightarrow\infty}\kappa_{1t}^{i}(h)=0$.

Next, consider $\kappa_{2t}^{i}(h)$. The assumption that $\lim_{t\rightarrow\infty}\sigma_{t}=\sigma$
and continuity of $Q_{\sigma}^{i}\ln Q_{\sigma}^{i}$ in $\sigma$
imply that $\lim_{t\rightarrow\infty}\kappa_{2t}^{i}(h)=-\sum_{(s^{i},x^{i})\in\mathbb{S}^{i}\times\mathbb{X}^{i}}E_{Q_{\sigma}(\cdot\mid s^{i},x^{i})}\left[\ln Q_{\sigma}^{i}(Y^{i}\mid s^{i},x^{i})\right]\sigma^{i}(x^{i}\mid s^{i})p_{S^{i}}(s^{i})$.

The limits of $\kappa_{1t}^{i},\kappa_{2t}^{i}$ imply that, $\forall\gamma>0$,
$\exists\hat{t}_{\gamma}$ such that, $\forall t\ge\hat{t}_{\gamma}$,
\begin{equation}
\bigl|\kappa_{1t}^{i}(h)+\kappa_{2t}^{i}(h)+\sum_{(s^{i},x^{i})\in\mathbb{S}^{i}\times\mathbb{X}^{i}}E_{Q_{\sigma}(\cdot\mid s^{i},x^{i})}\left[\ln Q_{\sigma}^{i}(Y^{i}\mid s^{i},x^{i})\right]\sigma^{i}(x^{i}\mid s^{i})p_{S^{i}}(s^{i})\bigr|\leq\gamma.\label{eq:kappa1-2}
\end{equation}

We now prove (i)-(iii) by characterizing the limit of $\kappa_{3t}^{i}(h,\theta^{i})$.

(i) For all $z^{i}\in\mathbb{Z}^{i}$, $\left|freq_{t}(z^{i})-\bar{P}_{\sigma}^{i}(z^{i})\right|\leq\bigl|t^{-1}\sum_{\tau=0}^{t-1}\bigl(\boldsymbol{1}_{z^{i}}(z_{\tau}^{i})-\bar{P}_{\sigma_{\tau}}^{i}(z^{i})\bigr)\bigr|+\bigl|t^{-1}\sum_{\tau=0}^{t-1}\bigl(\bar{P}_{\sigma_{\tau}}^{i}(z^{i})-\bar{P}_{\sigma}^{i}(z^{i})\bigr)\bigr|$.
The first term in the RHS goes to 0 (the proof is essentially identical
to the proof above that $\kappa_{1}^{i}$ goes to 0). The second term
goes to 0 because $\lim_{t\rightarrow\infty}\sigma_{t}=\sigma$ and
$\bar{P}_{\cdot}^{i}$ is continuous. Thus, $\forall\zeta>0$, $\exists\hat{t}_{\zeta}$
such that, $\forall t\geq\hat{t}_{\zeta}$, $\forall z^{i}\in\mathbb{Z}^{i}$,
\begin{equation}
\left|freq_{t}^{i}(z^{i})-\bar{P}_{\sigma}^{i}(z^{i})\right|<\zeta\label{eq:freq_converges_berk}
\end{equation}
Thus, since $\theta^{i}\in\hat{\Theta}^{i}$, $\lim_{t\rightarrow\infty}\kappa_{3t}^{i}(h,\theta^{i})=\sum_{(s^{i},x^{i})\in\mathbb{S}^{i}\times\mathbb{X}^{i}}E_{Q_{\sigma}(\cdot\mid s^{i},x^{i})}\left[\ln Q_{\theta^{i}}^{i}(Y^{i}\mid s^{i},x^{i})\right]\sigma^{i}(x^{i}\mid s^{i})p_{S^{i}}(s^{i})$.
This expression and (\ref{eq:kappa1-2}) establish part (i).

(ii) For all $\theta^{i}\notin\Theta_{\sigma,\xi}^{i}$, let $z_{\theta^{i}}^{i}$
be such that $\bar{P}_{\sigma}^{i}(z_{\theta^{i}}^{i})>0$ and $Q_{\theta^{i}}^{i}(z_{\theta^{i}}^{i})<\xi$.
By (\ref{eq:freq_converges_berk}), $\exists t_{p_{L}^{i}/2}$ such
$\forall t\geq t_{p_{L}^{i}/2}$, $\kappa_{3t}^{i}(h,\theta^{i})\leq freq_{t}^{i}(z_{\theta^{i}}^{i})\ln Q_{\theta^{i}}^{i}(z_{\theta^{i}}^{i})\leq\left(p_{L}^{i}/2\right)\ln\xi$
$\forall\theta^{i}\notin\Theta_{\sigma,\xi}^{i}$, where $p_{L}^{i}=\min_{\mathbb{Z}^{i}}\{\bar{P}_{\sigma}^{i}(z^{i}):\bar{P}_{\sigma}^{i}(z^{i})>0\}$.
This result and (\ref{eq:kappa1-2}) imply that, $\forall t\geq t_{1}\equiv\max\{t_{p_{L}^{i}/2},\hat{t}_{1}\}$,
\begin{align}
K_{t}^{i}(h,\theta^{i}) & \leq-\sum_{(s^{i},x^{i})\in\mathbb{S}^{i}\times\mathbb{X}^{i}}E_{Q_{\sigma}(\cdot\mid s^{i},x^{i})}\left[\ln Q_{\sigma}^{i}(Y^{i}\mid s^{i},x^{i})\right]\sigma^{i}(x^{i}\mid s^{i})p_{S^{i}}(s^{i})+1+\left(p_{L}^{i}/2\right)\ln\xi\nonumber \\
 & \leq\#\mathbb{Z}^{i}+1+\left(q_{L}^{i}/2\right)\ln\xi,\label{eq:KLcasei}
\end{align}
 $\forall\theta^{i}\notin\Theta_{\sigma,\xi}^{i}$, where the second
line follows from the facts that $\sigma^{i}(x^{i}\mid s^{i})p_{S^{i}}(s^{i})\leq1$
and $x\ln(x)\in[-1,0]$ $\forall x\in[0,1]$. In addition, the fact
that $K_{0}^{i}(\sigma)<\infty$ and $\alpha_{\varepsilon}\leq\bar{\alpha}<\infty$
$\forall\varepsilon\leq\bar{\varepsilon}$ implies that the RHS of
(\ref{eq:KLcasei}) can be made lower than $-(K_{0}^{i}(\sigma)+(3/2)\alpha_{\varepsilon})$
for some sufficiently small $\xi^{*}$.

(iii) For any $\xi>0$, let $\zeta_{\xi}=-\alpha_{\varepsilon}/(\#\mathbb{Z}^{i}4\ln\xi)>0$.
By (\ref{eq:freq_converges_berk}), $\exists\hat{t}_{\zeta_{\xi}}$
such that, $\forall t\geq\hat{t}_{\zeta_{\xi}}$, 
\begin{align*}
\kappa_{3t}^{i}(h,\theta^{i}) & \leq\sum_{\{z^{i}:\bar{P}_{\sigma}^{i}(z^{i})>0\}}freq_{t}^{i}(z^{i})\ln Q_{\theta^{i}}^{i}(z^{i})\leq\sum_{\{z^{i}:\bar{P}_{\sigma}^{i}(z^{i})>0\}}\left(\bar{P}_{\sigma}^{i}(z^{i})-\zeta_{\xi}\right)\ln Q_{\theta^{i}}^{i}(z^{i})\\
 & \leq\sum_{(s^{i},x^{i})\in\mathbb{S}^{i}\times\mathbb{X}^{i}}E_{Q_{\sigma}(\cdot\mid s^{i},x^{i})}\left[\ln Q_{\theta^{i}}^{i}(Y^{i}\mid s^{i},x^{i})\right]\sigma^{i}(x^{i}\mid s^{i})p_{S^{i}}(s^{i})-\#\mathbb{Z}^{i}\zeta_{\xi}\ln\xi,
\end{align*}
$\forall\theta^{i}\in\Theta_{\sigma,\xi}^{i}$ (since $Q_{\theta^{i}}^{i}(z^{i})\geq\xi$
$\forall z^{i}$ such that $\bar{P}_{\sigma}^{i}(z^{i})>0$). $ $The
above expression, the fact that \emph{$\alpha_{\varepsilon}/4=-\#\mathbb{Z}^{i}\zeta_{\xi}\ln\xi$,}
and (\ref{eq:kappa1-2}) imply that,\emph{ $\forall t\geq\hat{T}_{\xi}\equiv\max\{\hat{t}_{\zeta_{\xi}},\hat{t}_{\alpha_{\varepsilon}/4}\}$,
$K_{t}^{i}(h,\theta^{i})<-K^{i}(\sigma,\theta^{i})+\alpha_{\varepsilon}/2$}
$\forall\theta^{i}\in\Theta_{\sigma,\xi}^{i}$. This result and (\ref{eq:K_theta>K_0+2alpha_v2})
imply the desired result. $\square$

\pagebreak{}

\appendix
\setcounter{page}{1}

\part*{Online Appendix}

\section{\label{sec:OA_trade}Example: Trading with adverse selection}

In this section, we provide the formal details for the trading environment
in Example 2.5. Let $p\in\Delta(\mathbb{A}\times\mathbb{V})$ be the
true distribution; we use subscripts, such as $p_{A}$ and $p_{V\mid A}$,
to denote the corresponding marginal and conditional distributions.
Let $\mathbb{Y}=\mathbb{A}\times\mathbb{V}\cup\{\square\}$ denote
the space of observable consequences, where $\square$ will be a convenient
way to represent the fact that there is no trade. We denote the random
variable taking values in $\mathbb{V}\cup\{\square\}$ by $\hat{V}$.
Notice that the state space in this example is $\Omega=\mathbb{A}\times\mathbb{V}$.

Partial feedback is represented by the function $f^{P}:\mathbb{X}\times\mathbb{A}\times\mathbb{V}\rightarrow\mathbb{Y}$
such that $f^{P}(x,a,v)=(a,v)$ if $a\leq x$ and $f^{P}(x,a,v)=(a,\square)$
if $a>x$. Full feedback is represented by $f^{F}(x,a,v)=(a,v)$.
In all cases, payoffs are given by $\pi:\mathbb{X}\times\mathbb{Y}\rightarrow\mathbb{R}$,
where $\pi(x,a,v)=v-x$ if $a\leq x$ and 0 otherwise. The objective
distribution for the case of partial feedback, $Q^{P}$, is, $\forall x\in\mathbb{X}$,
$\forall(a,v)\in\mathbb{A}\times\mathbb{V}$, $Q^{P}(a,v\mid x)=p(a,v)1_{\{x\geq a\}}(x)$,
and, $\forall x\in\mathbb{X}$, $\forall a\in\mathbb{A}$, $Q^{P}(a,\square\mid x)=p_{A}(a)1_{\{x<a\}}(x)$.
The objective distribution for the case of full feedback, $Q^{F}$,
is, $\forall x\in\mathbb{X}$, $\forall(a,v)\in\mathbb{A}\times\mathbb{V}$,
$Q^{F}(a,v\mid x)=p(a,v)$, and, $\forall x\in\mathbb{X}$, $\forall a\in\mathbb{A}$,
$Q^{F}(a,\square\mid x)=0$.

The buyer knows the environment except for the distribution $p\in\Delta(\mathbb{A}\times\mathbb{V})$.
Then, for any distribution in the subjective model, $Q_{\theta}$,
the perceived expected profit from choosing $x\in\mathbb{X}$ is 
\begin{equation}
E_{Q_{\theta}(\cdot\mid x)}[\pi(x,A,\hat{V})]=\sum_{(a,v)\in\mathbb{A}\times\mathbb{V}}1_{\{x\geq a\}}(x)\left(v-x\right)Q_{\theta}(a,v\mid x).\label{eq:trade_expprofits}
\end{equation}

The buyer has either one of two misspecifications over $p$ captured
by the parameter sets $\Theta_{I}=\Delta(\mathbb{A})\times\Delta(\mathbb{V})$
(i.e., independent beliefs) or $\Theta_{A}=\times_{j}\Delta(\mathbb{A})\times\Delta(\mathbb{V})$
(i.e., analogy-based beliefs) defined in the main text. Thus, combining
feedback and parameter sets, we have four cases to consider, and,
for each case, we write down the subjective model and wKLD function.

\textbf{\emph{Cursed equilibrium}}. Feedback is $f^{F}$ and the parameter
set is $\Theta_{I}$. The subjective model is, $\forall x\in\mathbb{X}$,
$\forall(a,v)\in\mathbb{A}\times\mathbb{V}$, $Q_{\theta}^{C}(a,v\mid x)=\theta_{A}(a)\theta_{V}(v)$,
and, $\forall x\in\mathbb{X}$, $\forall a\in\mathbb{A}$, $Q_{\theta}^{C}(a,\square\mid x)=0$,
where $\theta=(\theta_{A},\theta_{V})\in\Theta_{I}$.\footnote{In fact, the symbol $\square$ is not necessary for this example,
but we keep it so that all feedback functions are defined over the
same space of consequences.} This is an analogy-based game. From (\ref{eq:trade_expprofits}),
the perceived expected profit from $x\in\mathbb{X}$ is 
\begin{equation}
Pr_{\theta_{A}}\left(A\leq x\right)\left(E_{\theta_{V}}\left[V\right]-x\right),\label{eq:trade_expprofits_C}
\end{equation}
where $Pr_{\mathbb{\theta}_{A}}$ denotes probability with respect
to $\theta_{A}$ and $ $$E_{\theta_{V}}$ denotes expectation with
respect to $\theta_{V}$. Also, for all (pure) strategies $x\in\mathbb{X}$,
the wKLD function is\footnote{In all cases, the extension to mixed strategies is straightforward.}
\begin{align*}
K^{C}(x,\theta) & =E_{Q^{F}(\cdot\mid x)}\bigl[\ln\frac{Q^{F}(A,\hat{V}\mid x)}{Q_{\theta}^{C}(A,\hat{V}\mid x)}\bigr]=\sum_{(a,v)\in\mathbb{A}\times\mathbb{V}}p(a,v)\ln\frac{p(a,v)}{\theta_{A}(a)\theta_{V}(v)}.
\end{align*}
For each $x\in\mathbb{X}$, $\theta(x)=(\theta_{A}(x),\theta_{V}(x))\in\Theta_{I}=\Delta(\mathbb{A})\times\Delta(\mathbb{V})$,
where $\theta_{A}(x)=p_{A}$ and $\theta_{V}(x)=p_{V}$ is the unique
parameter value that minimizes $K^{C}(x,\cdot)$. Together with (\ref{eq:trade_expprofits_C}),
we obtain equation $\Pi^{CE}$ in the main text.

\textbf{\emph{Behavioral equilibrium (naive version)}}. Feedback is
$f^{P}$ and the parameter set is $\Theta_{I}$. The subjective model
is, $\forall x\in\mathbb{X}$, $\forall(a,v)\in\mathbb{A}\times\mathbb{V}$,
$Q_{\theta}^{BE}(a,v\mid x)=\theta_{A}(a)\theta_{V}(v)1_{\{x\geq a\}}(x)$,
and, $\forall x\in\mathbb{X}$, $\forall a\in\mathbb{A}$, $Q_{\theta}^{BE}(a,\square\mid x)=\theta_{A}(a)1_{\{x<a\}}(x)$,
where $\theta=(\theta_{A},\theta_{V})\in\Theta_{I}$. From (\ref{eq:trade_expprofits}),
perceived expected profit from $x\in\mathbb{X}$ is as in equation
(\ref{eq:trade_expprofits_C}). Also, for all (pure) strategies $x\in\mathbb{X}$,
the wKLD function is 
\begin{align*}
K^{BE}(x,\theta) & =E_{Q^{P}(\cdot\mid x)}\bigl[\ln\frac{Q^{P}(A,\hat{V}\mid x)}{Q_{\theta}^{BE}(A,\hat{V}\mid x)}\bigr]\\
 & =\sum_{\{a\in\mathbb{A}:a>x\}}p_{A}(a)\ln\frac{p_{A}(a)}{\theta_{A}(a)}+\sum_{\{(a,v)\in\mathbb{A}\times\mathbb{V}:a\leq x\}}p(a,v)\ln\frac{p(a,v)}{\theta_{A}(a)\theta_{V}(v)}.
\end{align*}
For each $x\in\mathbb{X}$, $\theta(x)=(\theta_{A}(x),\theta_{V}(x))\in\Theta_{I}=\Delta(\mathbb{A})\times\Delta(\mathbb{V})$,
where $\theta_{A}(x)=p_{A}$ and $\theta_{V}(x)(v)=p_{V\mid A}(v\mid A\leq x)$
$\forall v\in\mathbb{V}$ is the unique parameter value that minimizes
$K^{BE}(x,\cdot)$. Together with (\ref{eq:trade_expprofits_C}),
we obtain equation $\Pi^{BE}$ in the main text.

\textbf{\emph{Analogy-based expectations equilibrium}}. Feedback is
$f^{F}$ and the parameter set is $\Theta_{A}$. The subjective model
is, $\forall x\in\mathbb{X}$, $\forall(a,v)\in\mathbb{A}\times\mathbb{V}_{j}$,
all $j=1,...,k$, $Q_{\theta}^{ABEE}(a,v\mid x)=\theta_{j}(a)\theta_{V}(v)$,
and, $\forall x\in\mathbb{X}$, $\forall a\in\mathbb{A}$, $Q_{\theta}^{ABEE}(a,\square\mid x)=0$,
where $\theta=(\theta_{1},...,\theta_{k},\theta_{V})\in\Theta_{A}$.
This is an analogy-based game. From (\ref{eq:trade_expprofits}),
perceived expected profit from $x\in\mathbb{X}$ is 
\begin{equation}
\sum_{j=1}^{k}Pr_{\theta_{V}}(V\in\mathbb{V}_{j})\left\{ Pr_{\theta_{j}}(A\leq x\mid V\in\mathbb{V}_{j})\left(E_{\theta_{V}}\left[V\mid V\in\mathbb{V}_{j}\right]-x\right)\right\} .\label{eq:trade_expprofits_ABEE}
\end{equation}
Also, for all (pure) strategies $x\in\mathbb{X}$, the wKLD function
is 
\begin{align*}
K^{ABEE}(x,\theta) & =E_{Q^{F}(\cdot\mid x)}\bigl[\ln\frac{Q^{F}(A,\hat{V}\mid x)}{Q_{\theta}^{ABEE}(A,\hat{V}\mid x)}\bigr]=\sum_{j=1}^{k}\sum_{(a,v)\in\mathbb{A}\times\mathbb{V}_{j}}p(a,v)\ln\frac{p(a,v)}{\theta_{j}(a)\theta_{V}(v)}.
\end{align*}
For each $x\in\mathbb{X}$, $\theta(x)=(\theta_{1}(x),...,\theta_{k}(x),\theta_{V}(x))\in\Theta_{A}=\times_{j}\Delta(\mathbb{A})\times\Delta(\mathbb{V})$,
where $\theta_{j}(x)(a)=p_{A\mid V_{j}}(a\mid V\in\mathbb{V}_{j})$
$\forall a\in\mathbb{A}$ and $\theta_{V}(x)=p_{V}$ is the unique
parameter value that minimizes $K^{ABEE}(x,\cdot)$. Together with
(\ref{eq:trade_expprofits_ABEE}), we obtain equation $\Pi^{ABEE}$
in the main text.

\textbf{\emph{Behavioral equilibrium (naive version) with analogy
classes}}. It is natural to also consider a case, unexplored in the
literature, where feedback $f^{P}$ is partial and the subjective
model is parameterized by $\Theta_{A}$. Suppose that the buyer's
behavior has stabilized to some price $x^{*}$. Due to the possible
correlation across analogy classes, the buyer might now believe that
deviating to a different price $x\neq x^{*}$ affects her valuation.
In particular, the buyer might have multiple beliefs at $x^{*}$.
To obtain a natural equilibrium refinement, we assume that the buyer
also observes the analogy class that contains her realized valuation,
whether she trades or not, and that $\Pr(V\in\mathbb{V}_{j},\,A\leq x)>0$
for all $j=1,...,k$ and $x\in\mathbb{X}$.\footnote{Alternatively, and more naturally, we could require the equilibrium
to be the limit of a sequence of mixed strategy equilibria with the
property that all prices are chosen with positive probability.} We denote this new feedback assumption by a function $f^{P^{*}}:\mathbb{X}\times\mathbb{A}\times\mathbb{V}\rightarrow\mathbb{Y}^{*}$
where $\mathbb{Y}^{*}=\mathbb{A}\times\mathbb{V}\cup\{1,...,k\}$
and $f^{P^{*}}(x,a,v)=(a,v)$ if $a\leq x$ and $f^{P^{*}}(x,a,v)=(a,j)$
if $a>x$ and $v\in\mathbb{V}_{j}$. The objective distribution given
this feedback function is, $\forall x\in\mathbb{X}$, $\forall(a,v)\in\mathbb{A}\times\mathbb{V}$,
$Q^{P^{*}}(a,v\mid x)=p(a,v)1_{\{x\geq a\}}(x)$, and, $\forall x\in\mathbb{X}$,
$\forall a\in\mathbb{A}$ and all $j=1,...,k$, $Q^{P^{*}}(a,j\mid x)=p_{A\mid V_{j}}(a\mid V\in\mathbb{V}_{j})p_{V}(\mathbb{V}_{j})1_{\{x<a\}}(x)$.
The subjective model is, $\forall x\in\mathbb{X}$, $\forall(a,v)\in\mathbb{A}\times\mathbb{V}_{j}$
and all $j=1,...,k$, $Q_{\theta}^{BEA}(a,v\mid x)=\theta_{j}(a)\theta_{V}(v)1_{\{x\geq a\}}(x)$,
and, $\forall x\in\mathbb{X}$, $\forall(a,v)\in\mathbb{A}\times\mathbb{V}_{j}$
and all $j=1,...,k$, $Q_{\theta}^{BEA}(a,j\mid x)=\theta_{j}(a)\left(\sum_{v\in\mathbb{V}_{j}}\theta_{V}(v)\right)1_{\{x<a\}}(x)$,
where $\theta=(\theta_{1},\theta_{2},...,\theta_{k},\theta_{V})\in\Theta_{A}$.
In particular, from (\ref{eq:trade_expprofits}), perceived expected
profit from $x\in\mathbb{X}$ is as in equation (\ref{eq:trade_expprofits_ABEE}).
Also, for all (pure) strategies $x\in\mathbb{X}$, the wKLD function
is 
\begin{align*}
K^{BEA}(x,\theta)= & E_{Q^{P^{*}}(\cdot\mid x)}\bigl[\ln\frac{Q^{P^{*}}(A,\hat{V}\mid x)}{Q_{\theta}^{BEA}(A,\hat{V}\mid x)}\bigr]=\sum_{j=1}^{k}\sum_{\{(a,v)\in\mathbb{A}\times\mathbb{V}_{j}:a\leq x\}}p(a,v)\ln\frac{p(a,v)}{\theta_{j}(a)\theta_{V}(v)}\\
+ & \sum_{\{(a,j)\in\mathbb{A}\times\{1,...,k\}:a>x\}}p_{A\mid V_{j}}(a\mid V\in\mathbb{V}_{j})p_{V}(\mathbb{V}_{j})\ln\frac{p_{A\mid V_{j}}(a\mid V\in\mathbb{V}_{j})p_{V}(\mathbb{V}_{j})}{\theta_{j}(a)\sum_{v\in\mathbb{V}_{j}}\theta_{V}(v)}.
\end{align*}
For each $x\in\mathbb{X}$, $\theta(x)=(\theta_{1}(x),...,\theta_{k}(x),\theta_{V}(x))\in\Theta_{A}=\times_{j}\Delta(\mathbb{A})\times\Delta(\mathbb{V})$,
where $\theta_{j}(x)(a)=p_{A\mid V_{j}}(a\mid V\in\mathbb{V}_{j})$
$\forall a\in\mathbb{A}$ and $\theta_{V}(x)(v)=p_{V\mid A}(v\mid V\in\mathbb{V}_{j},\,A\leq x)p_{V}(\mathbb{V}_{j})$
$\forall v\in\mathbb{V}_{j}$, all $j=1,...,k$ is the unique parameter
value that minimizes $K^{BEA}(x,\cdot)$. Together with (\ref{eq:trade_expprofits_ABEE}),
we obtain $\Pi^{BEA}(x,x^{*})=\sum_{i=j}^{k}\Pr(V\in\mathbb{V}_{j})\Pr(A\leq x\mid V\in\mathbb{V}_{j})\bigl(E\bigl[V\mid V\in\mathbb{V}_{j},A\leq x^{*}\bigr]-x\bigr).$

\section{\label{sec:converse}Proof of converse result: Theorem \ref{Theo:converse-1}}

\bigskip{}

Let $(\bar{\mu}^{i})_{i\in I}$ be a belief profile that supports
$\sigma$ as an equilibrium. Consider the following policy profile
$\phi=(\phi_{t}^{i})_{i,t}$: For all $i\in I$ and all $t$,
\[
(\mu^{i},s^{i},\xi^{i})\mapsto\phi_{t}^{i}(\mu^{i},s^{i},\xi^{i})\equiv\begin{cases}
\varphi^{i}(\bar{\mu}^{i},s^{i},\xi^{i}) & \text{ if }\max_{i\in I}||\bar{Q}_{\mu^{i}}^{i}-\bar{Q}_{\bar{\mu}^{i}}^{i}||\leq\frac{1}{2C}\varepsilon_{t}\\
\varphi^{i}(\mu^{i},s^{i},\xi^{i}) & \text{ otherwise},
\end{cases}
\]
where $\varphi^{i}$ is an arbitrary selection from $\Psi^{i}$, $C\equiv\max_{I}\left\{ \#\mathbb{Y}^{i}\times\sup_{\mathbb{X}^{i}\times\mathbb{Y}^{i}}|\pi^{i}(x^{i},y^{i})|\right\} <\infty$,
and the sequence $(\varepsilon_{t})_{t}$ will be defined below. For
all $i\in I$, fix any prior $\mu_{0}^{i}$ such that $\mu_{0}^{i}(\cdot|\Theta^{i}(\sigma))=\bar{\mu}^{i}$
(where for any $A\subset\Theta$ Borel, $\mu(\cdot|A)$ is the conditional
probability given $A$).

We now show that if $\varepsilon_{t}\geq0$ $\forall t$ and $\lim_{t\rightarrow\infty}\varepsilon_{t}=0$,
then $\phi$ is asymptotically optimal. Throughout this argument,
we fix an arbitrary $i\in I$. Abusing notation, let $U^{i}(\mu^{i},s^{i},\xi^{i},x^{i})=E_{\bar{Q}_{\mu^{i}}(\cdot|s^{i},x^{i})}\left[\pi^{i}(x^{i},Y^{i})\right]+\xi^{i}(x^{i})$.
It suffices to show that 
\begin{equation}
U^{i}(\mu^{i},s^{i},\xi^{i},\phi_{t}^{i}(\mu^{i},s^{i},\xi^{i}))\geq U^{i}(\mu^{i},s^{i},\xi^{i},x^{i})-\varepsilon_{t}\label{eq:pf_asympt_opt}
\end{equation}
for all $(i,t)$, all $(\mu^{i},s^{i},\xi^{i})$, and all $x^{i}$.
By construction of $\phi$, equation (\ref{eq:pf_asympt_opt}) is
satisfied if $\max_{i\in I}||\bar{Q}_{\mu^{i}}^{i}-\bar{Q}_{\bar{\mu}^{i}}^{i}||>\frac{1}{2C}\varepsilon_{t}$.
If, instead, $\max_{i\in I}||\bar{Q}_{\mu^{i}}^{i}-\bar{Q}_{\bar{\mu}^{i}}^{i}||\leq\frac{1}{2C}\varepsilon_{t}$,
then 
\begin{equation}
U^{i}(\bar{\mu}^{i},s^{i},\xi^{i},\phi_{t}^{i}(\mu^{i},s^{i},\xi^{i}))=U^{i}(\bar{\mu}^{i},s^{i},\xi^{i},\varphi^{i}(\bar{\mu}^{i},s^{i},\xi^{i}))\geq U^{i}(\bar{\mu}^{i},s^{i},\xi^{i},x^{i}),\label{eq:converse-0}
\end{equation}
 $\forall x^{i}\in\mathbb{X}^{i}$. Moreover, $\forall x^{i}$, 
\begin{align*}
\left|U^{i}(\bar{\mu}^{i},s^{i},\xi^{i},x^{i})-U^{i}(\mu^{i},s^{i},\xi^{i},x^{i})\right|= & \bigl|\sum_{y^{i}\in\mathbb{Y}^{i}}\pi(x^{i},y^{i})\bigl(\bar{Q}_{\bar{\mu}^{i}}^{i}(y^{i}\mid s^{i},x^{i})-\bar{Q}_{\mu^{i}}^{i}(y^{i}\mid s^{i},x^{i})\bigr)\bigr|\\
\leq & \sup_{\mathbb{X}^{i}\times\mathbb{Y}^{i}}|\pi^{i}(x^{i},y^{i})|\sum_{y^{i}\in\mathbb{Y}^{i}}\bigl|\bigl(\bar{Q}_{\bar{\mu}^{i}}^{i}(y^{i}\mid s^{i},x^{i})-\bar{Q}_{\mu^{i}}^{i}(y^{i}\mid s^{i},x^{i})\bigr)\bigr|\\
\leq & \sup_{\mathbb{X}^{i}\times\mathbb{Y}^{i}}|\pi^{i}(x^{i},y^{i})|\times\#\mathbb{Y}^{i}\times\max_{y^{i},x^{i},s^{i}}\bigl|\bar{Q}_{\bar{\mu}^{i}}^{i}(y^{i}\mid s^{i},x^{i})-\bar{Q}_{\mu^{i}}^{i}(y^{i}\mid s^{i},x^{i})\bigr|
\end{align*}
so by our choice of $C$, $\left|U^{i}(\bar{\mu}^{i},s^{i},\xi^{i},x^{i})-U^{i}(\mu^{i},s^{i},\xi^{i},x^{i})\right|\leq0.5\varepsilon_{t}$
$\forall x^{i}$. Therefore, equation (\ref{eq:converse-0}) implies
equation (\ref{eq:pf_asympt_opt}); thus $\phi$ is asymptotically
optimal if $\varepsilon_{t}\geq0$ $\forall t$ and $\lim_{t\rightarrow\infty}\varepsilon_{t}=0$.

We now construct a sequence $(\varepsilon_{t})_{t}$ such that $\varepsilon_{t}\geq0$
$\forall t$ and $\lim_{t\rightarrow\infty}\varepsilon_{t}=0$. Let
$\bar{\phi}^{i}=(\bar{\phi}_{t}^{i})_{t}$ be such that $\bar{\phi}_{t}^{i}(\mu^{i},\cdot,\cdot)=\varphi^{i}(\bar{\mu}^{i},\cdot,\cdot)$
$\forall\mu^{i}$; i.e., $\bar{\phi}^{i}$ is a stationary policy
that maximizes utility under the assumption that the belief is always
$\bar{\mu}^{i}$. Let $\zeta^{i}(\mu^{i})\equiv2C||\bar{Q}_{\mu^{i}}^{i}-\bar{Q}_{\bar{\mu}^{i}}^{i}||$
and suppose (the proof is at the end) that 
\begin{equation}
\boldsymbol{P}^{\mu_{0},\bar{\phi}}(\lim_{t\rightarrow\infty}\max_{i\in I}|\zeta^{i}(\mu_{t}^{i}(h))|=0)=1\label{eq:conv-eq-1}
\end{equation}
(recall that $\boldsymbol{P}^{\mu_{0},\bar{\phi}}$ is the probability
measure over $\mathbb{H}$ induced by the policy profile $\bar{\phi}$;
by definition of $\bar{\phi}$, $\boldsymbol{P}^{\mu_{0},\bar{\phi}}$
does not depend on $\mu_{0}$). Then by the 2nd Borel-Cantelli lemma
(\citet{billingsley1995probability}, pages 59-60), for any $\gamma>0$,
$\sum_{t}\boldsymbol{P}^{\mu_{0},\bar{\phi}}\left(\max_{i\in I}|\zeta^{i}(\mu_{t}^{i}(h))|\geq\gamma\right)<\infty$.
Hence, for any $a>0$, there exists a sequence $(\tau(j))_{j}$ such
that 
\begin{equation}
\sum_{t\geq\tau(j)}\boldsymbol{P}^{\mu_{0},\bar{\phi}}\bigl(\max_{i\in I}|\zeta^{i}(\mu_{t}^{i}(h))|\geq1/j\bigr)<\frac{3}{a}4^{-j}\label{eq:conv-eq-2}
\end{equation}
and $\lim_{j\rightarrow\infty}\tau(j)=\infty$. For all $t\leq\tau(1)$,
we set $\varepsilon_{t}=3C$, and, for any $t>\tau(1)$, we set $\varepsilon_{t}\equiv1/N(t)$,
where $N(t)\equiv\sum_{j=1}^{\infty}1\{\tau(j)\leq t\}$. Observe
that, since $\lim_{j\rightarrow\infty}\tau(j)=\infty$, $N(t)\rightarrow\infty$
as $t\rightarrow\infty$ and thus $\varepsilon_{t}\rightarrow0$.
We also note that our choice of $(\varepsilon_{t})_{t}$ depends on
the value of $a$ (through $(N(t))_{t}$) and, consequently, so does
the sequence of functions $(\phi^{i})_{t}$ for each $i\in I$.

Next, we show that $\mathbf{P}^{\mu_{0},\phi}\left(\lim_{t\rightarrow\infty}\left\Vert \sigma_{t}(h^{\infty})-\sigma\right\Vert =0\right)=1,$
where $(\sigma_{t})_{t}$ is the sequence of intended strategies given
$\phi$, i.e., $\sigma_{t}^{i}(h)(x^{i}\mid s^{i})=P_{\xi}\left(\xi^{i}:\phi_{t}^{i}(\mu_{t}^{i}(h),s^{i},\xi^{i})=x^{i}\right).$
Observe that, by definition, $\sigma^{i}(x^{i}\mid s^{i})=P_{\xi}\bigl(\xi^{i}:x^{i}\in\arg\max_{\hat{x}^{i}\in\mathbb{X}^{i}}E_{\bar{Q}_{\bar{\mu}^{i}}(\cdot\mid s^{i},\hat{x}^{i})}\left[\pi^{i}(\hat{x}^{i},Y^{i})\right]+\xi^{i}(\hat{x}^{i})\bigr).$
Since $\varphi^{i}\in\Psi^{i}$, it follows that we can write $\sigma^{i}(x^{i}\mid s^{i})=P_{\xi}\left(\xi^{i}:\varphi^{i}(\bar{\mu}^{i},s^{i},\xi^{i})=x^{i}\right)$.
Let $H\equiv\left\{ h\colon\left\Vert \sigma_{t}(h)-\sigma\right\Vert =0,\text{ for all \ensuremath{t}}\right\} $.
It is sufficient to show that $\mathbf{P}^{\mu_{0},\phi}\left(H\right)=1$.
To show this, observe that
\begin{align*}
\mathbf{P}^{\mu_{0},\phi}\left(H\right)\geq & \mathbf{P}^{\mu_{0},\phi}\left(\cap_{t}\{\max_{i}\zeta^{i}(\mu_{t})\leq\varepsilon_{t}\}\right)\\
= & \prod_{t=\tau(1)+1}^{\infty}\mathbf{P}^{\mu_{0},\phi}\left(\max_{i}\zeta^{i}(\mu_{t})\leq\varepsilon_{t}\text{ \ensuremath{\mid}}\cap_{l<t}\{\max_{i}\zeta^{i}(\mu_{l})\leq\varepsilon_{l}\}\right)\\
= & \prod_{t=\tau(1)+1}^{\infty}\mathbf{P}^{\mu_{0},\bar{\phi}}\left(\max_{i}\zeta^{i}(\mu_{t})\leq\varepsilon_{t}\text{ \ensuremath{\mid}}\cap_{l<t}\{\max_{i}\zeta^{i}(\mu_{l})\leq\varepsilon_{l}\}\right)\\
= & \mathbf{P}^{\mu_{0},\bar{\phi}}\left(\cap_{t>\tau(1)}\{\max_{i}\zeta^{i}(\mu_{t})\leq\varepsilon_{t}\}\right),
\end{align*}
where the second line omits the term $\mathbf{P}^{\mu_{0},\phi}\left(\max_{i}\zeta^{i}(\mu_{t})<\varepsilon_{t}\text{ for all \ensuremath{t}}\leq\tau(1)\right)$
because it is equal to 1 (since $\varepsilon_{t}\geq3C$ $\forall t\leq\tau(1)$);
the third line follows from the fact that $\phi_{t-1}^{i}=\bar{\phi}_{t-1}^{i}$
if $\zeta^{i}(\mu_{t-1})\leq\varepsilon_{t-1}$, so the probability
measure is equivalently given by $\boldsymbol{P}^{\mu_{0},\bar{\phi}}$;
and where the last line also uses the fact that $\boldsymbol{P}^{\mu_{0},\bar{\phi}}\left(\max_{i}\zeta^{i}(\mu_{t})<\varepsilon_{t}\text{ for all \ensuremath{t}}\leq\tau(1)\right)=1$.
In addition, $\forall a>0$,
\begin{align*}
\boldsymbol{P}^{\mu_{0},\bar{\phi}}\left(\cap_{t>\tau(1)}\{\max_{i}\zeta^{i}(\mu_{t})\leq\varepsilon_{t}\}\right)= & \boldsymbol{P}^{\mu_{0},\bar{\phi}}\left(\cap_{n\in\{1,2,...\}}\cap_{\{t>\tau(1):N(t)=n\}}\{\max_{i}\zeta^{i}(\mu_{t})\leq n^{-1}\}\right)\\
\geq & 1-\sum_{n=1}^{\infty}\sum_{\{t:N(t)=n\}}\boldsymbol{P}^{\mu_{0},\bar{\phi}}\left(\max_{i}\zeta^{i}(\mu_{t})\geq n^{-1}\right)\\
\geq & 1-\sum_{n=1}^{\infty}\frac{3}{a}4^{-n}=1-\frac{1}{a},
\end{align*}
where the last line follows from (\ref{eq:conv-eq-2}). Thus, we have
shown that $\mathbf{P}^{\mu_{0},\phi}\left(H\right)\geq1-1/a$ $\forall a>0$;
hence, $\mathbf{P}^{\mu_{0},\phi}\left(H\right)=1$.

We conclude the proof by showing that equation (\ref{eq:conv-eq-1})
indeed holds. Observe that $\sigma$ is trivially stable under $\bar{\phi}$.
By Lemma \ref{lem:Berk}, $\forall i\in I$ and all open sets $U^{i}\supseteq\Theta^{i}(\sigma)$,
\begin{equation}
\lim_{t\rightarrow\infty}\mu_{t}^{i}\left(U^{i}\right)=1\label{eq:Berk_conv}
\end{equation}
$a.s.-\boldsymbol{P}^{\mu_{0},\bar{\phi}}$ (over $\mathbb{H}$).
Let $\mathcal{H}$ denote the set of histories such that $x_{t}^{i}(h)=x^{i}$
and $s_{t}^{i}(h)=s^{i}$ implies that $\sigma^{i}(x^{i}\mid s^{i})>0$.
By definition of $\bar{\phi}$, $ $$\boldsymbol{P}^{\mu_{0},\bar{\phi}}(\mathcal{H})=1$.
Thus, it suffices to show that $\lim_{t\rightarrow\infty}\max_{i\in I}|\zeta^{i}(\mu_{t}^{i}(h))|=0$
a.s.-$\boldsymbol{P}^{\mu_{0},\bar{\phi}}$ over $\mathcal{H}$. To
do this, take any $A\subseteq\Theta$ that is closed. By equation
(\ref{eq:Berk_conv}), $\forall i\in I$, and almost all $h\in\mathcal{H}$,
\[
\limsup_{t\rightarrow\infty}\int1_{A}(\theta)\mu_{t+1}^{i}(d\theta)=\limsup_{t\rightarrow\infty}\int1_{A\cap\Theta^{i}(\sigma)}(\theta)\mu_{t+1}^{i}(d\theta).
\]
Moreover, 
\begin{align*}
\int1_{A\cap\Theta^{i}(\sigma)}(\theta)\mu_{t+1}^{i}(d\theta)\leq & \int1_{A\cap\Theta^{i}(\sigma)}(\theta)\left\{ \frac{\prod_{\tau=1}^{t}Q_{\theta}^{i}(y_{\tau}^{i}\mid s_{\tau}^{i},x_{\tau}^{i})\mu_{0}^{i}(d\theta)}{\int_{\Theta^{i}(\sigma)}\prod_{\tau=1}^{t}Q_{\theta}^{i}(y_{\tau}^{i}\mid s_{\tau}^{i},x_{\tau}^{i})\mu_{0}^{i}(d\theta)}\right\} \\
= & \mu_{0}^{i}(A\mid\Theta^{i}(\sigma))=\bar{\mu}^{i}(A),
\end{align*}
where the first inequality follows from the fact that $\Theta^{i}(\sigma)\subseteq\Theta^{i}$;
the first equality follows from the fact that, since $h\in\mathcal{H}$,
the fact that the game is weakly identified given $\sigma$ implies
that $\prod_{\tau=1}^{t}Q_{\theta}^{i}(y_{\tau}^{i}\mid s_{\tau}^{i},x_{\tau}^{i})$
is constant with respect to $\theta$ $\forall\theta\in\Theta^{i}(\sigma)$,
and the last equality follows from our choice of $\mu_{0}^{i}$. Therefore,
we established that a.s.-$\boldsymbol{P}^{\mu_{0},\bar{\phi}}$ over
$\mathcal{H}$, $\limsup_{t\rightarrow\infty}\mu_{t+1}^{i}(h)(A)\leq\bar{\mu}^{i}(A)$
for $A$ closed. By the portmanteau lemma, this implies that, a.s.
-$\boldsymbol{P}^{\mu_{0},\bar{\phi}}$ over $\mathcal{H}$, $\lim_{t\rightarrow\infty}\int_{\Theta}f(\theta)\mu_{t+1}^{i}(h)(d\theta)=\int_{\Theta}f(\theta)\bar{\mu}^{i}(d\theta)$
for any $f$ real-valued, bounded and continuous. Since, by assumption,
$\theta\mapsto Q_{\theta}^{i}(y^{i}\mid s^{i},x^{i})$ is bounded
and continuous, the previous result applies to $Q_{\theta}^{i}(y^{i}\mid s^{i},x^{i})$,
and since $y,s,x$ take a finite number of values, this result implies
that $\lim_{t\rightarrow\infty}||\bar{Q}_{\mu_{t}^{i}(h)}^{i}-\bar{Q}_{\bar{\mu}^{i}}^{i}||=0$
$\forall i\in I$ a.s. -$\boldsymbol{P}^{\mu_{0},\bar{\phi}}$ over
$\mathcal{H}$. $\square$

\section{\label{sec:Non-myopic}Non-myopic players}

In the main text, we proved the results for the case where players
are myopic. Here, we assume that players maximize discounted expected
payoffs, where $\delta^{i}\in[0,1)$ is the discount factor of player
$i$. In particular, players can be forward looking and decide to
experiment. Players believe, however, that they face a stationary
environment and, therefore, have no incentives to influence the future
behavior of other players. We assume for simplicity that players know
the distribution of their own payoff perturbations.

Because players believe that they face a stationary environment, they
solve a (subjective) dynamic optimization problem that can be cast
recursively as follows. By the Principle of Optimality, $V^{i}(\mu^{i},s^{i})$
denotes the maximum expected discounted payoffs (i.e., the value function)
of player $i$ who starts a period by observing signal $s^{i}$ and
by holding belief $\mu^{i}$ if and only if 
\begin{equation}
V^{i}(\mu^{i},s^{i})=\int_{\Xi^{i}}\left\{ \max_{x^{i}\in\mathbb{X}^{i}}E_{\bar{Q}_{\mu^{i}}(\cdot|s^{i},x^{i})}\left[\pi^{i}(x^{i},Y^{i})+\xi^{i}(x^{i})+\delta E_{p_{S^{i}}}\left[V^{i}(\hat{\mu}^{i},S^{i})\right]\right]\right\} P_{\xi}(d\xi^{i}),\label{eq:Bellman}
\end{equation}
where $\hat{\mu}^{i}=B^{i}(\mu^{i},s^{i},x^{i},Y^{i})$ is the updated
belief. For all $(\mu^{i},s^{i},\xi^{i})$, let 
\[
\Phi^{i}(\mu^{i},s^{i},\xi^{i})=\arg\max_{x^{i}\in\mathbb{X}^{i}}E_{\bar{Q}_{\mu^{i}}(\cdot|s^{i},x^{i})}\left[\pi^{i}(x^{i},Y^{i})+\xi^{i}(x^{i})+\delta E_{p_{S^{i}}}\left[V^{i}(\hat{\mu}^{i},S^{i})\right]\right].
\]
The proof of the next lemma relies on standard arguments and is, therefore,
omitted.\footnote{\citet{doraszelski2010theory} study a similarly perturbed version
of the Bellman equation.}

\bigskip{}

\begin{lem}
\label{lemma:bellman_beliefs+pert}There exists a unique solution
$V^{i}$ to the Bellman equation (\ref{eq:Bellman}); this solution
is bounded in $\Delta(\Theta^{i})\times\mathbb{S}^{i}$ and continuous
as a function of $\mu^{i}$. Moreover, $\Phi^{i}$ is single-valued
and continuous with respect to $\mu^{i}$, a.s.- $P_{\xi}$.
\end{lem}
\bigskip{}

Because players believe they face a stationary environment with i.i.d.
perturbations, it is without loss of generality to restrict behavior
to depend on the state of the recursive problem. Optimality of a policy
is defined as usual (with the requirement that \emph{$\phi_{t}^{i}\in\Phi^{i}$}
$\forall t$).

Lemma \ref{lem:Berk} implies that the \emph{support} of posteriors
converges, but posteriors need not converge. We can always find, however,
a subsequence of posteriors that converges. By continuity of dynamic
behavior in beliefs, the stable strategy profile is dynamically optimal
(in the sense of solving the dynamic optimization problem) given this
convergent posterior. For weakly identified games, the convergent
posterior is a fixed point of the Bayesian operator. Thus, the players'
limiting strategies will provide no new information. Since the value
of experimentation is nonnegative, it follows that the stable strategy
profile must also be myopically optimal (in the sense of solving the
optimization problem that ignores the future), which is the definition
of optimality used in the definition of Berk-Nash equilibrium. Thus,
we obtain the following characterization of the set of stable strategy
profiles when players follow optimal policies.

\bigskip{}

\begin{thm}
\label{theo:Stability_implies_equilibrium-1}Suppose that a strategy
profile $\sigma$ is stable under an optimal policy profile for a
perturbed and weakly identified game. Then $\sigma$ is a Berk-Nash
equilibrium of the game.
\end{thm}
\begin{proof}
The first part of the proof is identical to the proof of Theorem \ref{theo:Stability_implies_equilibrium}.
Here, we prove that, given that $\lim_{j\rightarrow\infty}\sigma_{t(j)}=\sigma$
and $\lim_{j\rightarrow\infty}\mu_{t(j)}^{i}=\mu_{\infty}^{i}\in\Delta(\Theta^{i}(\sigma))$
$\forall i$, then, $\forall i$, $\sigma^{i}$ is optimal\emph{ }for
the perturbed game given\emph{ $\mu_{\infty}^{i}\in\Delta(\Theta^{i})$,
i.e.,} $\forall(s^{i},x^{i})$, 
\begin{equation}
\sigma^{i}(x^{i}\mid s^{i})=P_{\xi}\left(\xi^{i}:\psi^{i}(\mu_{\infty}^{i},s^{i},\xi^{i})=\{x^{i}\}\right),\label{eq:pf(i)optimal-1}
\end{equation}
where $\psi^{i}(\mu_{\infty}^{i},s^{i},\xi^{i})\equiv\arg\max_{x^{i}\in\mathbb{X}^{i}}E_{\bar{Q}_{\mu_{\infty}^{i}}^{i}(\cdot\mid s^{i},x^{i})}\left[\pi^{i}(x^{i},Y^{i})\right]+\xi^{i}(x^{i})$.

To establish (\ref{eq:pf(i)optimal-1}), fix $i\in I$ and $s^{i}\in\mathbb{S}^{i}$.
Then 
\begin{align*}
\lim_{j\rightarrow\infty}\sigma_{t(j)}^{i}(h)(x^{i}|s^{i}) & =\lim_{j\rightarrow\infty}P_{\xi}\left(\xi^{i}:\phi_{t(j)}^{i}(\mu_{t(j)}^{i},s^{i},\xi^{i})=x^{i}\right)\\
 & =P_{\xi}\left(\xi^{i}:\Phi^{i}(\mu_{\infty}^{i},s^{i},\xi^{i})=\{x^{i}\}\right),
\end{align*}
where the second line follows by optimality of $\phi^{i}$ and Lemma
\ref{lemma:bellman_beliefs+pert}. This implies that $\sigma^{i}(x^{i}|s^{i})=P_{\xi}\left(\xi^{i}:\Phi^{i}(\mu_{\infty}^{i},s^{i},\xi^{i})=\{x^{i}\}\right)$.
Thus, it remains to show that 
\begin{equation}
P_{\xi}\left(\xi^{i}:\Phi^{i}(\mu_{\infty}^{i},s^{i},\xi^{i})=\{x^{i}\}\right)=P_{\xi}\left(\xi^{i}:\psi^{i}(\mu_{\infty}^{i},s^{i},\xi^{i})=\{x^{i}\}\right)\label{eq:myopic=00003Ddynamic-1}
\end{equation}
$\forall x^{i}$ such that $P_{\xi}\left(\xi^{i}:\Phi^{i}(\mu_{\infty}^{i},s^{i},\xi^{i})=\{x^{i}\}\right)>0$.
From now on, fix any such $x^{i}$. Since $\sigma^{i}(x^{i}\mid s^{i})>0$,
the assumption that the game is weakly identified implies that $Q_{\theta_{1}^{i}}^{i}(\cdot\mid x^{i},s^{i})=Q_{\theta_{2}^{i}}^{i}(\cdot\mid x^{i},s^{i})$
$\forall\theta_{1}^{i},\theta_{2}^{i}\in\Theta(\sigma)$. The fact
that $\mu_{\infty}^{i}\in\Delta(\Theta^{i}(\sigma))$ then implies
that 
\begin{equation}
B^{i}(\mu_{\infty}^{i},s^{i},x^{i},y^{i})=\mu_{\infty}^{i}\label{eq:no_updating-1}
\end{equation}
 $\forall y^{i}\in\mathbb{Y}^{i}$. Thus, $\Phi^{i}(\mu_{\infty}^{i},s^{i},\xi^{i})=\{x^{i}\}$
is equivalent to
\begin{align*}
E_{\bar{Q}_{\mu_{\infty}^{i}}(\cdot|s^{i},x^{i})} & \Bigl[\pi^{i}(x^{i},Y^{i})+\xi^{i}(x^{i})+\delta E_{p_{S^{i}}}\left[V^{i}(\mu_{\infty}^{i},S^{i})\right]\Bigr]\\
>\,\, & E_{\bar{Q}_{\mu_{\infty}^{i}}(\cdot|s^{i},\tilde{x}^{i})}\left[\pi^{i}(\tilde{x}^{i},Y^{i})+\xi^{i}(\tilde{x}^{i})+\delta E_{p_{S^{i}}}\left[V^{i}(B^{i}(\mu_{\infty}^{i},s^{i},\tilde{x}^{i},Y^{i}),S^{i})\right]\right]\\
\geq\,\, & E_{\bar{Q}_{\mu_{\infty}^{i}}(\cdot|s^{i},\tilde{x}^{i})}\left[\pi^{i}(\tilde{x}^{i},Y^{i})+\xi^{i}(\tilde{x}^{i})\right]+\delta E_{p_{S^{i}}}\left[V^{i}(E_{\bar{Q}_{\mu_{\infty}^{i}}(\cdot|s^{i},\tilde{x}^{i})}\left[B^{i}(\mu_{\infty}^{i},s^{i},\tilde{x}^{i},Y^{i})\right],S^{i})\right]\\
=\,\, & E_{\bar{Q}_{\mu_{\infty}^{i}}(\cdot|s^{i},\tilde{x}^{i})}\left[\pi^{i}(\tilde{x}^{i},Y^{i})+\xi^{i}(\tilde{x}^{i})\right]+\delta E_{p_{S^{i}}}\left[V^{i}(\mu_{\infty}^{i},S^{i})\right]
\end{align*}
$\forall\tilde{x}^{i}\in\mathbb{X}^{i}$, where the first line follows
by equation (\ref{eq:no_updating-1}) and definition of $\Phi^{i}$,
the second line follows by the convexity\footnote{See, for example, \citet{nyarko1994convexity}, for a proof of convexity
of the value function.} of $V^{i}$ as a function of $\mu^{i}$ and Jensen's inequality,
and the last line by the fact that Bayesian beliefs have the martingale
property. In turn, the above expression is equivalent to $\psi(\mu_{\infty}^{i},s^{i},\xi^{i})=\{x^{i}\}$.
\end{proof}

\section{\label{sec:Population-models}Population models}

We discuss some variants of population models that differ in the matching
technology and feedback. The right variant of population model will
depend on the specific application.\footnote{In some cases, it may be unrealistic to assume that players are able
to observe the private signals of previous generations, so some of
these models might be better suited to cases with public, but not
private, information.}

\textbf{\emph{$\textsc{Single pair model}$}}. Each period a single
pair of players is randomly selected from each of the $i$ populations
to play the game. At the end of the period, the signals, actions,
and outcomes of their own population are revealed to everyone.\footnote{Alternatively, we can think of different incarnations of players born
every period who are able to observe the history of previous generations.} Steady-state behavior in this case corresponds exactly to the notion
of Berk-Nash equilibrium described in the paper.

$\textsc{Random matching model}.$ Each period, all players are randomly
matched and observe only feedback from their own match. We now modify
the definition of Berk-Nash equilibrium to account for this random-matching
setting. The idea is similar to Fudenberg and Levine's (1993) definition
of a heterogeneous self-confirming equilibrium. Now each agent in
population $i$ can have different experiences and, hence, have different
beliefs and play different strategies in steady state.

For all $i\in I$, define 
\[
BR^{i}(\sigma^{-i})=\left\{ \sigma^{i}:\sigma^{i}\,\,\mbox{is optimal given}\,\,\mu^{i}\in\Delta\left(\Theta^{i}(\sigma^{i},\sigma^{-i})\right)\right\} .
\]
Note that $\sigma$ is a Berk-Nash equilibrium if and only if $\sigma^{i}\in BR^{i}(\sigma^{-i})$
$\forall i\in I$.\bigskip{}

\begin{defn}
A strategy profile $\sigma$ is a \textbf{heterogeneous Berk-Nash
equilibrium} of game\emph{ $\mathcal{G}$} if, for all $i\in I$,
$\sigma^{i}$ is in the convex hull of $BR^{i}(\sigma^{-i})$.
\end{defn}
\bigskip{}

Intuitively, a heterogenous equilibrium strategy $\sigma^{i}$ is
the result of convex combinations of strategies that belong to $BR^{i}(\sigma^{-i})$;
the idea is that each of these strategies is followed by a segment
of the population $i$.\footnote{Unlike the case of heterogeneous self-confirming equilibrium, a definition
where each action in the support of $\sigma$ is supported by a (possibly
different) belief would not be appropriate here. The reason is that
$BR^{i}(\sigma^{-i})$ might contain only mixed, but not pure strategies
(e.g., Example 1). } 

$\textsc{Random-matching model with population feedback}.$ Each period
all players are randomly matched; at the end of the period, each player
in population $i$ observes the signals, actions, and outcomes of
their own population. Define

\[
\bar{BR}^{i}(\sigma^{i},\sigma^{-i})=\left\{ \hat{\sigma}^{i}:\hat{\sigma}^{i}\,\,\mbox{is optimal given}\,\,\mu^{i}\in\Delta\left(\Theta^{i}(\sigma^{i},\sigma^{-i})\right)\right\} .
\]

\bigskip{}

\begin{defn}
A strategy profile $\sigma$ is a \textbf{heterogeneous Berk-Nash
equilibrium with population feedback} of game\emph{ $\mathcal{G}$}
if, for all $i\in I$, $\sigma^{i}$ is in the convex hull of $\bar{BR}^{i}(\sigma^{i},\sigma^{-i})$.
\end{defn}
\bigskip{}

The main difference when players receive population feedback is that
their beliefs no longer depend on their own strategies but rather
on the aggregate population strategies.

\subsection{Equilibrium foundation}

Using arguments similar to the ones in the text, it is now straightforward
to conclude that the definition of heterogenous Berk-Nash equilibrium
captures the steady state of a learning environment with a population
of agents in the role of each player. To see the idea, let each population
$i$ be composed of a continuum of agents in the unit interval $K\equiv[0,1]$.
A strategy of agent $ik$ (meaning agent $k\in K$ from population
$i$) is denoted by $\sigma^{ik}$. The aggregate strategy of population
(i.e., player) $i$ is $\sigma^{i}=\int_{K}\sigma^{ik}dk$.

$\textsc{Random matching model}$. Suppose that each agent is optimizing
and that, for all $i$, $(\sigma_{t}^{ik})$ converges to $\sigma^{ik}$
a.s. in $K$, so that individual behavior stabilizes.\footnote{We need individual behavior to stabilize; it is not enough that it
stabilizes in the aggregate. This is natural, for example, if we believe
that agents whose behavior is unstable will eventually realize they
have a misspecified model.} Then Lemma \ref{lem:Berk} says that the support of beliefs must
eventually be $\Theta^{i}(\sigma^{ik},\sigma^{-i})$ for agent $ik$.
Next, for each $ik$, take a convergent subsequence of beliefs $\mu_{t}^{ik}$
and denote it $\mu_{\infty}^{ik}$. It follows that $\mu_{\infty}^{ik}\in\Delta(\Theta^{i}(\sigma^{ik},\sigma^{-i}))$
and, by continuity of behavior in beliefs, $\sigma^{ik}$ is optimal
given $\mu_{\infty}^{ik}$. In particular, $\sigma^{ik}\in BR^{i}(\sigma^{-i})$
for all $ik$ and, since $\sigma^{i}=\int_{K}\sigma^{ik}dk$, it follows
that $\sigma^{i}$ is in the convex hull of $BR^{i}(\sigma^{-i})$.

$\textsc{Random matching model with population feedback}$. Suppose
that each agent is optimizing and that, for all $i$, $\sigma_{t}^{i}=\int_{K}\sigma_{t}^{ik}dk$
converges to $\sigma^{i}$. Then Lemma \ref{lem:Berk} says that the
support of beliefs must eventually be $\Theta^{i}(\sigma^{i},\sigma^{-i})$
for any agent in population $i$. Next, for each $ik$, take a convergent
subsequence of beliefs $\mu_{t}^{ik}$ and denote it $\mu_{\infty}^{ik}$.
It follows that $\mu_{\infty}^{ik}\in\Delta(\Theta^{i}(\sigma^{i},\sigma^{-i}))$
and, by continuity of behavior in beliefs, $\sigma^{ik}$ is optimal
given $\mu_{\infty}^{ik}$. In particular, $\sigma^{ik}\in\bar{BR^{i}}(\sigma^{-i})$
for all $i,k$ and, since $\sigma^{i}=\int_{K}\sigma^{ik}dk$, it
follows that $\sigma^{i}$ is in the convex hull of $\bar{BR^{i}}(\sigma^{-i})$.

\section{\label{sec:Lack-of-payoff}Lack of payoff feedback}

In the paper, players are assumed to observe their own payoffs. We
now provide two alternatives to relax this assumption. In the first
alternative, players observe no feedback about payoffs; in the second
alternative, players may observe partial feedback.

\emph{No payoff feedback}. In the paper we had a single, deterministic
payoff function $\pi^{i}:\mathbb{X}^{i}\times\mathbb{Y}^{i}\rightarrow\mathbb{R}$,
which can be represented in vector form as an element $\pi^{i}\in\mathbb{R}^{\#(\mathbb{X}^{i}\times\mathbb{Y}^{i})}$.
We now generalize it to allow for uncertain payoffs. Player $i$ is
endowed with a probability distribution $P_{\pi^{i}}\in\Delta(\mathbb{R}^{\#(\mathbb{X}^{i}\times\mathbb{Y}^{i})})$
over the possible payoff functions. In particular, the random variable
$\pi^{i}$ is independent of $Y^{i}$, and so there is nothing new
to learn about payoffs from observing consequences. With random payoff
functions, the results extend provided that optimality is defined
as follows: A strategy $\sigma^{i}$ for player $i$ is \textbf{optimal}
given\emph{ $\mu^{i}\in\Delta(\Theta^{i})$} if $\sigma^{i}(x^{i}\mid s^{i})>0$
implies that 
\[
x^{i}\in\arg\max_{\bar{x}^{i}\in\mathbb{X}^{i}}E_{P_{\pi^{i}}}E_{\bar{Q}_{\mu^{i}}^{i}(\cdot\mid s^{i},\bar{x}^{i})}\left[\pi^{i}(\bar{x}^{i},Y^{i})\right].
\]
Note that by interchanging the order of integration, this notion of
optimality is equivalent to the notion in the paper where the deterministic
payoff function is given by $E_{P_{\pi^{i}}}\pi^{i}(\cdot,\cdot)$.

\emph{Partial payoff feedback}. Suppose that player $i$ knows her
own consequence function $f^{i}:\mathbb{X}\times\Omega\rightarrow\mathbb{Y}^{i}$
and that her payoff function is now given by $\pi^{i}:\mathbb{X}\times\Omega\rightarrow\mathbb{R}$.
In particular, player $i$ may not observe her own payoff, but observing
a consequence may provide partial information about $(x^{-i},\omega)$
and, therefore, about payoffs. Unlike the case in the text where payoffs
are observed, a belief \emph{$\mu^{i}\in\Delta(\Theta^{i})$ }may
not uniquely determine expected payoffs. The reason is that the distribution
over consequences implied by $\mu^{i}$ may be consistent with several
distributions over $\mathbb{X}^{-i}\times\Omega$; i.e., the distribution
over $\mathbb{X}^{-i}\times\Omega$ is only partially identified.
Define the set ${\cal M}_{\mu^{i}}\subseteq\Delta(\mathbb{X}^{-i}\times\Omega)^{\mathbb{S}^{i}\times\mathbb{X}^{i}}$
to be the set of conditional distributions over $\mathbb{X}^{-i}\times\Omega$
given $(s^{i},x^{i})\in\mathbb{S}^{i}\times\mathbb{X}^{i}$ that are
consistent with belief $\mu^{i}\in\Delta(\Theta^{i})$, i.e., $m\in{\cal M}_{\mu^{i}}$
if and only if $\bar{Q}_{\mu^{i}}^{i}(y^{i}\mid s^{i},x^{i})=m\left(f^{i}(x^{i},X^{-i},W)=y^{i}\mid s^{i},x^{i}\right)$
for all $(s^{i},x^{i})\in\mathbb{S}^{i}\times\mathbb{X}^{i}$ and
$y^{i}\in\mathbb{Y}^{i}$. Then optimality should be defined as follows:
A strategy $\sigma^{i}$ for player $i$ is \textbf{optimal} given\emph{
$\mu^{i}\in\Delta(\Theta^{i})$} if there exists $m_{\mu^{i}}\in{\cal M}_{\mu^{i}}$
such that $\sigma^{i}(x^{i}\mid s^{i})>0$ implies that 
\[
x^{i}\in\arg\max_{\bar{x}^{i}\in\mathbb{X}^{i}}E_{m_{\mu^{i}}(\cdot\mid s^{i},\bar{x}^{i})}\left[\pi^{i}(\bar{x}^{i},X^{-i},W)\right].
\]

Finally, the definition of identification would also need to be changed
to require not only that there is a unique distribution over consequences
that matches the observed data, but also that this unique distribution
implies a unique expected utility function.

\section{\label{sec:Global-stability}Global stability: Example 2.1 (monopoly
with unknown demand).}

Theorem \ref{Theo:converse-1} says that all Berk-Nash equilibria
can be approached with probability 1 provided we allow for vanishing
optimization mistakes. In this appendix, we illustrate how to use
the techniques of stochastic approximation theory to establish stability
of equilibria under the assumption that players make no optimization
mistakes. We present the explicit learning dynamics for the monopolist
with unknown demand, Example 2.1, and show that the unique equilibrium
in this example is globally stable. The intuition behind global stability
is that switching from the equilibrium strategy to a strategy that
puts more weight on a price of 2 changes beliefs in a way that makes
the monopoly want to put less weight on a price of 2, and similarly
for a deviation to a price of 10.

We first construct a perturbed version of the game. Then we show that
the learning problem is characterized by a nonlinear stochastic system
of difference equations and employ stochastic approximation methods
for studying the asymptotic behavior of such system. Finally, we take
the payoff perturbations to zero.

In order to simplify the exposition and thus better illustrate the
mechanism driving the dynamics, we modify the subjective model slightly.
We assume the monopolist only learns about the parameter $b\in\mathbb{R}$;
i.e., her beliefs about parameter $a$ are degenerate at a point $a=40\ne a^{0}$
and thus are never updated. Therefore, beliefs $\mu$ are probability
distributions over $\mathbb{R}$, i.e., $\mu\in\Delta(\mathbb{R})$. 

$\textsc{Perturbed Game}.$ Let $\xi$ be a real-valued random variable
distributed according to $P_{\xi}$; we use $F$ to denote the associated
cdf and $f$ the pdf. The perturbed payoffs are given by $yx-\xi1\{x=10\}$.
Thus, given beliefs $\mu\in\Delta(\mathbb{R})$, the probability of
optimally playing $x=10$ is 
\[
\sigma(\mu)=F\left(8a-96E_{\mu}[B]\right).
\]
Note that the only aspect of $\mu$ that matters for the decision
of the monopolist is $E_{\mu}[B]$. Thus, letting $m=E_{\mu}[B]$
and slightly abusing notation, we use $\sigma(\mu)=\sigma(m)$ as
the optimal strategy.

$\textsc{Bayesian Updating}.$ We now derive the Bayesian updating
procedure. We assume that the the prior $\mu_{0}$ is given by a Gaussian
distribution with mean and variance $m_{0},\tau_{0}^{2}$.\footnote{This choice of prior is standard in Gaussian settings like ours. As
shown below this choice simplifies the exposition considerably. } It is possible to show that, given a realization $(y,x)$ and a prior
$N(m,\tau^{2})$, the posterior is also Gaussian and the mean and
variance evolve as follows: $m_{t+1}=m_{t}+\left(\frac{-(Y_{t+1}-a)}{X_{t+1}}-m_{t}\right)\left(\frac{X_{t+1}^{2}}{X_{t+1}^{2}+\tau_{t}^{-2}}\right)$
and $\tau_{t+1}^{2}=\frac{1}{(X_{t+1}^{2}+\tau_{t}^{-2})}$. 

$\textsc{Nonlinear Stochastic Difference Equations and Stochastic Approximation.}$
For simplicity, let $r_{t+1}\equiv\frac{1}{t+1}\left(\tau_{t}^{-2}+X_{t+1}^{2}\right)$
and note that the previous nonlinear system of stochastic difference
equations can be written as
\begin{align*}
m_{t+1} & =m_{t}+\frac{1}{t+1}\frac{X_{t+1}^{2}}{r_{t+1}}\left(\frac{-(Y_{t+1}-a)}{X_{t+1}}-m_{t}\right)\\
r_{t+1} & =r_{t}+\frac{1}{t+1}\left(X_{t+1}^{2}-r_{t}\right).
\end{align*}
 Let $\beta_{t}=(m_{t},r_{t})'$, $Z_{t}=(X_{t},Y_{t})$, 
\[
G(\beta_{t},z_{t+1})=\left[\begin{array}{c}
\frac{x_{t+1}^{2}}{r_{t+1}}\left(\frac{-(y_{t+1}-a)}{x_{t+1}}-m_{t}\right)\\
\left(x_{t+1}^{2}-r_{t}\right)
\end{array}\right]
\]
and 
\begin{align*}
\mathbb{G}(\beta) & =\left[\begin{array}{c}
\mathbb{G}_{1}(\beta)\\
\mathbb{G}_{2}(\beta)
\end{array}\right]=E_{P_{\sigma}}\left[G(\beta,Z_{t+1})\right]\\
 & =\left[\begin{array}{c}
F(8a-96m)\frac{100}{r}\left(\frac{-(a_{0}-a-b_{0}10)}{10}-m\right)+\left(1-F(8a-96m)\right)\frac{4}{r}\left(\frac{-(a_{0}-a-b_{0}2)}{2}-m\right)\\
\left(4+F(8a-96m)96-r\right)
\end{array}\right]
\end{align*}
where $P_{\sigma}$ is the probability over $Z$ induced by $\sigma$
(and $y=a^{0}-b^{0}x+\omega$). Therefore, the dynamical system can
be cast as 
\[
\beta_{t+1}=\beta_{t}+\frac{1}{t+1}\mathbb{G}(\beta_{t})+\frac{1}{t+1}V_{t+1}
\]
 with $V_{t+1}=G(\beta_{t},Z_{t+1})-\mathbb{G}(\beta_{t}).$ Stochastic
approximation theory (e.g., \citet{KushnerYin2003}) implies, roughly
speaking, that in order to study the asymptotic behavior of $(\beta_{t})_{t}$
it is enough to study the behavior of the orbits of the following
ODE 
\[
\dot{\beta}(t)=\mathbb{G}(\beta(t)).
\]

$\textsc{Characterization of the Steady States.}$ In order to find
the steady states of $(\beta_{t})_{t}$, it is enough to find $\beta^{\ast}$
such that $\mathbb{G}(\beta^{\ast})=0$. Let $H_{1}(m)\equiv F(8a-96m)10$

$\left(-(a_{0}-a)+\left(b_{0}-m\right)10\right)+\left(1-F(8a-96m)\right)2\left(-(a_{0}-a)+\left(b_{0}-m\right)2\right)$.
Observe that $\mathbb{G}_{1}(\beta)=r^{-1}H_{1}(m)$ and that $H_{1}$
is continuous and $\lim_{m\rightarrow-\infty}H_{1}(m)=\infty$ and
$\lim_{m\rightarrow\infty}H_{1}(m)=-\infty$. Thus, there exists at
least one solution $H_{1}(m)=0$. Therefore, there exists at least
one $\beta^{\ast}$ such that $\mathbb{G}(\beta^{\ast})=0$. 

Let $\bar{b}=b_{0}-\frac{a_{0}-a}{10}=4-\frac{1}{5}=\frac{19}{5}$
and $\underline{b}=b_{0}-\frac{a_{0}-a}{2}=4-\frac{42-40}{2}=3$,
$\bar{r}=4+F(8a-96\underline{b})96$ and $\underline{r}=4+F(8a-96\bar{b})96$,
and $\mathbb{B}\equiv[\underline{b},\bar{b}]\times[\underline{r},\bar{r}]$.
It follows that $H_{1}(m)<0$ $\forall m>\bar{b}$, and thus $m^{\ast}$
must be such that $m^{\ast}\leq\bar{b}$. It is also easy to see that
$m^{\ast}\geq\underline{b}$. Moreover, $\frac{dH_{1}(m)}{dm}=96f(8a-96m)\bigl(8(a_{0}-a)-\bigl(b_{0}-m\bigr)96\bigr)-4-96F(8a-96m).$
Thus, for any $m\leq\bar{b}$ , $\frac{dH_{1}(m)}{dm}<0$, because
$m\leq\bar{b}$ implies $8(a_{0}-a)\leq(b_{0}-m)80<(b_{0}-m)96$. 

Therefore, on the relevant domain $m\in[\underline{b},\bar{b}]$,
$H_{1}$ is decreasing, thus implying that there exists only one $m^{\ast}$
such that $H_{1}(m^{\ast})=0$. Therefore, there exists only one $\beta^{\ast}$
such that $\mathbb{G}(\beta^{\ast})=0$ .

We are now interested in characterizing the limit of $\beta^{\ast}$
as the perturbation vanishes, i.e. as $F$ converges to $1\{\xi\geq0\}$.
To do this we introduce some notation. We consider a sequence $(F_{n})_{n}$
that converges to $1\{\xi\geq0\}$ and use $\beta_{n}^{\ast}$ to
denote the steady state associated to $F_{n}$. Finally, we use $H_{1}^{n}$
to denote the $H_{1}$ associated to $F_{n}$. 

We proceed as follows. First note that since $\beta_{n}^{\ast}\in\mathbb{B}$
$\forall n$, the limit exists (going to a subsequence if needed).
We show that $m^{*}\equiv\lim_{n\rightarrow\infty}m_{n}^{\ast}=\frac{8a}{96}=8\frac{40}{96}=\frac{10}{3}$.
Suppose not, in particular, suppose that $\lim_{n\rightarrow\infty}m_{n}^{\ast}<\frac{8a}{96}=\frac{10}{3}$
(the argument for the reverse inequality is analogous and thus omitted).
In this case $\lim_{n\rightarrow\infty}8a-96m_{n}^{\ast}>0$, and
thus $\lim_{n\rightarrow\infty}F_{n}(8a-96m_{n}^{\ast})=1$. Therefore
\[
\lim_{n\rightarrow\infty}H_{1}^{n}(\beta_{n}^{\ast})=10\left(-(a_{0}-a)+\left(b_{0}-m^{*}\right)10\right)\geq10\left(-2+\frac{2}{3}10\right)>0.
\]
But this implies that $\exists N$ such that $H_{1}^{n}(\beta_{n}^{\ast})>0$
$\forall n\geq N$ which is a contradiction. 

Moreover, define $\sigma_{n}^{*}=F_{n}(8a-96m_{n}^{\ast})$ and $\sigma^{*}=\lim_{n\rightarrow\infty}\sigma_{n}$.
Since $H_{1}^{n}(m_{n}^{*})=0$ $\forall n$ and $m^{*}=\frac{10}{3}$
, it follows that 
\[
\sigma^{\ast}=\frac{-2\left(-2+\left(4-\frac{10}{3}\right)2\right)}{10\left(-2+\left(4-\frac{10}{3}\right)10\right)-2\left(-2+\left(4-\frac{10}{3}\right)2\right)}=\frac{1}{36}.
\]

$\textsc{Global convergence to the Steady State}.$ In our example,
it is in fact possible to establish that behavior converges with probability
1 to the unique equilibrium. By the results in \citet{Benaim99} Section
6.3, it is sufficient to establish the $\textit{global}$ asymptotic
stability of $\beta_{n}^{\ast}$ for any $n$, i.e., the basin of
attraction of $\beta_{n}^{\ast}$ is all of $\mathbb{B}$.

In order to do this let $L(\beta)=\left(\beta-\beta_{n}^{\ast}\right)'P\left(\beta-\beta_{n}^{\ast}\right)$
for all $\beta$; where $P\in\mathbb{R}^{2\times2}$ is positive definite
and $\emph{diagonal}$ and will be determined later. Note that $L(\beta)=0$
iff $\beta=\beta_{n}^{\ast}$ . Also
\begin{align*}
\frac{dL(\beta(t))}{dt} & =\nabla L(\beta(t))'\mathbb{G}(\beta(t))\\
 & =2\left(\beta(t)-\beta_{n}^{\ast}\right)'P\left(\mathbb{G}(\beta(t))\right)\\
 & =2\left\{ \left(m(t)-m_{n}^{\ast}\right)P_{[11]}\mathbb{G}_{1}(\beta(t))+\left(r(t)-r_{n}^{\ast}\right)P_{[22]}\mathbb{G}_{2}(\beta(t))\right\} .
\end{align*}
Since $\mathbb{G}(\beta_{n}^{\ast})=0$, 
\begin{align*}
\frac{dL(\beta(t))}{dt}= & 2\left(\beta(t)-\beta_{n}^{\ast}\right)'P\left(\mathbb{G}(\beta(t))-\mathbb{G}\left(\beta_{n}^{\ast}\right)\right)\\
= & 2\left(m(t)-m_{n}^{\ast}\right)P_{[11]}\left(\mathbb{G}_{1}(\beta(t))-\mathbb{G}_{1}\left(\beta_{n}^{\ast}\right)\right)\\
 & +2\left(r(t)-r_{n}^{\ast}\right)P_{[22]}\left(\mathbb{G}_{2}(\beta(t))-\mathbb{G}_{2}\left(\beta_{n}^{\ast}\right)\right)\\
= & 2\left(m(t)-m_{n}^{\ast}\right)^{2}P_{[11]}\int_{0}^{1}\frac{\partial\mathbb{G}_{1}(m_{n}^{\ast}+s(m(t)-m_{n}^{\ast}),r_{n}^{*})}{\partial m}ds\\
 & +2\left(r(t)-r_{n}^{\ast}\right)^{2}P_{[22]}\int_{0}^{1}\frac{\partial\mathbb{G}_{2}(m_{n}^{*},r_{n}^{\ast}+s(r(t)-r_{n}^{\ast}))}{\partial r}ds
\end{align*}
where the last equality holds by the mean value theorem. Note that
$\frac{d\mathbb{G}_{2}(m_{n}^{*},r_{n}^{\ast}+s(r(t)-r_{n}^{\ast}))}{dr}=-1$
and $\int_{0}^{1}\frac{d\mathbb{G}_{1}(m_{n}^{\ast}+s(m(t)-m_{n}^{\ast}),r_{n}^{*})}{dm}ds=\int_{0}^{1}\left(r_{n}^{\ast}\right)^{-1}\frac{dH_{1}(m_{n}^{\ast}+s(m(t)-m_{n}^{\ast}))}{dm}ds$.
Since $r(t)>0$ and $r_{n}^{\ast}\geq0$ the first term is positive
and we already established that $\frac{dH_{1}(m)}{dm}<0$ $\forall m$
in the relevant domain. Thus, by choosing $P_{[11]}>0$ and $P_{[22]}>0$
it follows that $\frac{dL(\beta(t))}{dt}<0$.

Therefore, we show that $L$ satisfies the following properties: is
strictly positive $\forall\beta\ne\beta_{n}^{\ast}$ and $L(\beta_{n}^{\ast})=0$,
and $\frac{dL(\beta(t))}{dt}<0$. Thus, the function satisfies all
the conditions of a Lyapounov function and, therefore, $\beta_{n}^{\ast}$
is globally asymptotically stable $\forall n$ (see \citet{HirschSmaleDevaney04}
p. 194).
\end{document}